\documentclass[10pt]{amsart}
\usepackage{fullpage} 
\usepackage{amssymb} 
\usepackage{graphicx} 
\usepackage{amsmath,amsthm,amssymb,latexsym} 
\usepackage{enumerate} 
\usepackage{color} 
\usepackage{tikz}
\usetikzlibrary{matrix,graphs,arrows,positioning,calc,decorations.markings}
\usepackage{accents} 
\usepackage[pdftex,bookmarks,colorlinks,breaklinks]{hyperref}  
\definecolor{dullmagenta}{rgb}{0.4,0,0.4}   
\definecolor{darkblue}{rgb}{0,0,0.4}
\hypersetup{linkcolor=red,citecolor=blue,filecolor=dullmagenta,urlcolor=darkblue} 

\newtheorem{theorem}{Theorem}[section]

\theoremstyle{definition}

\newtheorem{assumption}[theorem]{Assumption}

\theoremstyle{remark}
\newtheorem{remark}[theorem]{Remark}

\numberwithin{equation}{section}

\begin{document}

\title[Reductions from Schlesinger to Painlev\'e]{Geometric Analysis of Reductions from Schlesinger 
	Transformations to Difference Painlev\'e Equations}

	\author{Anton Dzhamay}
	\address{School of Mathematical Sciences\\ 
	The University of Northern Colorado\\ 
	Campus Box 122\\ 
	501 20th Street\\ 
	Greeley, CO 80639, USA}
	\email{\href{mailto:adzham@unco.edu}{\texttt{adzham@unco.edu}}}

	\author{Tomoyuki Takenawa}
	\address{Faculty of Marine Technology\\
	Tokyo University of Marine Science and Technology\\
	2-1-6 Etchu-jima, Koto-ku\\ 
	Tokyo, 135-8533, Japan}
	\email{\href{mailto:takenawa@kaiyodai.ac.jp}{\texttt{takenawa@kaiyodai.ac.jp}}}

	\keywords{Integrable systems, Painlev\'e equations, difference equations, isomonodromic transformations, birational transformations}
	\subjclass[2010]{34M55, 34M56, 14E07}

\keywords{Integrable systems, Painlev\'e equations, difference equations, isomonodromic transformations, birational transformations}
\subjclass[2010]{34M55, 34M56, 14E07}

\date{}

\begin{abstract}
	We present two examples of reductions from the evolution equations describing 
	discrete Schlesinger transformations of Fuchsian systems to difference Painlev\'e 
	equations: difference Painlev\'e equation d-$P\left({A}_{2}^{(1)*}\right)$ with the symmetry
	group ${E}^{(1)}_{6}$ and difference Painlev\'e equation d-$P\left({A}_{1}^{(1)*}\right)$ 
	with the symmetry group ${E}^{(1)}_{7}$. In both cases we describe in detail how to 
	compute their Okamoto space of the initial conditions and emphasize the role played by 
	geometry in helping us to understand the structure of the reduction, a choice of a good coordinate
	system describing the equation, and how to compare it with other instances of equations
	of the same type.
\end{abstract}

\maketitle


\section{Introduction} 
\label{sec:introduction}
It is well-known that many special functions, e.g.~the Airy, Bessel,
or Legendre functions, originate as series solutions of \emph{linear}
ordinary differential equations (ODEs). Such functions play a crucial
role in describing a wide range of important physical and mathematical phenomena.
An essential point here is that for linear ordinary differential equations singularities of solutions
can only occur at the points where the coefficients of the equation itself become singular. This makes it 
possible to talk about global properties of solutions (and hence, of the corresponding special functions), 
and the asymptotic behavior of solutions near those fixed singular points.

For nonlinear equations the situation is very different. Although the Cauchy existence theorem guarantees
\emph{local} existence of the solution to a given initial value problem at an ordinary (regular) point, in general 
the domain in which this solution is defined depends not just on the equation itself, but on the initial 
conditions as well --- solutions acquire \emph{movable} (i.e., dependent on the initial values or, equivalently, on the 
constants of integration) singularities. Such singularities are called \emph{critical} if a solution 
looses its single-valued character in a neighborhood of the singularity (e.g., when a singular point is a \emph{branch point}).   
An ODE is said to satisfy the \emph{Painlev\'e property} if its general solution is free of \emph{movable} 
critical singular points. Otherwise, the Riemann surface uniformizing such solution becomes dependent on the 
constants of integration, which prevents global analysis. Thus, equivalently, the Painlev\'e property of an ODE
is the uniformizability of its general solution, see \cite{Con:1999:TPATNODE} (as well as the other excellent
articles in the volume \cite{Cont:1999:TPP}) for a careful overview of these ideas. 

It is clear that liner equations satisfy Painlev\'e property. The importance of the Painlev\'e property for nonlinear 
equations is that, similar to the linear case, solutions of these equations give rise to new transcendental
functions. In that sense, according to M.Kruskal as quoted in \cite{GraRam:2014:DPEAIP}, nonlinear equations satisfying this property 
are on the border between trivially integrable linear
equations and nonlinear equations that are not integrable, and so the Painlev\'e property is essentially
equivalent to (and is a criterion of) integrability.

The search for new transcendental functions was the original motivation in the work of P.Painlev\'e who,
together with his student B.Gambier, had classified all of the rational second-order differential equations
that have the Painlev\'e property, \cite{Pai:1902:SLEDDSOEDSDLGEU}, \cite{Pai:1973:OPPT}, \cite{Gam:1910:SLEDDSOEDPDDLGEAPCF}.
Among 50 classes of equations that they found, only \emph{six} can not be reduced to linear equations or
integrated by quadratures. These equations are now known as $P_\text{I}-P_\text{VI}$. Solutions to 
these equations, the so-called \emph{Painlev\'e transcendents}, are playing an increasingly important 
role in describing a wide range of nonlinear phenomena in mathematics and physics \cite{IwaKimShiYos:1991:FGP}. 

Almost simultaneously with the work of P.Painlev\'e and B.Gambier, the most general Painlev\'e VI equation 
was obtained by R.Fuchs  \cite{Fuc:1905:SQEDLDSO} in the theory of \emph{isomonodromic deformations} of 
Fuchsian systems. This theory, developed in the works of R.Fuchs \cite{Fuc:1907:ULHDZOMDIEGWSS}, 
L.Schlesinger \cite{Sch:1912:UEKDBOFKP}, R.Garnier \cite{Gar:1926:SDPDRPLSDLDSO}, and then extended to the 
non-Fuchsian case by M.Jimbo, T.Miwa, and K.Ueno \cite{JimMiwUen:1981:MPDLODEWRCGTF,JimMiw:1982:MPDLODEWRC}, 
and also by H.Flaschka and A.Newell \cite{FlaNew:1980:MASDI}, as well as the related
Riemann-Hilbert approach \cite{ItsNov:1986:IDMTPE}, \cite{FokItsKapNov:2006:PT}, 
are now among the most powerful methods for studying the structure of the Painlev\'e 
transcendents.

Over the last thirty years a significant effort has been put towards understanding and generalizing results and
methods of the classical theory of completely integrable systems to the \emph{discrete} case. This is true for the 
theory of Painlev\'e equations as well. Discrete Painlev\'e equations were originally defined as second order nonlinear 
difference equations that have usual Painlev\'e equations as continuous limits \cite{BreKaz:1990:ESFTCS}, 
\cite{GroMig:1990:ANTOTQG}. A systematic study of discrete Painlev\'e equations was started by 
B.Grammaticos, J.Hietarinta,  F.Nijhoff, V.Papageorgiou and A.Ramani, \cite{NijPap:1991:SROILADAOTPRIE},
\cite{RamGraHie:1991:DVOTPE}, \cite{GraRamPap:1991:IMHPP}, and many different examples of such equations were 
obtained in a series of papers by Grammaticos, Ramani, and their collaborators by a systematic application 
of the singularity confinement criterion, see reviews \cite{GraRam:2004:DPER}, \cite{GraRam:2014:DPEAIP},
and many references therein. Discrete Painlev\'e equations
also appear in a broad spectrum of important nonlinear problems in mathematics and physics, among which are the 
theory of orthogonal polynomials, quantum gravity, determinantal random point processes, 
reductions of integrable lattice equations, and, notably, as B\"acklund transformations of differential 
Painlev\'e equations. Some of these problems are discrete anlogues or direct discretizations of the corresponding 
nonlinear problems, and some describe purely discrete phenomena. 

It turned out that classifying discrete Painlev\'e equations by their continuous limits, as well as the singularity 
confinement criterion, is not a very good approach,
since such correspondence is far from being bijective. It is both possible for
the same discrete equation to have different continuous Painlev\'e equations as continuous limits 
under different limiting procedures, and for different discrete equations to have the same continuous limit. 
In the seminal paper \cite{Sak:2001:RSAWARSGPE} H.Sakai showed that an effective 
way to understand and classify discrete Painlev\'e equations is through algebraic geometry. In this approach,
to each equation, if we consider it as a two-dimensional first-order nonlinear system, we put in correspondence a family of
algebraic surfaces $\mathcal{X}_{\mathbf{b}}$, where $\mathbf{b} = \{b_{i}\}$ is some collection of
parameters that change, depending on the type of the equation, 
in an additive, multiplicative, or elliptic fashion as functions of a discrete ``time'' parameter. This family $\mathcal{X}_{\mathbf{b}}$,
by a slight abuse of terminology, is called the \emph{Okamoto space of initial conditions} (that we often denote 
simply by $\mathcal{X}$ omitting explicit dependence on parameters $b_{i}$) and it is obtained by resolving the indeterminacy
points of the corresponding map $\varphi: \mathbb{P}^{2} \to \mathbb{P}^{2}$ via the blowing-up procedure. By Sakai's
theory, the complete resolution of indeterminacies is obtained by blowing up 9 (possibly infinitely close) points 
on $\mathbb{P}^{2}$ (or eight points on a birationally equivalent $\mathbb{P}^{1} \times \mathbb{P}^{1}$ compactification 
of $\mathbb{C}^{2}$). The resulting surface $\mathcal{X}$ has a special property that it admits a unique anti-canonical 
divisor $D\in |-K_{\mathcal{X}}|$ of \emph{canonical type}. The orthogonal complement of $-K_{\mathcal{X}}$ in the Picard lattice
$\operatorname{Pic}(\mathcal{X}) \simeq H^{2}(\mathcal{X};\mathbb{Z})$ is described by the affine Dynkin diagram ${E}^{(1)}_{8}$ that
has two intersecting root subsystems of affine type: $R$, that is generated by classes $D_{i}$ of irreducible components of the anticanonical
divisor $D$, and its orthogonal complement $R^{\perp}$ whose corresponding root lattice $Q(R^{\perp})$ is called the \emph{symmetry sub-lattice}.
Then the \emph{type} of the discrete Painlev\'e equation is the same as the \emph{type of its surface} $\mathcal{X}_{\mathbf{b}}$, 
which is just the type of an affine Dynkin diagram describing the root system $R$ of irreducible components $D_{i}$ of 
$D$ (essentially, the degeneration structure of the positions of the blowup points). 
Moreover, nonlinear Painlev\'e dynamic now becomes a translation in the symmetry sub-lattice $Q(R^{\perp})$, and so  
it becomes ``\emph{linearized}'' there, see \cite{Sak:2001:RSAWARSGPE} for details. This is somewhat similar to the algebro-geometric integrability of 
classical integrable systems and soliton equations, where nonlinear dynamic is mapped to commuting linear flows on the Jacobian
of the spectral curve of the associated linear problem via the Abel-Jacobi map.

One important observation from Sakai's geometric approach is that in addition to \emph{additive} (difference) and \emph{multiplicative} ($q$-difference)
\emph{discretizations} of continuous Painlev\'e equations, there are some purely discrete Painlev\'e equations. It also led to the 
discovery of the master \emph{elliptic} discrete Painlev\'e equation such that all of the other Painlev\'e equations can be  obtained from it
through degenerations (which corresponds to more and more special configurations of the blow-up points). This degenerations can be described
by the following scheme, where letters stand for the symmetry type of the equation (which is the type of the affine Dynkin diagram of the root 
subsystem $R^{\perp}$), and the subscripts $e$, $q$, $\delta$, and $c$ stand for elliptic, multiplicative, 
additive and differential Painlev\'e equations respectively, see Figure~\ref{fig:pain-diag}.

\begin{figure}[h]
	{\small
	\begin{equation*}
		\begin{tikzpicture}[>=stealth]
		  \matrix (m) [matrix of math nodes,column sep={1pt},row sep={15pt}]
		  {	
			\left(E^{(1)}_{8}\right)^{e} &&&&&&& \left(A_{1}^{(1)}\right)^{q}  \\
			\left(E^{(1)}_{8}\right)^{q} & \left(E^{(1)}_{7}\right)^{q} & \left(E^{(1)}_{6}\right)^{q} &  
			\node{\left(D^{(1)}_{5}\right)^{q}  };& 
			\left(A^{(1)}_{4}\right)^{q} & \left((A_{1}\!\!+\!\! A_{2})^{(1)}\right)^{q} &
			\left((A_{1}\!\!+\!\!A_{1})^{(1)}\right)^{q} & \left(A_{1}^{(1)}\right)^{q}  & \left(A_{0}^{(1)}\right)^{q} \\
			\left(E^{(1)}_{8}\right)^{\delta} & \left(E^{(1)}_{7}\right)^{\delta} & \left(E^{(1)}_{6}\right)^{\delta} &  
			\left(D^{(1)}_{4}\right)^{c,\delta} & \left(A^{(1)}_{3}\right)^{c,\delta} & \left((2A_{1})^{(1)}\right)^{c,\delta} &
			\left(A_{1}^{(1)}\right)^{c,\delta} & \left(A_{0}^{(1)}\right)^{c} \\
			&&&&&\left(A_{2}^{(1)}\right)^{c,\delta} &
			\left(A_{1}^{(1)}\right)^{c,\delta} & \left(A_{0}^{(1)}\right)^{c} \\		
		  };
			\node [blue] at ($(m-2-4.north) + (0,0.1)$) {q-$P_{\text{VI}}$};
			\node [blue] at ($(m-3-4.south) + (0,-0.1)$) {$P_{\text{VI}}$,\! d-$P_{\text{V}}$};
			\node [blue] at ($(m-2-5.north) + (0,0.1)$) {q-$P_{\text{V}}$};
			\node [blue] at ($(m-1-8.north) + (0,0.1)$) {q-$P_{\text{I}}$};			
			\node [blue] at ($(m-3-5.south) + (0,-0.1)$) {$P_{\text{V}}$,\! d-$P_{\text{IV}}$};
			\node [blue] at ($(m-3-5.south) + (0,-0.4)$) {d-$P_{\text{III}}$};
			\node [blue] at ($(m-2-6.north) + (0,0.1)$) {q-$P_{\text{IV}}$,\! q-$P_{\text{III}}$};	  
			\node [blue] at ($(m-3-6.south) + (0,-0.1)$) {$P_{\text{III}}$};
			\node [blue] at ($(m-3-6.south) + (0.1,-0.4)$) {alt.\! d-$P_{\text{II}}$};
			\node [blue] at ($(m-4-6.south) + (0,-0.1)$) {$P_{\text{IV}}$,\! d-$P_{\text{II}}$};
			\node [right,scale=0.8] at ($(m-1-8.south east) + (-0.3,0.15)$) {${|\alpha|^{2}=8}$};	  
			\node [blue] at ($(m-4-7.south) + (0,-0.1)$) {$P_{\text{II}}$,\! alt.d-$P_{\text{I}}$};
			\node [blue] at ($(m-4-8.south) + (0,-0.1)$) {$P_{\text{I}}$};
			\node [right,scale=0.8] at ($(m-2-7.south east) + (-0.95,0.02)$) {${|\alpha|^{2}=14}$};	  
			\node [right,scale=0.8] at ($(m-2-8.south east) + (-0.65,0.02)$) {${|\alpha|^{2}=4}$};	  
			\draw[->]
			(m-1-1) edge (m-2-1)
			(m-2-1) edge (m-3-1)
			(m-2-2) edge (m-3-2)
			(m-2-3) edge (m-3-3)
			(m-2-4) edge (m-3-4)
			(m-2-5) edge (m-3-5)
			(m-2-6) edge (m-3-6)
			(m-2-7) edge (m-3-7)
			(m-2-8) edge (m-3-8)
			(m-3-5) edge (m-4-6)
			(m-3-6) edge (m-4-7)
			(m-3-8) edge (m-4-8)
			(m-2-7) edge (m-1-8)
			(m-1-8) edge (m-2-9)
			(m-2-9) edge (m-4-8)
			(m-2-8) edge (m-4-7)
			($(m-2-1.east) + (-0.15,0)$) edge ($(m-2-2.west) + (0.05,0)$)
			($(m-2-2.east) + (-0.15,0)$) edge ($(m-2-3.west) + (0.05,0)$)
			($(m-2-3.east) + (-0.15,0)$) edge ($(m-2-4.west) + (0.05,0)$)
			($(m-2-4.east) + (-0.15,0)$) edge ($(m-2-5.west) + (0.05,0)$)
			($(m-2-5.east) + (-0.15,0)$) edge ($(m-2-6.west) + (0.05,0)$)
			($(m-2-6.east) + (-0.15,0)$) edge ($(m-2-7.west) + (0.05,0)$)
			($(m-2-7.east) + (-0.15,0)$) edge ($(m-2-8.west) + (0.05,0)$)
			%
			($(m-3-1.east) + (-0.15,0)$) edge ($(m-3-2.west) + (0.05,0)$)
			($(m-3-2.east) + (-0.15,0)$) edge ($(m-3-3.west) + (0.05,0)$)
			($(m-3-3.east) + (-0.15,0)$) edge ($(m-3-4.west) + (0.05,0)$)
			($(m-3-4.east) + (-0.15,0)$) edge ($(m-3-5.west) + (0.05,0)$)
			($(m-3-5.east) + (-0.15,0)$) edge ($(m-3-6.west) + (0.05,0)$)
			($(m-3-6.east) + (-0.15,0)$) edge ($(m-3-7.west) + (0.05,0)$)
			($(m-4-6.east) + (-0.15,0)$) edge ($(m-4-7.west) + (0.05,0)$)
			($(m-4-7.east) + (-0.15,0)$) edge ($(m-4-8.west) + (0.05,0)$)
			($(m-3-7.east) + (-0.15,0)$) -> ($(m-3-8.west) + (0.05,0)$)
			;
			\draw [-] ($(m-2-6.south) + (-0.48,0)$) to [bend right = 50] ($(m-4-6.north) + (-0.48,0.02)$);
		\end{tikzpicture}
	\end{equation*}}	
	\caption{Inclusion scheme for the symmetry sub-lattices and corresponding Painlev\'e equations}
	\label{fig:pain-diag}
\end{figure}
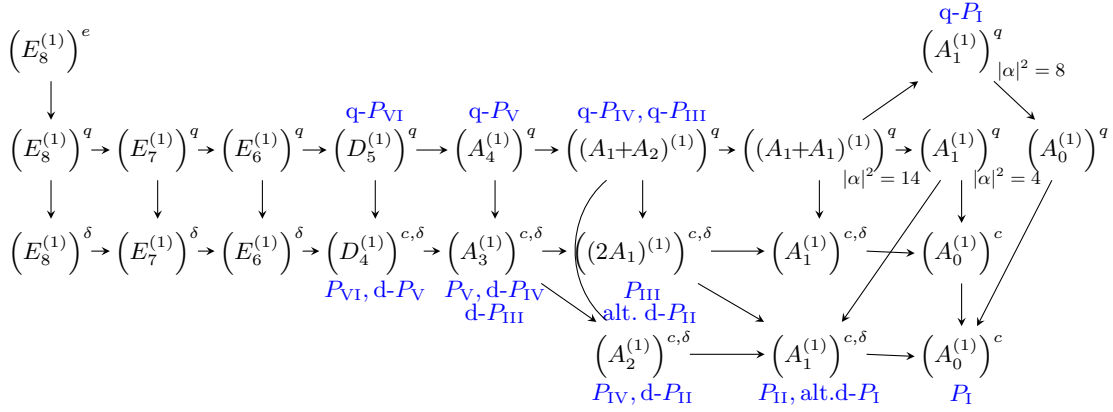

The following question then becomes natural and important: \emph{How to represent these new purely discrete
equations in the isomonodromic framework?} This question was posed by Sakai in \cite{Sak:2007:PDPEATLF} (\textbf{Problem A} 
for the difference case and \textbf{Problems B,C} for the $q$-difference case). 

More precisely, in both continuous and discrete difference case we start with some Fuchsian system and consider its 
isomonodromic deformations. In the continuous case, deformation parameters are
locations of singular points of the system. The resulting isomonodromic flows on the space of coefficients
of the Fuchsian system are given by Schlesinger \emph{equations}. In particular, for a
$2\times 2$ Fuchsian system with four poles, Schlesinger equations
reduce to the most general $P_{\text{VI}}$ equation. In the discrete difference case, deformation parameters
are the characteristic indices of the system and since the isomonodromy
condition requires that the indices change by integral shifts, the
resulting dynamic is indeed \emph{discrete}. It is expressed in the form
of \emph{difference} equations called Schlesinger \emph{transformations}. 
It is also possible to get the isomonodromic description of difference and $q$-difference Painlev\'e equations 
by studying deformations of linear difference and $q$-difference
analogues of Fuchsian systems, see \cite{JimSak:1996:AQOTSPE} and \cite{Bor:2004:ITLSDE} for details.

In \cite{DzhSakTak:2013:DSTTHFADPE} we studied in detail a particular class of Schlesinger 
transformations that are called \emph{elementary}. These transformations change only two of the
characteristic indices of the underlying Fuchsian system (and any other Schlesinger transformation
not involving characteristic indices with multiplicity can be represented as a composition of the 
elementary ones). In particular, we obtained explicit 
evolution equations governing the resulting discrete dynamic. Our objective for the present paper 
is to carefully and very explicitly describe reductions of these \emph{discrete Schlesinger evolution equations} 
to the \emph{difference} Painlev\'e equations.

Since the Painlev\'e equations are of second order, we focus on Fuchsian systems that have two-dimensional
moduli spaces (coordinates on such moduli space are known as accessory parameters). It turns out that, 
modulo two natural transformations called \emph{Katz's addition} and \emph{middle convolution} \cite{Kat:1996:RLS},  
that preserve the deformation equations \cite{HarFil:2007:MCADFFS}, there are only four such systems, \cite{Kos:2001:TDPFZIOR},
that have the \emph{spectral type} $(11,11,11,11)$, $(111,111,111)$, $(22,1111,1111)$, and $(33,222,111111)$ (spectral
type of a Fuchsian system encodes the degeneracies of the characteristic indices, or eigenvalues of residue
matrices at singular points, and it is carefully defined in the next section).  

Isomonodromic deformations of a $(11,11,11,11)$ spectral type Fuchsian system are well known --- continuous deformations reduce 
to Painlev\'e VI equation and Schlesinger transformations reduce to the difference Painlev\'e d-$P(D^{(1)}_{4})$ 
equation, also known as d-$P_{\text{V}}$, and in \cite{DzhSakTak:2013:DSTTHFADPE} we showed that 
in this case our \emph{discrete Schlesinger evolution equations} indeed can be reduced to the standard form of d-$P(D^{(1)}_{4})$.

In \cite{Boa:2009:QADPE} P.Boalch showed that for Fuchsian systems with the spectral types $(111,111,111)$, $(22,1111,1111)$, and 
$(33,222,111111)$ their Schlesinger transformations reduce to 
difference Painlev\'e equations with the required symmetry groups ${E}^{(1)}_{6}$, ${E}^{(1)}_{7}$, and ${E}^{(1)}_{8}$ respectively,
thus providing a theoretical answer to Sakai's \textbf{Problem A}. 

Our goal for the present paper is to make this statement
very concrete via explicit direct computation of reductions of discrete Schlesinger evolution to difference Painlev\'e equations 
with symmetry groups ${E}^{(1)}_{6}$ and ${E}^{(1)}_{7}$ (we plan to consider deformations of a Fuchsian system of the
spectral type $(33,222,111111)$ with the symmetry group ${E}^{(1)}_{8}$ elsewhere). In addition to establishing that the resulting 
difference Painlev\'e equations 
have the required types d-$P(A^{(1)*}_{2})$ and d-$P(A^{(1)*}_{1})$, we explicitly compare the resulting equations with the previously 
known instances of equations of the same type. We do so by finding an explicit identification between their Okamoto spaces of 
initial conditions, which allows us to compute and compare the translation directions for different equations w.r.t. the same root basis,
and also to match generic parameters to the characteristic indices of the Fuchsian system, which in turn allows us to 
see and compare these different equations via their actions on the Riemann scheme of our Fuchsian system. We show that in both examples 
elementary Schlesinger dynamic is indeed more elementary in the sense that standard examples of difference Painlev\'e equations 
can be realized as compositions of elementary Schlesinger transformations. Of particular interest here is
the ${E}^{(1)}_{7}$ case which has two characteristic indices of multiplicity $2$. We show that in that case the standard from 
of the equation can not be represented as a composition of elementary Schlesinger transformations of rank one. Thus we first generalize
our discrete Schlesinger evolution equations from \cite{DzhSakTak:2013:DSTTHFADPE} to elementary Schlesinger transformations of rank two,
and then show how to represent the standard dynamic as a composition of two such rank two transformations. 
We also provide a lot of details on how to compute the Okamoto space of initial conditions for our equations and how to identify 
two different instances of such spaces, hoping that this will be helpful 
for other researchers who are interested in the geometric approach to discrete Painlev\'e equations.

The paper is organized as follows. In Section~\ref{sec:preliminaries} we briefly describe our parameterization
of a Fuchsian system by its spectral and eigensystem data, define elementary Schlesinger transformations,
present evolution equations for elementary Schlesinger transformations as a dynamic on the space of coefficient of our
Fuchsian system, and then show how to split them to get the dynamic on the space of eigenvectors of the coefficient
matrices (this is a brief overview of our paper \cite{DzhSakTak:2013:DSTTHFADPE}). Next we show how to generalize 
this to elementary Schlesinger transformations of rank two, which is a new result.
In Section~\ref{sec:reductions}  we consider two examples of reductions of
the elementary Schlesinger transformation dynamic. The first example of a difference Painlev\'e equation of type 
d-$P\left(\tilde{A}_{2}^{*}\right)$ with the symmetry group $\tilde{E}_{6}$ was also briefly presented 
in \cite{DzhSakTak:2013:DSTTHFADPE}, here we go into a lot more detail and show how the choice of good coordinates, which was 
essentially guessed in \cite{DzhSakTak:2013:DSTTHFADPE}, is really forced on us by the geometric considerations. 
The next example of a $4\times4$ Fuchsian system of the \emph{spectral type} $22,1111,1111$ (i.e., with three poles and two double
eigenvalues at one pole) is a first example that we have which has degenerate eigenvalues, and this is a completely new result.
Finally, we give a brief summary in Section~\ref{sec:conclusion}.

\subsection*{Acknowledgements} 
Part of this work was done when A.D.~was visiting T.T.~at the Tokyo University of Marine Science and Technology and Nalini Joshi at 
the University of Sydney, and A.D.~would like to thank both Universities for the stimulating working environment and, together 
with the University of Northern Colorado, for the generous travel support. 


\section{Preliminaries} 
\label{sec:preliminaries}
The goal of this section is to write down evolution equations for elementary Schlesinger transformations, as well as
to introduce the necessary notation. Our presentation here is very brief and we refer the interested reader to 
\cite{DzhSakTak:2013:DSTTHFADPE} for details. The main new and important result of this section is the generalization
of equations governing elementary Schlesinger transformation dynamic on the decomposition space from rank-one to rank-two 
Schlesinger transformations.

\subsection{Fuchsian Systems} 
Consider a generic \emph{Fuchsian system}
(or a \emph{Fuchsian equation}) written in the 
\emph{Schlesinger normal form}: 
\begin{equation}	
	\frac{d \mathbf{Y}}{dz} = \mathbf{A}(z) \mathbf{Y} = \left( \sum_{i=1}^{n} \frac{ \mathbf{A}_{i} }{ z-z_{i} }  \right) \mathbf{Y},\qquad
	z_{i}\neq z_{j}\text{ for }i\neq j,
	\label{eq:fuchs-1}
\end{equation}
where $\mathbf{A}_{i} = \operatorname{ res }_{z_{i}} \mathbf{A}(z)\, dz$ are constant $m\times m$ matrices. In addition to 
simple poles at finite distinct points $z_{1},\dots, z_{n}$, this system  also has a simple pole at 
 $z_{0} = \infty\in \mathbb{P}^{1}$ if 
$\mathbf{A}_{\infty} = \operatorname{res}_{\infty} \mathbf{A}(z)\, dz = -\sum_{i=1}^{n} \mathbf{A}_{i}\neq \mathbf{0}$. The 
\emph{spectral data} of system~(\ref{eq:fuchs-1}) consists of locations of the simple poles $z_{1},\dots,z_{n}$ and the
eigenvalues (also called the \emph{characteristic indices}) $\theta_{i}^{j}$ of $\mathbf{A}_{i}$ and their multiplicities.
These multiplicities are encoded by the \emph{spectral type} of the system, 
\begin{equation*}
	\mathfrak{m} = m_{1}^{1}\cdots m_{1}^{l_{1}},	m_{2}^{1}\cdots m_{2}^{l_{2}}, \cdots, 
	m_{n}^{1}\cdots m_{n}^{l_{n}}, 	m_{\infty}^{1}\cdots m_{\infty}^{l_{\infty}},
\end{equation*}	 
where partitions $m = m_{i}^{1} + \cdots + m_{i}^{l_{i}}$, $m_{i}^{1}\geq\cdots\geq m_{i}^{l_{i}}\geq 1$ describe 
the multiplicities of the eigenvalues of $\mathbf{A}_{i}$. Spectral type classifies Fuchsian systems up to isomorphisms 
and the operations of addition and middle convolution. 

\subsection{Schlesinger Transformations} 
\label{sub:schlesinger_transformations}

Schlesinger transformations are discrete analogues of the usual Schlesinger differential equations describing isomonodromic 
deformations of our Fuchsian system. They are rational 
transformations preserving the singularity structure and the monodromy data of the system (\ref{eq:fuchs-1}), 
except for the integral shifts in the characteristic indices $\theta_{i}^{j}$, and so the coefficient matrix 
now depends on $\theta_{i}^{j}$,  $\mathbf{A} = \mathbf{A}(z; \mathbf{\Theta})$. Schlesinger transformations 
are given by the following \emph{differential--difference Lax Pair}:
\begin{equation*}
	\left\{ 
	\begin{aligned}
		\frac{d \mathbf{Y}}{dz} 
		& = \mathbf{A}(z;\mathbf{\Theta}) \mathbf{Y} = 
		\left( \sum_{i=1}^{n} \frac{ \mathbf{A}_{i}(\mathbf{\Theta}) }{ z-z_{i} }  \right) \mathbf{Y}, \\
		\bar{\mathbf{Y}}(z) &= \mathbf{R}(z) \mathbf{Y}(z). 
	\end{aligned}
	\right.,
\end{equation*}
where 
$\mathbf{R}(z)$ is a specially chosen rational matrix function called the \emph{multiplier} of the transformation.
The compatibility condition of this Lax Pair is 
\begin{equation}
	\bar{\mathbf{A}}(z;\mathbf{\Theta})\mathbf{R}(z) = \mathbf{R}(z) \mathbf{A}(z;\mathbf{\Theta}) + \frac{d \mathbf{R}(z)}{dz}. \label{eq:sch-transf}
\end{equation}

In \cite{DzhSakTak:2013:DSTTHFADPE} we considered a special class of Schlesinger transformations for which the multiplier matrix
has the form 
\begin{equation}
	\mathbf{R}(z) = \mathbf{I} + \frac{ z_{0} -  \zeta_{0} }{ z - z_{0} }\mathbf{P},\quad \text{where } 
	\mathbf{P}=\mathbf{P}^{2}\text{ is a \emph{projector}}. \label{eq:mult-proj-gen}
\end{equation}
It turns out that in this case it is possible to solve equation~(\ref{eq:sch-transf}) explicitly to obtain a discrete dynamic on 
the space of coefficient matrices. Namely, after substituting $\mathbf{R}(z)$ of the form~(\ref{eq:mult-proj-gen}) in~(\ref{eq:sch-transf})
(and its inverse) we immediately see that $z_{0},\zeta_{0}\in\{z_{i}\}_{i=1}^{n}$, and if we put $z_{0} = z_{\alpha}$ and
$\zeta_{0} = z_{\beta}$, we get the following equations on the coefficient matrices:
\begin{align}
	\bar{\mathbf{A}}_{i} &= \mathbf{R}(z_{i}) \mathbf{A}_{i} \mathbf{R}^{-1}(z_{i})\quad\text{for}\quad i\neq \alpha,\beta\qquad
	\text{(and therefore $\bar{\mathbf{\Theta}}_{i} = \mathbf{\Theta}_{i}$)}, \notag\\
	\mathbf{Q}\bar{\mathbf{A}}_{\alpha} &= \mathbf{A}_{\alpha}\mathbf{Q},\phantom{-\mathbf{P}}\qquad\, 
	\bar{\mathbf{A}}_{\beta} \mathbf{Q} = \mathbf{Q} \mathbf{A}_{\beta},
	\label{eq:Q-cond}\\
	\bar{\mathbf{A}}_{\alpha}\mathbf{P} &= \mathbf{P}\mathbf{A}_{\alpha}  - \mathbf{P},\qquad
	\mathbf{P}\bar{\mathbf{A}}_{\beta} = \mathbf{A}_{\beta}\mathbf{P} + \mathbf{P},\label{eq:P-cond}\\
	\bar{\mathbf{A}}_{\alpha} & = \mathbf{A}_{\alpha} + 
	\sum_{i\neq \alpha}\left(\frac{ z_{\alpha} - z_{\beta} }{ z- z_{\alpha} }\right)(\mathbf{P}\mathbf{A}_{i} - \bar{\mathbf{A}}_{i}\mathbf{P}),\notag\\
	\bar{\mathbf{A}}_{\beta} & = \mathbf{A}_{\beta} + 
	\sum_{i\neq \beta}\left(\frac{ z_{\beta} - z_{\alpha} }{ z- z_{\beta} }\right)(\mathbf{A}_{i}\mathbf{P} - \mathbf{P}\bar{\mathbf{A}}_{i}).\notag
\end{align}
Then either~(\ref{eq:Q-cond}) or~(\ref{eq:P-cond}) imposes important constraints on the projector $\mathbf{P}$:
\begin{equation}
	\mathbf{P}\mathbf{A}_{\alpha}\mathbf{Q} = \mathbf{0}\quad(\text{or } \mathbf{P}\mathbf{A}_{\alpha}\mathbf{P} = \mathbf{P}\mathbf{A}_{\alpha}),\qquad
	\mathbf{Q}\mathbf{A}_{\beta}\mathbf{P} = \mathbf{0}\quad(\text{or } \mathbf{P}\mathbf{A}_{\beta}\mathbf{P} = \mathbf{A}_{\beta}\mathbf{P}),
	\label{eq:P-constr}
\end{equation}
and if this condition is satisfied, we get the following dynamic on the space of coefficient matrices:
\begin{align}
	\bar{\mathbf{A}}_{i} &= \mathbf{R}(z_{i}) \mathbf{A}_{i} \mathbf{R}^{-1}(z_{i}),\qquad i\neq \alpha,\beta, \label{eq:Ai}\\
	\bar{\mathbf{A}}_{\alpha} &= \mathbf{A}_{\alpha} - \mathbf{Q}\mathbf{A}_{\alpha}\mathbf{P}
	- \mathbf{P} + \sum_{i\neq \alpha}\left(\frac{ z_{\beta} - z_{\alpha} }{ z_{i} - z_{\alpha} }\right)\mathbf{P}\mathbf{A}_{i}\mathbf{Q},
	\label{eq:Aalpha}\\
	\bar{\mathbf{A}}_{\beta} &= \mathbf{A}_{\beta} - \mathbf{P}\mathbf{A}_{\beta}\mathbf{Q} + \mathbf{P}
	+ \sum_{i\neq \beta} \left(\frac{ z_{\alpha} - z_{\beta} }{ z_{i} - z_{\beta} }\right)\mathbf{Q}\mathbf{A}_{i}\mathbf{P}.\label{eq:Abeta}
\end{align}
Indeed, 
\begin{align*}
	\bar{\mathbf{A}}_{\alpha} &= \bar{\mathbf{A}}_{\alpha}\mathbf{P} + \bar{\mathbf{A}}_{\alpha}\mathbf{Q}
	= \mathbf{P} \mathbf{A}_{\alpha} - \mathbf{P} + \mathbf{A}_{\alpha} (\mathbf{I} - \mathbf{P}) + 
	\sum_{i\neq \alpha} \left(\frac{ z_{\alpha} - z_{\beta} }{ z_{\alpha} - z_{i} }\right)\mathbf{P} \mathbf{A}_{i}\mathbf{Q}\\
	&= \mathbf{A}_{\alpha} - \mathbf{Q}\mathbf{A}_{\alpha}\mathbf{P}
	- \mathbf{P} +  \sum_{i\neq \alpha}\left(\frac{ z_{\beta} - z_{\alpha} }{ z_{i} - z_{\alpha} } \right)\mathbf{P} \mathbf{A}_{i}\mathbf{Q},
\end{align*}
and the equation for $\bar{\mathbf{A}}_{\beta}$ is obtained in a similar fashion. 
We call these equations \emph{Discrete Schlesinger Evolution Equations}.


\subsection{The Decomposition Space} 
\label{sub:the_decomposition_space}
It is sometimes more convenient to extend the dynamic from the space of coefficients of the Fuchsian system
to the space of eigenvectors of the coefficient matrices, we call this space the \emph{decomposition space}. 
In particular, this is the space on which both the
continuous (\cite{JimMiwMorSat:1980:DMIBFPT}) and the discrete (\cite{DzhSakTak:2013:DSTTHFADPE}) Hamiltonian 
equations for Schlesinger deformations can be written. Before defining this space it is convenient to 
reduce the number of parameters in our system by using scalar gauge transformations of the form 
$\tilde{\mathbf{Y}}(z) = w(z)^{-1} \mathbf{Y}(z)$, where $w(z)$ is a solution of the \emph{scalar} equation
\begin{equation*}
	\frac{ dw }{ dz } = \sum_{i=1}^{n} \frac{ \theta_{i}^{j}}{ z - z_{i} }w.
\end{equation*}
Such transformations change the residue matrices by $\tilde{\mathbf{A}}_{i} = \mathbf{A}_{i} - \theta_{i}^{j} \mathbf{I}$
(and consequently change the residue matrix at infinity by $\tilde{\mathbf{A}}_{\infty} = \mathbf{A}_{\infty} + \theta_{i}^{j} \mathbf{I}$). Hence we can
always make one of the eigenvalues $\theta_{i}^{j}=0$ by choosing a good representative w.r.t.~the action 
by the group of local \emph{scalar gauge transformations}. So we make the following assumption.
\begin{assumption}\label{assume:rank-reduce}
	We always assume that at the finite point $z_{i}$ the 
	eigenvalue  $\theta_{i}^{1}$ of	the highest multiplicity $m_{i}^{1}$ is zero. 
\end{assumption}	

We also need the following important \emph{semi-simplicity} assumption. 

\begin{assumption}\label{assume:diagonal}
	We assume that the coefficient matrices $\mathbf{A}_{i}$ are \emph{diagonalizable} (even when we have multiple eigenvalues).
\end{assumption}	

In view of these assumptions, coefficient matrices $\mathbf{A}_{i}$ are similar to  diagonal matrices 
$\operatorname{diag}\{\theta_{i}^{1},\dots, \theta_{i}^{r_{i}},0,\dots,0\}$,
where  $r_{i} = \operatorname{rank}(\mathbf{A}_{i})$. 
Omitting the zero eigenvalues, we put 
\begin{equation}
\mathbf{\Theta}_{i} = \operatorname{diag}\{\theta_{i}^{1},\dots, \theta_{i}^{r_{i}}\}.\label{eq:theta-i}	
\end{equation}
Further, in view of Assumption~(\ref{assume:diagonal}) there exist full sets of \emph{right} eigenvectors 
$\mathbf{b}_{i,j}$, $\mathbf{A}_{i} \mathbf{b}_{i,j} = \theta_{i}^{j} \mathbf{b}_{i,j}$, 
and \emph{left} eigenvectors $\mathbf{c}_{i}^{j\dag}$, $\mathbf{c}_{i}^{j\dag} \mathbf{A}_{i} = \theta_{i}^{j} \mathbf{c}_{i}^{j\dag}$ 
(we use the $\dag$ symbol to indicate a \emph{row}-vector or a matrix of row vectors).  In the matrix form, omitting vectors with indices 
$j>r_{i}$ that are in the kernel of $\mathbf{A}_{i}$, we can write 
\begin{equation*}
	\mathbf{B}_{i} = \begin{bmatrix}	 \mathbf{b}_{i,1} \cdots \mathbf{b}_{i,r_{i}}	\end{bmatrix}, \quad
	\mathbf{A}_{i} \mathbf{B}_{i} = \mathbf{B}_{i} \mathbf{\Theta}_{i},\quad 
	\mathbf{C}_{i}^{\dag}  = \begin{bmatrix} \mathbf{c}_{i}^{1\dag}\\ \vdots \\ \mathbf{c}_{i}^{r_{i}\dag}	\end{bmatrix}, \quad
	\mathbf{C}_{i}^{\dag} \mathbf{A}_{i} =  \mathbf{\Theta}_{i} \mathbf{C}_{i}^{\dag},
\end{equation*}
with $\mathbf{\Theta}_{i}$ defined by (\ref{eq:theta-i}). Then we have a decomposition
$\mathbf{A}_{i} = \mathbf{B}_{i} \mathbf{C}_{i}^{\dag}$, provided that 
$\mathbf{C}_{i}^{\dag} \mathbf{B}_{i} = \mathbf{\Theta}_{i}$. We call this last condition the 
\emph{orthogonality condition} (since $\mathbf{\Theta}_{i}$ is diagonal) and we assume that it holds even when we have 
repeating eigenvalues. This
condition is related to the normalization ambiguity of the eigenvectors.
Thus, given $\mathbf{A}_{i}$, we can construct (in a non-unique way) a corresponding decomposition pair 
$(\mathbf{B}_{i}, \mathbf{C}_{i}^{\dag})$.
The space of all such pairs for all finite indices $1\leq i\leq n$, without any additional conditions, 
is our \emph{decomposition space}. We denote it as 
\begin{align*}
	\mathcal{B}\times \mathcal{C} &= 
	(\mathbb{C}^{r_{1}}\times\cdots\times \mathbb{C}^{r_{n}})\times ((\mathbb{C}^{r_{1}})^{\dag}\times \cdots\times (\mathbb{C}^{r_{n}})^{\dag})\\
	&\simeq 
	(\mathbb{C}^{r_{1}}\times(\mathbb{C}^{r_{1}})^{\dag})\times\cdots\times (\mathbb{C}^{r_{n}}\times (\mathbb{C}^{r_{n}})^{\dag})
\end{align*}
and write an element $(\mathbf{B},\mathbf{C}^{\dag})$ of this space as a list of $n$ pairs 
$(\mathbf{B}_{1}, \mathbf{C}_{1}^\dag; \cdots; \mathbf{B}_{n}, \mathbf{C}_{n}^{\dag})$. 
Then, given a Riemann Scheme of a Fuchsian system (equivalently, a collection $\mathbf{\Theta} = \{\theta_{i}^{j}\}$ of the characteristic indices
having the correct multiplicities and satisfying the Fuchs relation), 
we denote by 
\begin{equation}
(\mathcal{B}\times \mathcal{C})_{\mathbf{\Theta}} = 
	\{(\mathbf{B}_{1}, \mathbf{C}_{1}^\dag; \cdots; \mathbf{B}_{n}, 
	\mathbf{C}_{n}^{\dag})\in \mathcal{B}\times \mathcal{C}
	\mid \mathbf{C}_{i}^{\dag} \mathbf{B}_{i} = \mathbf{\Theta}_{i},
	\sum_{i=1}^{n} \mathbf{B}_{i} \mathbf{C}_{i}^{\dag} = \mathbf{A}_{\infty}\sim  \mathbf{\Theta}_{\infty}\}
\end{equation}
the corresponding fiber in the decomposition space (since for Schlesinger transformations locations of the poles are just fixed parameters of the dynamic,
we occasionally omit them, as in the above notation). 

\begin{remark}\label{rem:transf}
	There are two natural actions on the decomposition space $\mathcal{B}\times \mathcal{C}$. First, the group  
	$\mathbb{GL}_{m}$ of global gauge transformations of the Fuchsian system induces the following
	action. Given $\mathbf{P}\in \mathbb{GL}_{m}$, 
	we have the action $\mathbf{A}_{i}\mapsto \mathbf{P} \mathbf{A}_{i} \mathbf{P}^{-1}$ which translates into the action
	$(\mathbf{B}_{i},\mathbf{C}_{i}^{\dag}) \mapsto 
	(\mathbf{P}\mathbf{B}_{i},\mathbf{C}_{i}^{\dag} \mathbf{P}^{-1})$. We refer to such transformations as
	\emph{similarity transformations}. It is often necessary to restrict this action to the subgroup 
	$G_{\mathbf{A}_{\infty}}$ preserving the form of $\mathbf{A}_{\infty}$. Second, for any pair
	$(\mathbf{B}_{i}, \mathbf{C}_{i}^{\dag})$ the pair
	$(\mathbf{B}_{i}\mathbf{Q}_{i}, \mathbf{Q}_{i}^{-1}\mathbf{C}_{i}^{\dag})$  
	determines the same matrix $\mathbf{A}_{i}$ for $\mathbf{Q}_{i}\in \mathbb{GL}_{r_{i}}$. 
	The condition 
	$\mathbf{Q}_{i}^{-1}\mathbf{C}_{i}^{\dag} \mathbf{B}_{i}\mathbf{Q}_{i} = 
	\mathbf{Q}_{i}^{-1}\mathbf{\Theta}_{i}\mathbf{Q}_{i} = \mathbf{\Theta}_{i}$
	restricts $\mathbf{Q}_{i}$ to the stabilizer subgroup $G_{\mathbf{\Theta}_{i}}$ of $\mathbb{GL}_{r_{i}}$.
	In particular, when all $\theta_{i}^{j}$ are distinct, 
	$\mathbf{Q}_{i}$ has to be a diagonal matrix.
	We refer to such transformations as \emph{trivial transformations}.  These two actions obviously commute
	with each other. The phase space for the Schlesinger dynamic is the quotient space of $(\mathcal{B}\times \mathcal{C})_{\mathbf{\Theta}}$
	by this action. 
\end{remark}


\subsection{Schlesinger Dynamic on the Decomposition Space} 
\label{sub:schlesinger_dynamic_on_the_decomposition_space}
In this section we explain how to lift the Schlesinger Evolution equations to the decomposition space.

\subsubsection{Rank One} 
\label{ssub:rank_one}
In \cite{DzhSakTak:2013:DSTTHFADPE} we focused on 
the \emph{elementary} Schlesinger transformations $\left\{\begin{smallmatrix}
	\alpha&\beta\\\mu&\nu
\end{smallmatrix}\right\}$ that only change two of the characteristic indices by 
unit shifts, i.e., $\bar{\theta}_{\alpha}^{\mu} = \theta_{\alpha}^{\mu} - 1$ and 
$\bar{\theta}_{\beta}^{\nu} = \theta_{\beta}^{\nu} + 1$, $\alpha\neq \beta$. For such transformations
the projector matrix $\mathbf{P}$ has rank one and the multiplier matrix has the form~(\ref{eq:mult-proj-gen}) with 
\begin{equation}
\mathbf{R}(z) = \mathbf{I} + \frac{ z_{\alpha} - z_{\beta} }{ z - z_{\alpha} }  \mathbf{P},\quad
\text{ where }	\mathbf{P} = \frac{ \mathbf{b}_{\beta,\nu} \mathbf{c}_{\alpha}^{\mu\dag } }{ \mathbf{c}_{\alpha}^{\mu\dag } \mathbf{b}_{\beta,\nu}  },
\quad\text{and we put } \mathbf{Q} = \mathbf{I} - \mathbf{P}.\label{eq:mult-elem}
\end{equation}
In this case, under the semi-simplicity Assumption~(\ref{assume:diagonal}) 
it is possible to decompose equations~(\ref{eq:Ai}--\ref{eq:Abeta}) to get the dynamic on the space 
$(\mathcal{B}\times \mathcal{C})_{\mathbf{\Theta}}$. 

\begin{theorem}[\cite{DzhSakTak:2013:DSTTHFADPE}]\label{thm:evolution}
		An elementary Schlesinger transformation
		$\left\{\begin{smallmatrix}
			\alpha&\beta\\\mu&\nu
		\end{smallmatrix}\right\}$
		defines the map 
		\begin{equation*}
			(\mathcal{B}\times \mathcal{C})_{\mathbf{\Theta}} \to 
			(\bar{\mathcal{B}}\times \bar{\mathcal{C}})_{\bar{\mathbf{\Theta}}}	
		\end{equation*}
		given by the following evolution equations (grouped for convenience), where $c_{i}^{j}$ are arbitrary 
		non-zero constants.
		\begin{enumerate}[(i)]
			\item Transformation vectors:
			\begin{equation}
				\bar{\mathbf{b}}_{\alpha,\mu} = \frac{ 1 }{ c_{\alpha}^{\mu} } \mathbf{b}_{\beta,\nu},\qquad
				\bar{\mathbf{c}}_{\beta}^{\nu\dag} = c_{\beta}^{\nu} \mathbf{c}_{\alpha}^{\mu\dag}.
				\label{eq:bb-cb-generators}				
			\end{equation}
			\item Generic indices:
			\begin{align}
				\bar{\mathbf{b}}_{i,j} &= 
				\frac{ 1 }{ c_{i}^{j} } \mathbf{R}(z_{i}) \mathbf{b}_{i,j},\,(i\neq \alpha\text{ and if } i= \beta, j\neq \nu);
				\label{eq:bb-generic}\\
				\bar{\mathbf{c}}_{i}^{j\dag} &= c_{i}^{j}  \mathbf{c}_{i}^{j\dag} \mathbf{R}^{-1}(z_{i}),\, (i\neq 
				\beta\text{ and if } i=\alpha, j\neq \mu).
				\label{eq:cb-generic}
			\end{align}
			\item Special indices:
			\begin{align}
				\bar{\mathbf{b}}_{\alpha,j} &= \frac{ 1 }{ c_{\alpha}^{j} }\left(\mathbf{I} - 
				 \frac{ \mathbf{P}}{ \theta_{\alpha}^{\mu} - \theta_{\alpha}^{j} - 1 } \left(\sum_{i\neq \alpha} 
				\frac{ z_{\beta} - z_{\alpha}}{ z_{i} - z_{\alpha} }\mathbf{A}_{i}
				\right) \right)\mathbf{b}_{\alpha,j},\qquad j\neq \mu;
				\label{eq:bb-ak}\\
				\bar{\mathbf{c}}_{\beta}^{j\dag} &= c_{\beta}^{j}\mathbf{c}_{\beta}^{j\dag} \left(
				\mathbf{I} - \left(\sum_{i\neq \beta}
				\frac{ z_{\alpha} - z_{\beta} }{ z_{i} - z_{\beta} } \mathbf{A}_{i}\right)
				\frac{ \mathbf{P}}{ 
				\theta_{\beta}^{\nu} - \theta_{\beta}^{j} + 1 } 
				\right),\qquad j\neq \nu;
				\label{eq:cb-bk}\\
				\bar{\mathbf{b}}_{\beta,\nu} &= \frac{ 1 }{ c_{\beta}^{\nu}  }
				\left( (\theta_{\beta}^{\nu} + 1) \mathbf{I} + 
				\phantom{\left(
				\sum_{i\neq \beta} \frac{ z_{\alpha} - z_{\beta}  }{ 
				z_i - z_{\beta}  } \mathbf{A}_{i}\right)} 
				\right. \notag\\
				&\qquad \left.
				\mathbf{Q}	
				\left(\mathbf{I} +  \sum_{j\neq \nu}\frac{ \mathbf{b}_{\beta,j} \mathbf{c}_{\beta}^{j\dag}}{ 
				\theta_{\beta}^{\nu} - \theta_{\beta}^{j} + 1 }\right) \left(
				\sum_{i\neq \beta} \frac{ z_{\alpha} - z_{\beta}  }{ 
				z_i - z_{\beta}  } \mathbf{A}_{i}\right)\right)
				\frac{ \mathbf{b}_{\beta,\nu} }{  \mathbf{c}_{\alpha}^{\mu\dag} \mathbf{b}_{\beta,\nu}  };
				\label{eq:bb-bn}\\
				\bar{\mathbf{c}}_{\alpha}^{\mu\dag} &= c_{\alpha}^{\mu} 
				\frac{ \mathbf{c}_{\alpha}^{\mu\dag} }{ \mathbf{c}_{\alpha}^{\mu\dag} \mathbf{b}_{\beta,\nu}  }\left(
				 (\theta_{\alpha}^{\mu} -  1) \mathbf{I} + 
				\phantom{\left(\sum_{i\neq \alpha}\frac{ z_{\beta} - z_{\alpha}
				 }{ z_{i} - z_{\alpha}  } \mathbf{A}_{i}\right)} 
				\right.	 \notag \\
				&\qquad \left. \left(\sum_{i\neq \alpha}\frac{ z_{\beta} - z_{\alpha}
				 }{ z_{i} - z_{\alpha}  } \mathbf{A}_{i}\right)
				\left( \mathbf{I} + \sum_{j\neq \mu} 
				\frac{ \mathbf{b}_{\alpha,j} \mathbf{c}_{\alpha}^{j\dag} }{ \theta_{\alpha}^{\mu} - \theta_{\alpha}^{j} -1 }
				\right) \mathbf{Q}
				\right).\label{eq:cb-am}
			\end{align}
		\end{enumerate}
		
\end{theorem}


\subsubsection{Rank Two} 
\label{ssub:rank_two}
For the difference Painlev\'e equation d-$P({A}_{1}^{(1)*})$ we need to study Schlesinger transformations of a Fuchsian system 
that has the spectral type $22,1111,1111$, and so we need to consider Schlesinger transformations that change not one but two eigenvalues
at each point $z_{\alpha}$ and $z_{\beta}$. In this section we show how to obtain the corresponding dynamic on the decomposition space.
The resulting equations suggest what happens in the general case of a projector $\mathbf{P}$ of arbitrary rank, but since the focus of the
present paper is on examples, we plan to consider the general case elsewhere.

Naively, we want to consider Schlesinger transformations of the form
\begin{equation}
	\left\{\begin{smallmatrix}
		\alpha & \beta \\ \mu_{1} & \nu_{1} \\ \mu_{2} & \nu_{2}
	\end{smallmatrix}\right\}	
	= \left\{\begin{smallmatrix}
		\alpha & \beta \\ \mu_{1} & \nu_{1} 
	\end{smallmatrix}\right\}	
	\circ
	\left\{\begin{smallmatrix}
			\alpha & \beta \\ \mu_{2} & \nu_{2} 
	\end{smallmatrix}\right\} = 
	\left\{\begin{smallmatrix}
		\alpha & \beta \\ \mu_{2} & \nu_{2} 
	\end{smallmatrix}\right\}	
	\circ
	\left\{\begin{smallmatrix}
			\alpha & \beta \\ \mu_{1} & \nu_{1} 
	\end{smallmatrix}\right\}.
\end{equation}
However, if one of the characteristic indices (say, $\alpha$) has multiplicity,
applying a rank-one elementary Schlesinger transformation will change the 
spectral type of the equation (e.g., in our example, a rank-one transformation
$\left\{\begin{smallmatrix}
	1 & 2 \\ 1 & 1 
\end{smallmatrix}\right\}$ maps the moduli space of Fuchsian equations 
of spectral type $(22,1111,1111)$ to a smaller moduli space $(112, 1111,1111)$,
and in fact our formulas in this case do not work, since some of the expressions
become undefined). Thus, we need to develop the rank-two version of the 
elementary Schlesinger transformation separately. We start with a composition of
two rank-one maps to get an insight on the structure of the multiplier matrix
in the rank-two case, but then proceed to derive the dynamic equations independently.
The resulting equations are then defined on moduli spaces of Fuchsian systems that 
have multiplicity in the spectral type (e.g., in our example, the map 
$\left\{\begin{smallmatrix}
	1 & 2 \\ 1 & 1 \\ 2 & 2
\end{smallmatrix}\right\}$
is defined on both moduli spaces of Fuchsian systems of spectral 
type $1111,1111,1111$ and $22,1111,1111$). 
%
So we start with the multiplier matrix that is a product (and for simplicity we put  $\mu_{i} = i$ and $\nu_{j} = j$ for this derivation)
\begin{align*}
	\mathbf{R}(z) &= \bar{\mathbf{R}}_{1}(z) \mathbf{R}_{2}(z) = 
	\left(\mathbf{I} + \frac{ z_{\alpha} - z_{\beta} }{ z - z_{\alpha} } \bar{\mathbf{P}}_{1}\right)
	\left(\mathbf{I} + \frac{ z_{\alpha} - z_{\beta} }{ z - z_{\alpha} } \mathbf{P}_{2}\right),\\
	\intertext{where, in view of (\ref{eq:mult-elem}) and (\ref{eq:bb-generic}--\ref{eq:cb-generic}),}
	\mathbf{P}_{i} &= \frac{ \mathbf{b}_{\beta,i} \mathbf{c}_{\alpha}^{i\dag} }{ \mathbf{c}_{\alpha}^{i\dag} \mathbf{b}_{\beta,i}  },\qquad	
	\text{and}\qquad
	\bar{\mathbf{P}}_{1} = 
	\frac{ \bar{\mathbf{b}}_{\beta,1} \bar{\mathbf{c}}_{\alpha}^{1\dag} }{ \bar{\mathbf{c}}_{\alpha}^{1\dag} \bar{\mathbf{b}}_{\beta,1}  } = 
	\frac{ \mathbf{Q}_{2}\mathbf{b}_{\beta,1} \mathbf{c}_{\alpha}^{1\dag}\mathbf{Q}_{2} }{ 
	\mathbf{c}_{\alpha}^{1\dag} \mathbf{Q}_{2} \mathbf{b}_{\beta,1}  } = 
	\frac{ \mathbf{Q}_{2} \mathbf{P}_{1} \mathbf{Q}_{2} }{ \operatorname{Tr}(\mathbf{Q}_{2} \mathbf{P}_{1}) }.
	%
	\intertext{Here $\mathbf{Q}_{i} = \mathbf{I} - \mathbf{P}_{i}$ is, as usual, the complementary projector. Then, since clearly 
	$\bar{\mathbf{P}}_{1}\mathbf{P}_{2} = \mathbf{0}$,	}
	\mathbf{R}(z) &= \mathbf{I} + \frac{ z_{\alpha} - z_{\beta} }{ z - z_{\alpha} } \mathcal{P},\qquad\text{where } 
	\mathcal{P} = \bar{\mathbf{P}}_{1} + \mathbf{P}_{2}\text{ is also a projector.}
\end{align*}
Let us now rewrite $\mathcal{P}$ in a more symmetric form. First note that, since $\mathbf{P}_{i}$ are rank-one projectors, 
\begin{equation*}
	\operatorname{Tr}(\mathbf{Q}_{2}\mathbf{P}_{1}) = 	\operatorname{Tr}(\mathbf{P}_{1} - \mathbf{P}_{2}\mathbf{P}_{1}) = 
		1 - \operatorname{Tr}(\mathbf{P}_{1}\mathbf{P}_{2}) = \operatorname{Tr}(\mathbf{Q}_{1}\mathbf{P}_{2}). 
\end{equation*}
Also, note that for any rank-one projector $\mathbf{S}$ and for any matrix $\mathbf{M}$ we have an identity 
$\mathbf{SMS} = \operatorname{Tr}(\mathbf{MS}) \mathbf{S}$.
Therefore,
\begin{align*}
	\mathcal{P} &= \frac{ \mathbf{Q}_{2}\mathbf{P}_{1} \mathbf{Q}_{2} + \mathbf{P}_{2}\mathbf{Q}_{1}\mathbf{P}_{2} }{ 
	\operatorname{Tr}(\mathbf{Q}_{2}\mathbf{P}_{1}) } = 
	\frac{\mathbf{Q}_{2}\mathbf{P}_{1}}{\operatorname{Tr}(\mathbf{Q}_{2}\mathbf{P}_{1})} + 
	\frac{\mathbf{Q}_{1}\mathbf{P}_{2}}{\operatorname{Tr}(\mathbf{Q}_{1}\mathbf{P}_{2})} = \mathcal{P}_{1} + \mathbf{\mathcal{P}}_{2},
\end{align*}
where 
\begin{equation}
	\mathcal{P}_{1} = \frac{\mathbf{Q}_{2}\mathbf{P}_{1}}{\operatorname{Tr}(\mathbf{Q}_{2}\mathbf{P}_{1})} = 
	\frac{\mathbf{Q}_{2}\mathbf{b}_{\beta,1}\mathbf{c}_{\alpha}^{1\dag}}{\mathbf{c}_{\alpha}^{1\dag}\mathbf{Q}_{2}\mathbf{b}_{\beta,1}},\qquad
	\mathcal{P}_{2} = \frac{\mathbf{Q}_{1}\mathbf{P}_{2}}{\operatorname{Tr}(\mathbf{Q}_{1}\mathbf{P}_{2})} = 
	\frac{\mathbf{Q}_{1}\mathbf{b}_{\beta,2}\mathbf{c}_{\alpha}^{2\dag}}{\mathbf{c}_{\alpha}^{2\dag}\mathbf{Q}_{1}\mathbf{b}_{\beta,2}}
	\label{eq:P-projs}
\end{equation}
are two mutually orthogonal rank-one projectors, $\mathcal{P}_{i}^{2} = \mathcal{P}_{i}$,
$\mathcal{P}_{1} \mathcal{P}_{2} = \mathcal{P}_{2} \mathcal{P}_{1} = \mathbf{0}$. 
At the same time, since obviously $\mathbf{Q}_{2}\mathbf{P}_{1} + \mathbf{Q}_{1}\mathbf{P}_{2} = \mathbf{P}_{1}\mathbf{Q}_{2} + \mathbf{P}_{2}\mathbf{Q}_{1}$,
$\mathcal{P} = \tilde{\mathcal{P}}_{1} + \tilde{\mathcal{P}}_{2}$, where
\begin{equation}
	\tilde{\mathcal{P}}_{1} = \frac{\mathbf{P}_{1}\mathbf{Q}_{2}}{\operatorname{Tr}(\mathbf{Q}_{2}\mathbf{P}_{1})} = 
	\frac{\mathbf{b}_{\beta,1}\mathbf{c}_{\alpha}^{1\dag}\mathbf{Q}_{2}}{\mathbf{c}_{\alpha}^{1\dag}\mathbf{Q}_{2}\mathbf{b}_{\beta,1}},\qquad
	\tilde{\mathcal{P}}_{2} = \frac{\mathbf{P}_{2}\mathbf{Q}_{1}}{\operatorname{Tr}(\mathbf{Q}_{1}\mathbf{P}_{2})} = 
	\frac{\mathbf{b}_{\beta,2}\mathbf{c}_{\alpha}^{2\dag}\mathbf{Q}_{1}}{\mathbf{c}_{\alpha}^{2\dag}\mathbf{Q}_{1}\mathbf{b}_{\beta,2}}.
	\label{eq:Pt-projs}
\end{equation}
We also put $\mathcal{Q}_{i} = \mathbf{I} - \mathcal{P}_{i}$, $\tilde{\mathcal{Q}}_{i} = \mathbf{I} - \tilde{\mathcal{P}}_{i}$, and 
$\mathcal{Q} = \mathbf{I} - \mathcal{P}$.
In view of the orthogonality conditions $\mathbf{C}_{i}^{\dag}\mathbf{B}_{i} = \mathbf{\Theta}_{i}$, it is easy to describe the eigenvectors for each 
of those projectors (we do it just for $\mathcal{P}$s since for $\mathcal{Q}$s eigenvectors are the same but eigenvalues swap between $0$ and $1$, 
below we use the notation $(\theta,\mathbf{w}^{\dag},\mathbf{v})$, where $\theta$ is an eigenvalue (which is either $0$ or $1$ for projectors),
$\mathbf{w}^{\dag}$ is a row (or left) eigenvector and $\mathbf{v}$ is a column (or right) eigenvector):
\begin{align}
	\operatorname{Eigen}(\mathbf{P}_{i}) &=\{(1; \mathbf{c}_{\alpha}^{i\dag}, \mathbf{b}_{\beta,i}),\,
	(0; \mathbf{c}_{\beta}^{j\dag},\mathbf{b}_{\alpha,j})\text{ for } j\neq i \},\quad i=1,2;\\
	\operatorname{Eigen}(\mathcal{P}_{1}) &= \{(1; \mathbf{c}_{\alpha}^{1\dag},\mathbf{Q}_{2}\mathbf{b}_{\beta,1}), \,
	(0; \mathbf{c}_{\alpha}^{2\dag},\mathbf{b}_{\alpha,2}),\, (0; \mathbf{c}_{\beta}^{j\dag},\mathbf{b}_{\alpha,j})\text{ for }j>2   \};\\
	\operatorname{Eigen}(\mathcal{P}_{2}) &= \{(1; \mathbf{c}_{\alpha}^{2\dag},\mathbf{Q}_{1}\mathbf{b}_{\beta,2}), \,
	(0; \mathbf{c}_{\alpha}^{1\dag},\mathbf{b}_{\alpha,1}),\, (0; \mathbf{c}_{\beta}^{j\dag},\mathbf{b}_{\alpha,j})\text{ for }j>2   \};\\
	\operatorname{Eigen}(\tilde{\mathcal{P}}_{1}) &= \{(1; \mathbf{c}_{\alpha}^{1\dag}\mathbf{Q}_{2},\mathbf{b}_{\beta,1}), \,
	(0; \mathbf{c}_{\beta}^{2\dag},\mathbf{b}_{\beta,2}),\, (0; \mathbf{c}_{\beta}^{j\dag},\mathbf{b}_{\alpha,j})\text{ for }j>2   \};\\
	\operatorname{Eigen}(\tilde{\mathcal{P}}_{1}) &= \{(1; \mathbf{c}_{\alpha}^{2\dag}\mathbf{Q}_{1},\mathbf{b}_{\beta,2}), \,
	(0; \mathbf{c}_{\beta}^{1\dag},\mathbf{b}_{\beta,1}),\, (0; \mathbf{c}_{\beta}^{j\dag},\mathbf{b}_{\alpha,j})\text{ for }j>2   \};\\
	\operatorname{Eigen}(\mathcal{P}) &= \{(1;\mathbf{c}_{\alpha}^{1\dag}, \mathbf{Q}_{2}\mathbf{b}_{\beta,1}),\, 
	(1; \mathbf{c}_{\alpha}^{2\dag}, \mathbf{Q}_{1}\mathbf{b}_{\beta,2}),\, (0; \mathbf{c}_{\beta}^{j\dag},\mathbf{b}_{\alpha,j})\text{ for }j>2 \}\\
										&= \{(1;\mathbf{c}_{\alpha}^{1\dag}\mathbf{Q}_{2},\mathbf{b}_{\beta,1}),\, 
	(1; \mathbf{c}_{\alpha}^{2\dag}\mathbf{Q}_{1},\mathbf{b}_{\beta,2}),\, (0; \mathbf{c}_{\beta}^{j\dag},\mathbf{b}_{\alpha,j})\text{ for }j>2 \}.
\end{align}

\begin{remark} Note that the sum of two rank-one projectors is not a projector. Here $\mathcal{P}$ is the ``correct'' way to add $\mathbf{P}_{1}$
	and $\mathbf{P}_{2}$ so that the result is a rank-two projector that is a sum of two orthogonal rank-one projectors and that has the same row and 
	column spaces as $\mathbf{P}_{1} + \mathbf{P}_{2}$. Also, note that there are many ways to choose bases in the row and column ranges of $\mathcal{P}$, 
	the choices above reflect the splittings $\mathcal{P} = \mathcal{P}_{1} + \mathcal{P}_{2} = \tilde{\mathcal{P}}_{1} + \tilde{\mathcal{P}}_{2}$.	
\end{remark}	

We can now use these projectors to split the discrete Schlesinger evolution equations to define dynamic on eigenvectors. The proof is very similar in 
spirit to the rank-one case proof in \cite{DzhSakTak:2013:DSTTHFADPE}.

\begin{theorem}\label{thm:rank2}
	Consider a  multiplier matrix in the form
	\begin{equation}
		\mathbf{R}(z) = \mathbf{I} + \frac{ z_{\alpha} - z_{\beta} }{ z - z_{\alpha} }\mathcal{P},\quad\text{where }
		\mathcal{P} = \mathcal{P}_{1} + \mathcal{P}_{2} = \tilde{\mathcal{P}}_{1} + \tilde{\mathcal{P}}_{2},
	\end{equation}
	and $\mathcal{P}_{i}$, $\tilde{\mathcal{P}}_{i}$ are given by~(\ref{eq:P-projs}--\ref{eq:Pt-projs}).
	Then $\mathcal{P}$ satisfies the constraints~(\ref{eq:P-constr}) and so defines a Schlesinger transformation. This Schlesinger transformation
	has the type $\left\{\begin{smallmatrix}	\alpha & \beta \\ 1 & 1 \\ 2 & 2	\end{smallmatrix}\right\}$ and the corresponding map
	\begin{equation*}
		(\mathcal{B}\times \mathcal{C})_{\mathbf{\Theta}} \to 
		(\bar{\mathcal{B}}\times \bar{\mathcal{C}})_{\bar{\mathbf{\Theta}}}	
	\end{equation*}
	is given by the following evolution equations, where $c_{i}^{j}$ are again arbitrary non-zero constants.
	\begin{enumerate}[(i)]
		\item Transformation vectors:
		\begin{equation}
			\bar{\mathbf{b}}_{\alpha,1} = \frac{ 1 }{ c_{\alpha}^{1} } \mathbf{Q}_{2}\mathbf{b}_{\beta,1},\,
			\bar{\mathbf{b}}_{\alpha,2} = \frac{ 1 }{ c_{\alpha}^{2} } \mathbf{Q}_{1}\mathbf{b}_{\beta,2},\quad
			\bar{\mathbf{c}}_{\beta}^{1\dag} = c_{\beta}^{1} \mathbf{c}_{\alpha}^{1\dag}\mathbf{Q}_{2},\,
			\bar{\mathbf{c}}_{\beta}^{2\dag} = c_{\beta}^{2} \mathbf{c}_{\alpha}^{2\dag}\mathbf{Q}_{1}.
			\label{eq:bb-cb-generators-rank2}				
		\end{equation}
		\item Generic indices:
		\begin{align}
			\bar{\mathbf{b}}_{i,j} &= 
			\frac{ 1 }{ c_{i}^{j} } \mathbf{R}(z_{i}) \mathbf{b}_{i,j},\,(i\neq \alpha\text{ and if } i= \beta, j>2);
			\label{eq:bb-generic-rank2}\\
			\bar{\mathbf{c}}_{i}^{j\dag} &= c_{i}^{j}  \mathbf{c}_{i}^{j\dag} \mathbf{R}^{-1}(z_{i}),\, (i\neq 
			\beta\text{ and if } i=\alpha, j>2).
			\label{eq:cb-generic-rank2}
		\end{align}
		\item Special indices (here $k=1,2$, $k' = 3- k$, and $j>2$):
		\begin{align}
			\bar{\mathbf{b}}_{\alpha,j} &= \frac{ 1 }{ c_{\alpha}^{j} }\left(\mathbf{I} - 
			 \left(\frac{ \mathcal{P}_{1}}{ \theta_{\alpha}^{1} - \theta_{\alpha}^{j} - 1 } + 
			\frac{ \mathcal{P}_{2}}{ \theta_{\alpha}^{2} - \theta_{\alpha}^{j} - 1 }\right) \left(\sum_{i\neq \alpha} 
			\frac{ z_{\beta} - z_{\alpha}}{ z_{i} - z_{\alpha} }\mathbf{A}_{i}
			\right) \right)\mathbf{b}_{\alpha,j};
			\label{eq:bb-ak-rank2}\\
			\bar{\mathbf{c}}_{\beta}^{j\dag} &= c_{\beta}^{j}\mathbf{c}_{\beta}^{j\dag} \left(
			\mathbf{I} - \left(\sum_{i\neq \beta}
			\frac{ z_{\alpha} - z_{\beta} }{ z_{i} - z_{\beta} } \mathbf{A}_{i}\right)
			\left(\frac{ \tilde{\mathcal{P}}_{1}}{ \theta_{\beta}^{1} - \theta_{\beta}^{j} + 1 } + 
			\frac{ \tilde{\mathcal{P}}_{2}}{ \theta_{\beta}^{2} - \theta_{\beta}^{j} + 1 }\right)\right);
			\label{eq:cb-bk-rank2}\\
			\bar{\mathbf{b}}_{\beta,k} &= \frac{ 1 }{ c_{\beta}^{k}  }
			\left( (\theta_{\beta}^{k} + 1) \mathbf{I} + 
			\phantom{\left(
			\sum_{i\neq \beta} \frac{ z_{\alpha} - z_{\beta}  }{ 
			z_i - z_{\beta}  } \mathbf{A}_{i}\right)} 
			\right. \label{eq:bb-bn-rank2}\\
			&\qquad \left.
			\mathcal{Q}	
			\left(\mathbf{I} +  \sum_{j>2}\frac{ \mathbf{b}_{\beta,j} \mathbf{c}_{\beta}^{j\dag}}{ 
			\theta_{\beta}^{1} - \theta_{\beta}^{j} + 1 }\right) \left(
			\sum_{i\neq \beta} \frac{ z_{\alpha} - z_{\beta}  }{ 
			z_i - z_{\beta}  } \mathbf{A}_{i}\right)\right)
			\frac{ \mathbf{b}_{\beta,k} }{  \mathbf{c}_{\alpha}^{k\dag} \mathbf{Q}_{k'} \mathbf{b}_{\beta,k}  };
			\notag\\
			\bar{\mathbf{c}}_{\alpha}^{k\dag} &= c_{\alpha}^{k} 
			\frac{ \mathbf{c}_{\alpha}^{k\dag} }{ \mathbf{c}_{\alpha}^{k\dag} \mathbf{Q}_{k'}\mathbf{b}_{\beta,k}  }\left(
			 (\theta_{\alpha}^{k} -  1) \mathbf{I} + 
			\phantom{\left(\sum_{i\neq \alpha}\frac{ z_{\beta} - z_{\alpha}
			 }{ z_{i} - z_{\alpha}  } \mathbf{A}_{i}\right)} 
			\right.	 \label{eq:cb-am-rank2}\\
			&\qquad \left. \left(\sum_{i\neq \alpha}\frac{ z_{\beta} - z_{\alpha}
			 }{ z_{i} - z_{\alpha}  } \mathbf{A}_{i}\right)
			\left( \mathbf{I} + \sum_{j>2} 
			\frac{ \mathbf{b}_{\alpha,j} \mathbf{c}_{\alpha}^{j\dag} }{ \theta_{\alpha}^{k} - \theta_{\alpha}^{j} - 1 }
			\right) \mathcal{Q}
			\right).\notag 
		\end{align}
	\end{enumerate}
\end{theorem}	
\begin{proof}
	Of course the statement that $\mathcal{P}$ defines an elementary Schlesinger transformation of the type 
	$\left\{\begin{smallmatrix}	\alpha & \beta \\ 1 & 1 \\ 2 & 2	\end{smallmatrix}\right\}$ follows from how we 
	derived it, but it can also be seen directly. E.g., conditions (\ref{eq:P-projs}--\ref{eq:Pt-projs}) follow immediately from
	\begin{equation}
		\mathcal{P}\mathbf{A}_{\alpha}= (\theta_{\alpha}^{1} \mathcal{P}_{1} + \theta_{\alpha}^{2} \mathcal{P}_{2}),\qquad 
		\mathbf{A}_{\beta}\mathcal{P} = \theta_{\beta}^{1}\tilde{\mathcal{P}}_{1} + \theta_{\beta}^{2}\tilde{\mathcal{P}}_{2}, 
	\end{equation}
	and the fact that 
	\begin{equation*}
		\bar{\theta}_{\alpha}^{i} = \theta_{\alpha}^{i} - 1,\quad \bar{\theta}_{\beta}^{i} = \theta_{\beta}^{i} + 1\quad\text{ for }i=1,2\quad\text{ and }
		\quad\bar{\theta}_{i}^{j} = \theta_{i}^{j}\quad\text{ otherwise}
	\end{equation*}
	can be seen, in particular, from our derivation of the evolution equations below.

To establish (i), we use~(\ref{eq:P-cond}): 
\begin{equation*}
	\bar{\mathbf{A}}_{\alpha}\mathcal{P} = \mathcal{P} \mathbf{A}_{\alpha} - \mathcal{P} = 
	(\theta_{\alpha}^{1} - 1) \mathcal{P}_{1} + (\theta_{\alpha}^{2} - 1)\mathcal{P}_{2}.
\end{equation*}	
Since $\mathcal{P}_{1}\mathbf{Q}_{2}\mathbf{b}_{\beta,1} = \mathbf{Q}_{2}\mathbf{b}_{\beta,1} $ and 
$\mathcal{P}_{2}\mathbf{Q}_{2}\mathbf{b}_{\beta,1} = \mathbf{0}$, we see that $\bar{\theta}_{\alpha}^{1} = \theta_{\alpha}^{1}-1$ and
$\bar{\mathbf{b}}_{\alpha,1}\sim \mathbf{Q}_{2}\mathbf{b}_{\beta,1}$, where $\sim$ stands for `proportional'. Then 
$\bar{\mathbf{b}}_{\alpha,1} =  \mathbf{Q}_{2}\mathbf{b}_{\beta,1}/c_{\alpha}^{1}$, where $c_{\alpha}^{1}$ is some non-zero proportionality 
constant. The other equations in this part are proved similarly. Note that the consequence of (i) is that we can write
\begin{equation}
	\mathcal{P}_{i} = \frac{ \bar{\mathbf{b}}_{\alpha,i} \mathbf{c}_{\alpha}^{i\dag} }{ \mathbf{c}_{\alpha}^{i\dag} \bar{\mathbf{b}}_{\alpha,i}},\,
	\tilde{\mathcal{P}}_{i} = \frac{ \mathbf{b}_{\beta,i} \bar{\mathbf{c}}_{\beta}^{i\dag} }{ \bar{\mathbf{c}}_{\beta}^{i\dag} \mathbf{b}_{\beta,i}},\quad
	\mathcal{P} =  \frac{ \bar{\mathbf{b}}_{\alpha,1} \mathbf{c}_{\alpha}^{1\dag} }{ \mathbf{c}_{\alpha}^{1\dag} \bar{\mathbf{b}}_{\alpha,1}} + 
	 \frac{ \bar{\mathbf{b}}_{\alpha,2} \mathbf{c}_{\alpha}^{2\dag} }{ \mathbf{c}_{\alpha}^{2\dag} \bar{\mathbf{b}}_{\alpha,2}}
	= \frac{ \mathbf{b}_{\beta,1} \bar{\mathbf{c}}_{\beta}^{1\dag} }{ \bar{\mathbf{c}}_{\beta}^{1\dag} \mathbf{b}_{\beta,1}} + 
	\frac{ \mathbf{b}_{\beta,2} \bar{\mathbf{c}}_{\beta}^{2\dag} }{ \bar{\mathbf{c}}_{\beta}^{2\dag} \mathbf{b}_{\beta,2}}.\label{eq:P-mixed}
\end{equation}

For the generic case $i\neq \alpha,\beta$ in (ii) the proof is identical to the rank-one case. Since it is also short, we opted to include it 
to make the paper more self-contained. From~(\ref{eq:Ai}) we see that 
\begin{equation*}
	\bar{\mathbf{A}}_{i}\mathbf{R}(z_{i})\mathbf{B}_{i} = \mathbf{R}(z_{i})\mathbf{A}_{i}\mathbf{B}_{i} = \mathbf{R}(z_{i})\mathbf{B}_{i}\mathbf{\Theta}_{i},
\end{equation*}
and so $\bar{\mathbf{\Theta}}_{i} = \mathbf{\Theta}_{i}$ and  
$\bar{\mathbf{B}}_{i} \bar{\mathbf{D}}_{i} = \mathbf{R}(z_{i})\mathbf{B}_{i}$, where $\bar{\mathbf{D}}_{i} = \operatorname{diag}\{c_{i}^{j}\}$ is a diagonal 
matrix of non-zero proportionality constants. Similarly, $\bar{\mathbf{\Delta}}_{i}\bar{\mathbf{C}}_{i}^{\dag} = \mathbf{C}_{i}^{\dag}\mathbf{R}^{-1}(z_{i})$.
The orthogonality condition $\bar{\mathbf{C}}_{i}^{\dag}\bar{\mathbf{B}}_{i} = \mathbf{\Theta}_{i}$ implies that 
$\bar{\mathbf{\Delta}}_{i} \bar{\mathbf{D}}_{i} = \mathbf{I}$, which gives (ii) for generic indices. For $i=\alpha$, from~(\ref{eq:Q-cond}) we see that 
$\mathbf{C}_{\alpha}^{\dag}\mathcal{Q}\bar{\mathbf{A}}_{\alpha} = \mathbf{\Theta}_{\alpha} \mathbf{C}_{\alpha}^{\dag} \mathcal{Q}$, and so again 
$\bar{\mathbf{\Delta}}_{\alpha}\bar{\mathbf{C}}_{\alpha}^{\dag} = \mathbf{C}_{\alpha}^{\dag}\mathcal{Q}$. However, since 
$\mathbf{c}_{\alpha}^{1\dag}\mathcal{Q} = \mathbf{c}_{\alpha}^{2\dag}\mathcal{Q} = \mathbf{0}$, 
$(\bar{\mathbf{\Delta}}_{\alpha})_{1}^{1} = (\bar{\mathbf{\Delta}}_{\alpha})_{2}^{2} = 0$ and we can not recover $\bar{\mathbf{c}}_{\alpha}^{1\dag}$
and $\bar{\mathbf{c}}_{\alpha}^{2\dag}$. The case $i= \beta$ is similar.

Finally, let us consider special indices. To find $\bar{\mathbf{b}}_{\alpha,j}$ for $j>2$ start with~(\ref{eq:Aalpha}) and~(\ref{eq:cb-generic-rank2}):
\begin{equation*}
	\bar{\mathbf{A}}_{\alpha} = \bar{\mathbf{b}}_{\alpha,1}\bar{\mathbf{c}}_{\alpha}^{1} + \bar{\mathbf{b}}_{\alpha,2}\bar{\mathbf{c}}_{\alpha}^{2}
	+ \sum_{j>2}\bar{\mathbf{b}}_{\alpha,j}(c_{\alpha}^{j}\mathcal{Q}\mathbf{c}_{\alpha}^{j}) = \mathbf{A}_{\alpha} - \mathcal{Q}\mathbf{A}_{\alpha}\mathcal{P}
	+ \sum_{i\neq \alpha} \left(\frac{ z_{\beta} - z_{\alpha} }{ z_{i} - z_{\alpha} }\right)\mathcal{P} \mathbf{A}_{i}\mathcal{Q}.
\end{equation*}
Multiplying on the right by $\mathbf{b}_{\alpha,j}$, using the orthogonality conditions and $\mathcal{P}\mathbf{b}_{\alpha,j} = \mathbf{0}$,
$\mathcal{Q}\mathbf{b}_{\alpha,j} = \mathbf{b}_{\alpha,j}$, we get 
\begin{equation}
	\bar{\mathbf{b}}_{\alpha,1}(\bar{\mathbf{c}}_{\alpha}^{1\dag} \mathbf{b}_{\alpha,j}) + 
	\bar{\mathbf{b}}_{\alpha,2}(\bar{\mathbf{c}}_{\alpha}^{2\dag} \mathbf{b}_{\alpha,j}) + c_{\alpha}^{j} \theta_{\alpha}^{j} \bar{\mathbf{b}}_{\alpha,j}
	= \theta_{\alpha}^{j}\mathbf{b}_{\alpha,j} + 
	\sum_{i\neq \alpha} \left(\frac{ z_{\beta} - z_{\alpha} }{ z_{i} - z_{\alpha} }\right)\mathcal{P} \mathbf{A}_{i}\mathbf{b}_{\alpha,j}.
	\label{eq:baj-rank2}
\end{equation}
Now left-multiply by $\bar{\mathbf{c}}_{\alpha}^{1\dag}$ and use expression~(\ref{eq:P-mixed}) for $\mathcal{P}$ and orthogonality conditions 
again to get
\begin{equation*}
	\bar{\theta}_{\alpha}^{1}(\bar{\mathbf{c}}_{\alpha}^{1\dag} \mathbf{b}_{\alpha,j}) = 
	\theta_{\alpha}^{j}(\bar{\mathbf{c}}_{\alpha}^{1\dag} \mathbf{b}_{\alpha,j}) + \bar{\theta}_{\alpha}^{1} 
	\frac{ \mathbf{c}_{\alpha}^{1\dag} }{ \mathbf{c}_{\alpha}^{1\dag}\bar{\mathbf{b}}_{\alpha,1} }
	\sum_{i\neq \alpha} \left(\frac{ z_{\beta} - z_{\alpha} }{ z_{i} - z_{\alpha} }\right) \mathbf{A}_{i}\mathbf{b}_{\alpha,j}.
\end{equation*}
This gives 
\begin{align*}
	(\bar{\mathbf{c}}_{\alpha}^{1\dag} \mathbf{b}_{\alpha,j}) &= \frac{ \bar{\theta}_{\alpha}^{1} }{ \bar{\theta}_{\alpha}^{1} - \theta_{\alpha}^{j} }
	\frac{ \mathbf{c}_{\alpha}^{1\dag} }{ \mathbf{c}_{\alpha}^{1\dag}\bar{\mathbf{b}}_{\alpha,1} }
	\sum_{i\neq \alpha} \left(\frac{ z_{\beta} - z_{\alpha} }{ z_{i} - z_{\alpha} }\right) \mathbf{A}_{i}\mathbf{b}_{\alpha,j},\\
	\bar{\mathbf{b}}_{\alpha,1}(\bar{\mathbf{c}}_{\alpha}^{1\dag} \mathbf{b}_{\alpha,j}) &= 
	\frac{ \bar{\theta}_{\alpha}^{1} }{ \bar{\theta}_{\alpha}^{1} - \theta_{\alpha}^{j} }
	\mathcal{P}_{1}
	\sum_{i\neq \alpha} \left(\frac{ z_{\beta} - z_{\alpha} }{ z_{i} - z_{\alpha} }\right) \mathbf{A}_{i}\mathbf{b}_{\alpha,j}.
\end{align*}
Repeating the same steps for $\bar{\mathbf{b}}_{\alpha,2}(\bar{\mathbf{c}}_{\alpha}^{2\dag} \mathbf{b}_{\alpha,j})$, substituting the result
into~(\ref{eq:baj-rank2}), solving for $\bar{\mathbf{b}}_{\alpha,j}$ and simplifying gives~(\ref{eq:bb-ak-rank2});~(\ref{eq:cb-bk-rank2}) is
proved in a similar fashion.

Finally, to get expressions for $\bar{\mathbf{b}}_{\beta,1}$ and $\bar{\mathbf{b}}_{\beta,2}$, use all of the previously obtained expressions to write
\begin{align*}
	\bar{\mathbf{A}}_{\beta} &= \bar{\mathbf{b}}_{\beta,1}\bar{\mathbf{c}}_{\beta}^{1\dag} + \bar{\mathbf{b}}_{\beta,2}\bar{\mathbf{c}}_{\beta}^{2\dag}
	+ \sum_{j>2} \bar{\mathbf{b}}_{\beta,j}\bar{\mathbf{c}}_{\beta}^{j\dag}\\
	&= c_{\beta}^{1} \bar{\mathbf{b}}_{\beta,1}\mathbf{c}_{\alpha}^{1\dag}\mathbf{Q}_{2} + 
	c_{\beta}^{2} \bar{\mathbf{b}}_{\beta,2}\mathbf{c}_{\alpha}^{2\dag}\mathbf{Q}_{1}\\ 
	&\quad + 
	\sum_{j>2}\mathcal{Q} \mathbf{b}_{\beta,j} \mathbf{c}_{\beta}^{j\dag}\left(
	\mathbf{I} - \left(\sum_{i\neq \beta}
	\frac{ z_{\alpha} - z_{\beta} }{ z_{i} - z_{\beta} } \mathbf{A}_{i}\right)
	\left(\frac{ \tilde{\mathcal{P}}_{1}}{ \theta_{\beta}^{1} - \theta_{\beta}^{j} + 1 } + 
	\frac{ \tilde{\mathcal{P}}_{2}}{ \theta_{\beta}^{2} - \theta_{\beta}^{j} + 1 }\right)\right)\\
	\intertext{which, in view of~(\ref{eq:Abeta}), also can be written as}
	&= \mathbf{A}_{\beta} - \mathbf{P}\mathbf{A}_{\beta}\mathcal{Q} + \mathcal{P}
	+ \sum_{i\neq \beta} \left(\frac{ z_{\alpha} - z_{\beta} }{ z_{i} - z_{\beta} }\right)\mathcal{Q}\mathbf{A}_{i}\mathcal{P}.
\end{align*}
Multiplying on the right by $\mathbf{b}_{\beta,1}$ we get
\begin{align*}
	\bar{\mathbf{A}}_{\beta}\mathbf{b}_{\beta,1} &= c_{\beta}^{1} (\mathbf{c}_{\alpha}^{1\dag}\mathbf{Q}_{2}\mathbf{b}_{\beta,1})\bar{\mathbf{b}}_{\beta,1}
	+ \sum_{j>2} \mathcal{Q} \mathbf{b}_{\beta,j} \mathbf{c}_{\beta}^{j\dag} \left(\mathbf{I} - \sum_{i\neq \beta}\frac{ z_{\alpha} - z_{\beta}  }{ 
	z_i - z_{\beta}  } \frac{ \mathbf{A}_{i} }{ \theta_{\beta}^{1} - \theta_{\beta}^{j}+ 1 }\right)\mathbf{b}_{\beta,1}\\
	&= c_{\beta}^{1} (\mathbf{c}_{\alpha}^{1\dag}\mathbf{Q}_{2}\mathbf{b}_{\beta,1})\bar{\mathbf{b}}_{\beta,1}
	- \mathcal{Q}\left(\sum_{i\neq \beta}\frac{ z_{\alpha} - z_{\beta}  }{ 
	z_i - z_{\beta}  } \mathbf{A}_{i}\right)\left(\sum_{j>2} 
	\frac{ \mathbf{b}_{\beta,j} \mathbf{c}_{\beta}^{j\dag} }{ \theta_{\beta}^{1} - \theta_{\beta}^{j}+ 1 }\right) \mathbf{b}_{\beta,1}\\
	&= \theta_{\beta}^{1} \mathbf{b}_{\beta,1} + \mathbf{b}_{\beta,1} + \mathcal{Q}\left(\sum_{i\neq \beta}\frac{ z_{\alpha} - z_{\beta}  }{ 
	z_i - z_{\beta}  } \mathbf{A}_{i}\right)\mathbf{b}_{\beta,1}.
\end{align*}
Solving for $\bar{\mathbf{b}}_{\beta,1}$ gives~(\ref{eq:bb-bn-rank2}) for $k=1$, and the expression for $\mathbf{b}_{\beta,2}$ is obtained by right-multiplying
by $\mathbf{b}_{\beta,2}$ instead. Equations~(\ref{eq:cb-am-rank2}) are obtained along the same lines. 
\end{proof}


\section{Reductions from Schlesinger Transformations to Difference Painlev\'e Equations} 
\label{sec:reductions}

In this section, which is the central section of the paper, we consider two examples of 
reductions from the Schlesinger dynamic on the decomposition space to difference 
Painlev\'e equations. First we consider Schlesinger transformations of a Fuchsian
system of spectral type $111,111,111$. Resulting difference Painlev\'e equation is of type 
d-$P({A}_{2}^{(1)*})$ and has the symmetry group ${E}^{(1)}_{6}$. We have previously considered 
this example in \cite{DzhSakTak:2013:DSTTHFADPE}, but the exposition there was very brief and 
it relied on a nontrivial observation on how to choose good coordinates parameterizing our 
Fuchsian system. Here we not only provide more details but also show how geometric considerations 
\emph{lead us} to the appropriate coordinate choice. In the second example we consider Schlesinger 
transformations of a Fuchsian system of spectral type $22,1111,1111$, which gives 
difference Painlev\'e equation of type d-$P({A}_{1}^{(1)*})$ with the symmetry group ${E}_{7}^{(1)}$. 
This example is completely new and here, in addition to elementary Schlesinger transformations of rank one,
we also, for the first time, consider elementary Schlesinger 
transformations of rank two --- we need such transformations to represent the standard example of 
a difference Painlev\'e equation of type d-$P({A}_{1}^{(1)*})$, as written in \cite{GraRamOht:2003:AUDOTAQVADIEATST}, 
\cite{Sak:2007:PDPEATLF}, as a composition of elementary Schlesinger transformations.

\subsection{Reduction to difference Painlev\'e equation of type d-$P({A}_{2}^{(1)*})$ with the
symmetry group ${E}^{(1)}_{6}$.} 
\label{sub:reduction_to_difference_painlev'e_equation_of_type_d_p_a__2_with_the_symmetry_group_e_6_}

\subsubsection{Model Example} 
\label{ssub:model_exampleA2}
For our model example of type d-$P({A}_{2}^{(1)*})$ we take the equation that was first written by 
Grammaticos, Ramani, and Ohta, \cite{GraRamOht:2003:AUDOTAQVADIEATST}, see also
Murata \cite{Mur:2004:NEFDPE} and Sakai \cite{Sak:2007:PDPEATLF}. Following Sakai's geometric approach, we view this equation as a
birational map $\varphi: \mathbb{P}^{1}\times \mathbb{P}^{1} \dashrightarrow \mathbb{P}^{1}\times \mathbb{P}^{1}$ 
with parameters $b_{1},\dots, b_{8}$ 
\begin{equation}
	\left(\begin{matrix}
		b_{1} & b_{2} & b_{3} & b_{4}\\
		b_{5} & b_{6} & b_{7} & b_{8}
	\end{matrix}; f,g\right) \mapsto 
	\left(\begin{matrix}
		\bar{b}_{1} & \bar{b}_{2} & \bar{b}_{3} & \bar{b}_{4}\\
		\bar{b}_{5} & \bar{b}_{6} & \bar{b}_{7} & \bar{b}_{8}
	\end{matrix}; \bar{f},\bar{g}\right),
\end{equation} 
where $\bar{b}_{1} = b_{1}$, $\bar{b}_{2} = b_{2}$, $\bar{b}_{3} = b_{3}$, $\bar{b}_{4} = b_{4}$,
$\bar{b}_{5} = b_{5} + \delta$, $\bar{b}_{6} = b_{6} + \delta$, $\bar{b}_{7} = b_{7} - \delta$,
$\bar{b}_{8} = b_{8} - \delta$, 
$\delta = b_{1} + \cdots + b_{8}$, and   $\bar{f}$ and $\bar{g}$ are given by the equation
\begin{equation}
	\left\{
	\begin{aligned}
		(f + g)(\bar{f}+g) & =\frac{(g+b_1)(g+b_2)(g+b_3)(g+b_4)}{(g-b_5)(g-b_6)}\\
		(\bar{f}+g)(\bar{f}+\bar{g})& =\frac{(\bar{f}-b_1)(\bar{f}-b_2)(\bar{f}-b_3)(\bar{f}-b_4)}{(\bar{f}+b_7-\delta)(\bar{f}+b_8-\delta)}
	\end{aligned}
	\right..\label{eq:dpA2-st}	
\end{equation}
This map has the following eight indeterminate points:
\begin{alignat*}{4}
	&p_{1}(b_{1},-b_{1}),&\quad	&p_{3}(b_{3},-b_{3}),&\quad	&p_{5}(\infty,b_{5}),&\quad	p_{7}(-b_{7},\infty),\\
	&p_{2}(b_{2},-b_{2}),&\quad	&p_{4}(b_{4},-b_{4}),&\quad	&p_{6}(\infty,b_{6}),&\quad	p_{8}(-b_{8},\infty),
\end{alignat*}
resolving which by the blow-up procedure then gives us a rational surface $\mathcal{X}_{\mathbf{b}}$, known as the 
\emph{Okamoto space of initial conditions} for this difference Painlev\'e equation, that is
described by the blow-up diagram on Figure~\ref{fig:dpa2-standard-blowup}.

\begin{figure}[h]
	\centering
		\includegraphics{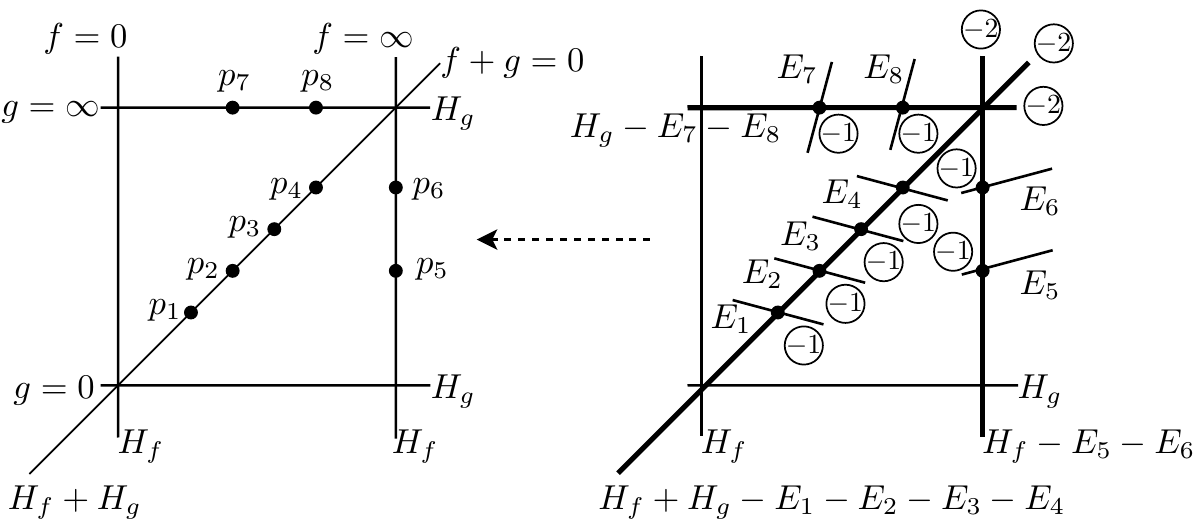}
	\caption{Okamoto surface $\mathcal{X}_{\mathbf{b}}$ for the model form of d-$P({A}_{2}^{(1)*})$.}
	\label{fig:dpa2-standard-blowup}
\end{figure}	

The Picard lattice of $\mathcal{X}_{\mathbf{b}}$ is generated by the total transforms 
$H_{f}$ and $H_{g}$ of the coordinate lines and the classes of the exceptional divisors $E_{i}$,
\begin{equation*}
	\operatorname{Pic}(\mathcal{X}) = \mathbb{Z} H_{f} \oplus \mathbb{Z} H_{g} \oplus \bigoplus_{i=1}^{8} \mathbb{Z}E_{i}.
\end{equation*}
The anti-canonical divisor $-K_{\mathcal{X}}=2 H_{f} + 2 H_{g} - \sum_{i=1}^{8} E_{i}$ uniquely decomposes as 
a positive linear combination of $-2$-curves $D_{i}$, $-K_{\mathcal{X}} = D_{0} + D_{1} + D_{2}$,
where the irreducible components $D_{i}$, in bold on Figure~\ref{fig:dpa2-standard-blowup}, are given by
\begin{equation*}
	D_{0} = H_{f} + H_{g} - E_{1} - E_{2} - E_{3} - E_{4},\quad D_{1} = H_{f} - E_{5} - E_{6}, \quad D_{2} = H_{g} - E_{7} - E_{8}. 
\end{equation*}
The configuration of components $D_{i}$ is described by the 
Dynkin diagram of type ${A}^{(1)}_{2}$ (with nodes corresponding to classes of self-intersection $-2$ and edges connecting
classes of intersection index $1$). To this diagram
correspond two different types of surfaces, the generic one corresponds to the multiplicative
system of type ${A}^{(1)}_{2}$, and the degenerate configuration, where all three components $D_{i}$ 
intersect at one point,
corresponds to the additive system denoted by ${A}_{2}^{(1)*}$, which is clearly our case, see Figure~\ref{fig:dpa2-configs}.

\begin{figure}[h]
	\begin{center}
		\begin{tabular}{ccc}
			\includegraphics{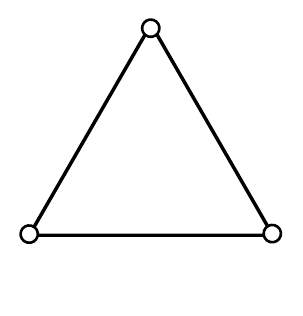} & \quad
			\includegraphics{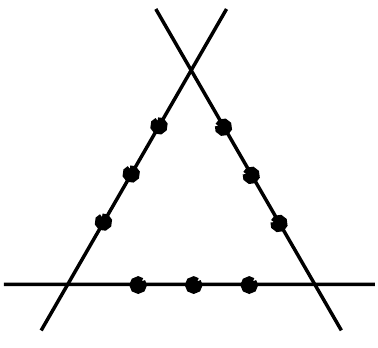} \quad & 
			\includegraphics{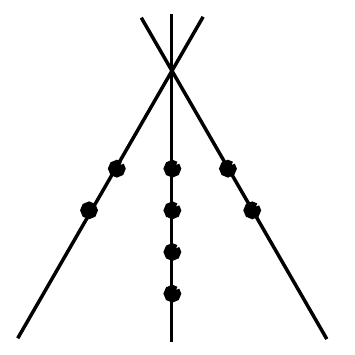} \\
			Dynkin diagram ${A}_{2}^{(1)}$ & ${A}^{(1)}_{2}$-surface & ${A}_{2}^{(1)*}$-surface.
		\end{tabular}
	\end{center}
	\caption{Configurations of type ${A}_{2}^{(1)}$}
	\label{fig:dpa2-configs}
\end{figure}

Components $D_{i}$ of $-K_{\mathcal{X}}$ span the sub-lattice $R = \operatorname{Span}_{\mathbb{Z}}\{D_{1}, D_{2}, D_{3}\}$, and its
orthogonal complement $R^{\perp}$ is called the \emph{symmetry sub-lattice}. In our case, it is easy to see that 
$R^{\perp}=\operatorname{Span}_{\mathbb{Z}}\{\alpha_{0},\dots,\alpha_{6}\}$ is of type ${E}^{(1)}_{6}$,
see Figure~\ref{fig:dpa2-symm}. 

\begin{figure}[h]
\begin{equation*}
	\begin{aligned}
		\alpha_{0}&= E_{3} - E_{4},& \quad \alpha_{1}&= E_{2} - E_{3}, \\  
		\alpha_{2}&= E_{1} - E_{2},& \quad  \alpha_{3}&= H_{f} - E_{1} - E_{7},\\
		\alpha_{4}&= E_{7} - E_{8},& \quad  \alpha_{5}&= H_{g} - E_{1} - E_{5}, \\
		\alpha_{6}&= E_{5} - E_{6}
	\end{aligned}\qquad\qquad
		\raisebox{-0.5in}{\includegraphics{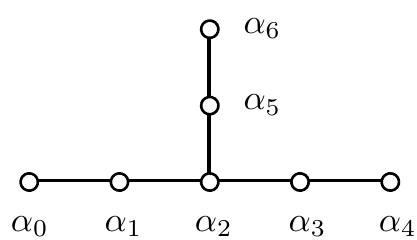}}
\end{equation*}
	\caption{Symmetry sub-lattice for d-$P(\tilde{A}_{2}^{*})$}
	\label{fig:dpa2-symm}
\end{figure}	

Finally, we compute the action of $\varphi_{*}$ on $\operatorname{Pic}(\mathcal{X})$ to be
\begin{align}
	 H_{f}&\mapsto 6 H_{f} + 3 H_{g} -2 E_{1} - 2 E_{2} - 2 E_{3} - 2 E_{4} - E_{5} - E_{6} - 3 E_{7} - 3 E_{8},\label{eq:dpa-tr-std}\\
	 H_{g}&\mapsto 3 H_{f} + H_{g}  - E_{1} - E_{2} - E_{3} - E_{4} - E_{7} - E_{8},\notag\\
	 E_{1}&\mapsto 2 H_{f} + H_{g} - E_{2} - E_{3} - E_{4} - E_{7} - E_{8},\notag\\
	 E_{2}&\mapsto 2 H_{f} + H_{g} - E_{1} - E_{3} - E_{4} - E_{7} - E_{8},\notag\\
	 E_{3}&\mapsto 2 H_{f} + H_{g} - E_{1} - E_{2} - E_{4} - E_{7} - E_{8},\notag\\
	 E_{4}&\mapsto 2 H_{f} + H_{g} - E_{1} - E_{2} - E_{3} - E_{7} - E_{8},\notag\\
	 E_{5}&\mapsto 3 H_{f} + H_{g} - E_{1} - E_{2} - E_{3} - E_{4} - E_{6} - E_{7} - E_{8},\notag\\
	 E_{6}&\mapsto 3 H_{f} + H_{g} - E_{1} - E_{2} - E_{3} - E_{4} - E_{5} - E_{7} - E_{8},\notag\\
	 E_{7}&\mapsto H_{f} - E_{8},\notag\\
	 E_{8}&\mapsto H_{f} - E_{7},\notag
\end{align}
and so the induced action $\varphi_{*}$ on the sub-lattice $R^{\perp}$	is given by the following \emph{translation}:
\begin{equation*}
	(\alpha_{0}, \alpha_{1}, \alpha_{2}, \alpha_{3}, \alpha_{4}, \alpha_{5}, \alpha_{6})\mapsto	
	(\alpha_{0}, \alpha_{1}, \alpha_{2}, \alpha_{3}, \alpha_{4}, \alpha_{5}, \alpha_{6}) + 
	(0,0,0,1,0,-1,0)(-K_{\mathcal{X}}),\label{eq:dpa-trans-std}
\end{equation*}
as well as the permutation $(D_{0}D_{1}D_{2})$ of the irreducible components of $-K_{\mathcal{X}}$. 
We now want to compare this standard picture with the one that is obtained from Schlesinger 
transformations.


\subsubsection{Schelsinger Transformations} 
\label{ssub:schelsinger_transformationsA2}
Consider a $3\times3$ Fuchsian system of the spectral type $111,111,111$. This system has three poles and it is 
convenient to assume that one of them is at $z_{3}=\infty$, since our elementary Schlesinger 
transformations preserve $\mathbf{A}_{\infty}$. Also, in view of scalar gauge transformations we can assume that
$\operatorname{rank}(\mathbf{A}_{i}) = 2$ at finite poles (and using 
M\"obius transformation preserving $z=\infty$, we can in principle map those poles to $z_{1} = 0$ and $z_{2}= 1$). Thus,
\begin{equation*}
	\mathbf{A}_{i} = \mathbf{B}_{i} \mathbf{C}_{i}^{\dag} = 
	\begin{bmatrix}
	\mathbf{b}_{i,1} & \mathbf{b}_{i,2}	
	\end{bmatrix} \begin{bmatrix}
		\mathbf{c}_{i}^{1\dag}\\[2pt] \mathbf{c}_{i}^{2\dag}
	\end{bmatrix},\qquad i=1,2.
\end{equation*}

So the \emph{Riemann scheme} and the \emph{Fuchs relation} for our system are
\begin{equation*}
	\left\{
	\begin{tabular}{cccc}
		$z_{1}$ 		& $	z_{2}$		& 	$z_{3}$	\\
		$\theta_{1}^{1}$	& $\theta_{2}^{1}$	& 	$\theta_{3}^{1}$ \\
		$\theta_{1}^{2}$	& $\theta_{2}^{2}$	& 	$\theta_{3}^{2}$ \\
		$ 0 $			& 	$ 0$		& 	$\theta_{3}^{3}$
	\end{tabular}
	\right\},\qquad 
	\theta_{1}^{1} + \theta_{1}^{2} + \theta_{2}^{1} + \theta_{2}^{2}  + \sum_{j=1}^{3} \theta_{3}^{j}= 0.
\end{equation*}
This example does not have 
any continuous deformation parameters but it admits non-trivial Schlesinger transformation. Consider
an elementary Schlesinger transformation
$\left\{\begin{smallmatrix} 1&2\\1&1\end{smallmatrix}\right\}$ that changes 
$\bar{\theta}_{1}^{1} = \theta_{1}^{1} - 1$, $\bar{\theta}_{2}^{1} = \theta_{2}^{1} + 1$, and 
fixes the remaining characteristic indices. The projector matrices for this map are
\begin{equation*}
	\mathbf{P} = \frac{ \mathbf{b}_{2,1}\mathbf{c}_{1}^{1\dag} }{ \mathbf{c}_{1}^{1\dag} \mathbf{b}_{2,1} },\qquad
	\mathbf{Q} = \mathbf{I} - \mathbf{P},
\end{equation*}
and the evolution equations~(\ref{eq:bb-cb-generators}--\ref{eq:cb-am}) take the form
\begin{alignat*}{2}
	\bar{\mathbf{b}}_{1,1} &= \frac{ 1 }{ c_{1}^{1} } \mathbf{b}_{2,1}, &\quad 
	\bar{\mathbf{b}}_{1,2} &= \frac{ 1 }{ c_{1}^{2} }
	\left(\mathbf{I} - \frac{ \mathbf{P}\mathbf{A}_{2} }{ \theta_{1}^{1} - \theta_{1}^{2} - 1}\right) \mathbf{b}_{1,2},\\
	\bar{\mathbf{b}}_{2,2} &= \frac{ 1 }{ c_{2}^{2} }\mathbf{Q}\mathbf{b}_{2,2}, &\quad 
	\bar{\mathbf{b}}_{2,1}&= \frac{ 1 }{ c_{2}^{1} } \left( (\theta_{2}^{1} + 1) \mathbf{I}
	 + \mathbf{Q}\left(\mathbf{I} + \frac{ \mathbf{b}_{2,2}\mathbf{c}_{2}^{2\dag} }{ \theta_{2}^{1} - \theta_{2}^{2} + 1 }\right)\mathbf{A}_{1}
	\right) \frac{ \mathbf{b}_{2,1} }{ \mathbf{c}_{1}^{1\dag}\mathbf{b}_{2,1} }, \\
	\bar{\mathbf{c}}_{1}^{2\dag} &= c_{1}^{2} \mathbf{c}_{1}^{2\dag}\mathbf{Q}, &\quad
	\bar{\mathbf{c}}_{1}^{1\dag} &= c_{1}^{1} \frac{ \mathbf{c}_{1}^{1\dag} }{ \mathbf{c}_{1}^{1\dag}\mathbf{b}_{2,1} }
	\left((\theta_{1}^{1} - 1)\mathbf{I} + \mathbf{A}_{2} \left(\mathbf{I} + 
	\frac{ \mathbf{b}_{1,2}\mathbf{c}_{1}^{2\dag} }{ \theta_{1}^{1} - \theta_{1}^{2} - 1}\right)\mathbf{Q}\right), \\
	\mathbf{c}_{2}^{1\dag} &= c_{2}^{1}\mathbf{c}_{1}^{1\dag}, &\quad 
	\mathbf{c}_{2}^{2\dag} &= c_{2}^{2} \mathbf{c}_{2}^{2\dag}\left(\mathbf{I} - \frac{ \mathbf{A}_{1}\mathbf{P} }{ \theta_{2}^{1}- \theta_{2}^{2} + 1 }\right),
\end{alignat*}
where $c_{i}^{j}$ are arbitrary non-zero constants (corresponding to trivial gauge transformations).

We now explicitly show that the space of accessory parameters for Fuchsian systems of this type
is two-dimensional by using various gauge transformations to put vectors $\mathbf{b}_{i,j}$
and $\mathbf{c}_{i}^{j\dag}$ in some normal form, and then introduce a coordinate system on this 
phase space. First, assuming that we are in a generic situation, we use a global similarity 
transformation to map the vectors $\mathbf{b}_{1,1}$, $\mathbf{b}_{1,2}$, and $\mathbf{b}_{2,1}$ 
to the standard basis, and then use trivial gauge transformations (i.e., choose appropriate constants $c_{i}^{j}$) 
to make all components of $\mathbf{b}_{2,2}$ equal to $1$. Then the orthogonality conditions 
$\mathbf{C}_{i}^{\dag} \mathbf{B}_{i} = \mathbf{\Theta}_{i}$ give us the following 
parameterization:
\begin{equation*}
	\mathbf{B}_{1} = \begin{bmatrix}
		1 & 0 \\ 0 & 1 \\ 0 & 0
	\end{bmatrix},\quad \mathbf{C}_{1}^{\dag} = \begin{bmatrix}
		\theta_{1}^{1} & 0 & \alpha \\ 0 & \theta_{1}^{2} & \beta
	\end{bmatrix},\quad
	\mathbf{B}_{2} = \begin{bmatrix}
		0 & 1 \\ 0 & 1 \\ 1 & 1
	\end{bmatrix}, \quad
	\mathbf{C}_{2}^{\dag} = \begin{bmatrix}
		-x - \theta_{2}^{1} & x & \theta_{2}^{1} \\ \theta_{2}^{2}-y & y  & 0
	\end{bmatrix}.
\end{equation*}
Here we choose $x$ and $y$ as our coordinates, and we can express $\alpha = \alpha(x,y)$
and $\beta = \beta(x,y)$ from the condition that the eigenvalues of 
$\mathbf{A}_{\infty} = - \mathbf{B}_{1} \mathbf{C}_{1}^{\dag} - \mathbf{B}_{2} \mathbf{C}_{2}^{\dag}$
are $\kappa_{1}$, $\kappa_{2}$, and $\kappa_{3}$ (the resulting 
expressions, although easy to obtain, are quite large and we omit them).
We then get the following dynamic in the coordinates $(x,y)$: 
\begin{equation*}
	\left\{
	\begin{aligned}
		\bar{x} &= \frac{ \alpha - \beta }{ \alpha(\theta_{1}^{2} - \theta_{1}^{1} + 1) }\left( \alpha (x + y) + \theta_{1}^{1} y\right)\\
		\bar{y} &= \frac{ \alpha - \beta }{ \alpha(\theta_{1}^{2} - \theta_{1}^{1} + 1) }
		\left( \frac{ \alpha(\alpha(x + y) + y(\theta_{1}^{2} + 1))(\theta_{1}^{1} - \theta_{2}^{2} + 1) }{ 
		\alpha(\theta_{2}^{1} + 1) - (\alpha - \beta)y } - \alpha (x + y) - \theta_{1}^{1} y\right)\label{eq:dpa-21*-sch}
	\end{aligned}
	\right.\,,		
\end{equation*}
where we still need to substitute $\alpha=\alpha(x,y)$ and $\beta = \beta(x,y)$. So this map is quite complicated and it reflects the fact that 
our choice of the coordinates was rather arbitrary. To better understand the map we again go back to geometry.

The indeterminate points of the map $\psi: (x,y)\to (\bar{x},\bar{y})$ are 
\begin{alignat*}{2}
	& p_{1}\left(\frac{ (\theta_{1}^{1} + \theta_{2}^{1} + \theta_{3}^{1})(\theta_{1}^{2} + \theta_{3}^{1}) }{ \theta_{1}^{1} - \theta_{1}^{2} },
	-\frac{ (\theta_{1}^{1} + \theta_{2}^{2} + \theta_{3}^{1})(\theta_{1}^{2} + \theta_{3}^{1}) }{ \theta_{1}^{1} - \theta_{1}^{2} }\right),&\quad 
	&p_{4}(0,0),\\
	&p_{2}\left(\frac{ (\theta_{1}^{1} + \theta_{2}^{1} + \theta_{3}^{2})(\theta_{1}^{2} + \theta_{3}^{2}) }{ \theta_{1}^{1} - \theta_{1}^{2} },
	-\frac{ (\theta_{1}^{1} + \theta_{2}^{2} + \theta_{3}^{2})(\theta_{1}^{2} + \theta_{3}^{2}) }{ \theta_{1}^{1} - \theta_{1}^{2} }\right),&\quad 
	&p_{5}(-\theta_{2}^{1},\theta_{2}^{2}),\\
	&p_{3}\left(\frac{ (\theta_{1}^{1} + \theta_{2}^{1} + \theta_{3}^{3})(\theta_{1}^{2} + \theta_{3}^{3}) }{ \theta_{1}^{1} - \theta_{1}^{2} },
	-\frac{ (\theta_{1}^{1} + \theta_{2}^{2} + \theta_{3}^{3})(\theta_{1}^{2} + \theta_{3}^{3}) }{ \theta_{1}^{1} - \theta_{2}^{2} }\right),
\end{alignat*}
as well as the sequence of infinitely close points
\begin{align*}
	p_{6}\left(\frac{ 1 }{ x } = 0,\frac{ 1 }{ y } = 0\right)&\longleftarrow p_{7}\left(\frac{ 1 }{ x } = 0,\frac{ x }{ y } = -1\right)\\
	&\longleftarrow p_{8}\left(\frac{ 1 }{ x } = 0,\frac{ x }{ y }=-1,
	\frac{ x(x+y) }{ y } = \frac{ (\theta_{1}^{2} +1) (\theta_{2}^{1} - \theta_{2}^{2})}{ \theta_{1}^{2} - \theta_{1}^{1} }\right).
\end{align*}
Note also that the points $p_{1},\dots,p_{6}$ (and, after blowing up, the point $p_{7}$ as well) all lie on a $(2,2)$-curve $Q$ given by the equation
\begin{equation}
	(\theta_{1}^{1} - \theta_{1}^{2})(x + y)(x + y + \theta_{2}^{1} - \theta_{2}^{2}) + 
	(\theta_{2}^{1} - \theta_{2}^{2})(\theta_{2}^{2} x + \theta_{2}^{1}y) = 0.
\end{equation}

\begin{figure}[h]
	\centering
		\includegraphics{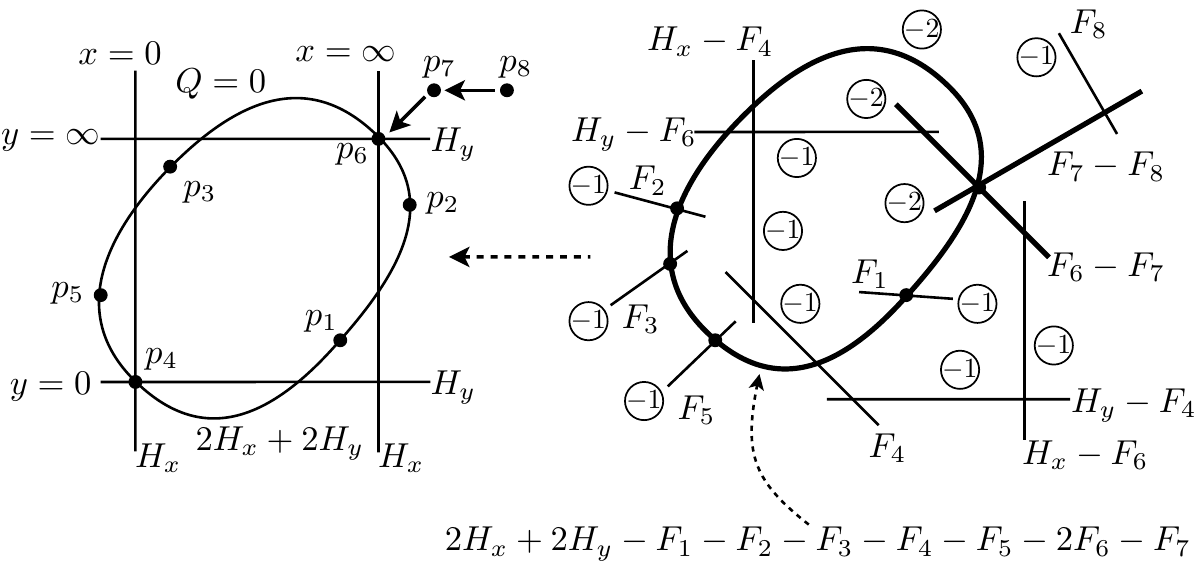}
	\caption{Okamoto surface $\mathcal{X}_{\mathbf{\theta}}$ for the Schlesinger transformations reduction to d-$P({A}_{2}^{(1)*})$ 
	using $\mathbb{P}^{1}\times \mathbb{P}^{1}$ compactification of $\mathbb{C}^{2}$.}
		\label{fig:dpa2-schlesinger-blowup-p1p1}
\end{figure}

Resolving indeterminate points of this map using blow-ups gives us the Okamoto surface $\mathcal{X}_{\mathbf{\theta}}$ 
pictured on Figure~\ref{fig:dpa2-schlesinger-blowup-p1p1}. We can immediately see that in this case
\begin{equation*}
	-K_{\mathcal{X}_{\mathbf{\theta}}} = (2 H_{x} + 2 H_{y}-F_{1}-F_{2}-F_{3}-F_{4}-F_{5}-2F_{6}-F_{7}) + (F_{6} - F_{7}) + (F_{7} - F_{8}),
\end{equation*}
where $F_{i}$ stand for classes of exceptional divisors, and, since all of the $-2$-curves intersect at one point, 
$\mathcal{X}$ indeed has the type ${A}_{2}^{(1)*}$.
Unfortunately, two of the three irreducible components of $-K_{\mathcal{X}}$ are now completely in the blow-up region. This makes
identification with the standard example more difficult since we have to go through a sequence of coordinate charts to 
do the computation. A better approach is to use $\mathbb{P}^{2}$ compactification of $\mathbb{C}^{2}$ (recall that in this case,
according to general theory, we expect to have nine blow-up points instead of eight).

In this compactification we still have the same finite points that, in the homogeneous coordinates, are
\begin{alignat*}{2}
	&p_{1}\left(\frac{ (\theta_{1}^{1} + \theta_{2}^{1} + \theta_{3}^{1})(\theta_{1}^{2} + \theta_{3}^{1}) }{ \theta_{1}^{1} - \theta_{1}^{2} }:
	-\frac{ (\theta_{1}^{1} + \theta_{2}^{2} + \theta_{3}^{1})(\theta_{1}^{2} + \theta_{3}^{1}) }{ \theta_{1}^{1} - \theta_{1}^{2} }:1\right),&\quad 
	&p_{4}(0:0:1),\\
	&p_{4}\left(\frac{ (\theta_{1}^{1} + \theta_{2}^{1} + \theta_{3}^{2})(\theta_{1}^{2} + \theta_{3}^{2}) }{ \theta_{1}^{1} - \theta_{1}^{2} }:
	-\frac{ (\theta_{1}^{1} + \theta_{2}^{2} + \theta_{3}^{2})(\theta_{1}^{2} + \theta_{3}^{2}) }{ \theta_{1}^{1} - \theta_{1}^{2} }:1\right),
	&\quad &p_{5}(-\theta_{2}^{1}:\theta_{2}^{2}:1),\\
	&p_{3}\left(\frac{ (\theta_{1}^{1} + \theta_{2}^{1} + \theta_{3}^{3})(\theta_{1}^{2} + \theta_{3}^{3}) }{ \theta_{1}^{1} - \theta_{1}^{2} }:
	-\frac{ (\theta_{1}^{1} + \theta_{2}^{2} + \theta_{3}^{3})(\theta_{1}^{2} + \theta_{3}^{3}) }{ \theta_{1}^{1} - \theta_{1}^{2} }:1\right).
\end{alignat*}
There are also three more points on the line at infinity, and one infinitely close point $p_{9}$:	
\begin{equation*}
	p_{6}(1:-1:0)\leftarrow 
	p_{9}\left(0,\frac{ \theta_{1}^{1} - \theta_{1}^{2} }{ (\theta_{2}^{1} - \theta_{2}^{2})(\theta_{1}^{2} + 1) }\right),\quad 
	p_{7}(0:1:0),\quad p_{8}(1:0:0),
\end{equation*}
where coordinates of $p_{9}$ are w.r.t the coordinate system $u = \frac{ X + Y }{ X }$, $v = \frac{ Z }{ X+Y }$ in the chart $X\neq 0$.
Points $p_{1},...,p_{6}$ lie on the projectivization of the $(2,2)$-curve $Q$ whose homogeneous equation in $\mathbb{P}^{2}$
is
\begin{equation}
	(\theta_{1}^{1} - \theta_{1}^{2})(X + Y)(X + Y + (\theta_{2}^{1} - \theta_{2}^{2})Z) + 
	(\theta_{2}^{1} - \theta_{2}^{2})(\theta_{2}^{2} X + \theta_{2}^{1}Y) Z= 0.
	\label{eq:dpa2-quadric-proj}
\end{equation}
The resulting blow-up diagram is depicted on Figure~\ref{fig:dpa2-schlesinger-blowup-p2}.

\begin{figure}[h]
	\centering
		\includegraphics{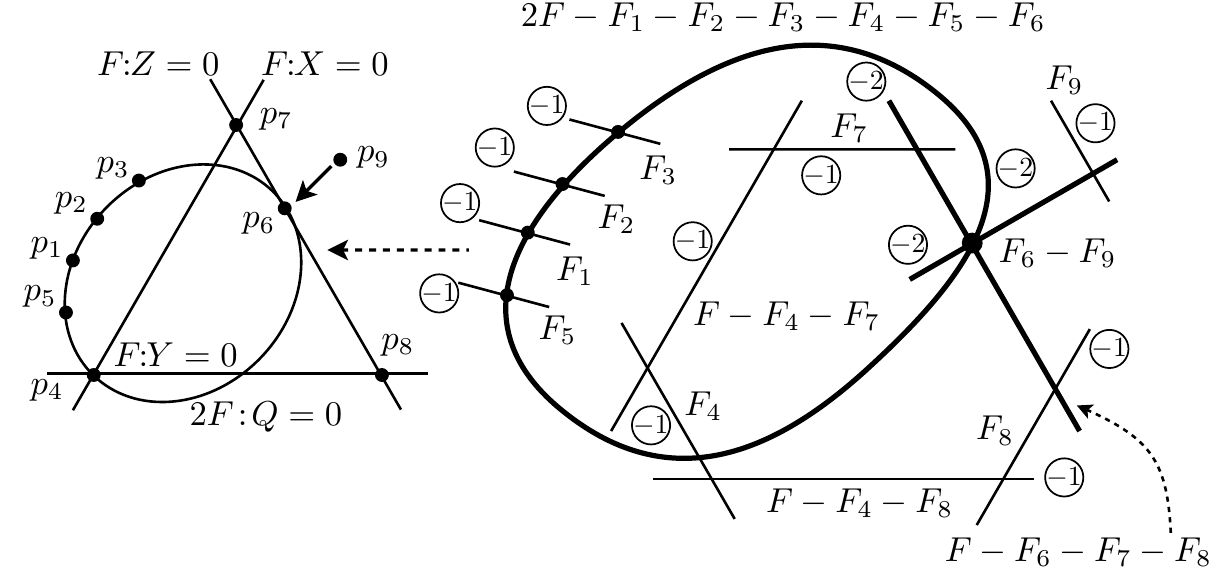}
	\caption{Okamoto surface $\mathcal{X}_{\mathbf{\theta}}$ for the Schlesinger transformations reduction to d-$P({A}_{2}^{(1)*})$ 
	using $\mathbb{P}^{2}$ compactification of $\mathbb{C}^{2}$.}
		\label{fig:dpa2-schlesinger-blowup-p2}
\end{figure}

As before, we see that the anti-canonical divisor $-K_{\mathcal{X}}$ uniquely decomposes as a positive linear combination of $-2$-curves $D_{i}$,
	\begin{equation*}
		-K_{\mathcal{X}} = 3F - \sum_{i=1}^{9} F_{i} = D_{0} + D_{1} + D_{2},
	\end{equation*}
	where
	\begin{equation*}
		D_{0} = 2F - F_{1}-F_{2}-F_{3}-F_{4}-F_{5}-F_{6},\quad  
		D_{1} = F - F_{6} - F_{7} - F_{8},\quad
		D_{2} = F_{6} - F_{9}.
	\end{equation*}
	The configuration of components $D_{i}$ is again described by the 
	Dynkin diagram of type ${A}^{(1)}_{2}$, and since all three $-2$-curves intersect at one point,
	this is a surface of type ${A}_{2}^{(1)*}$. To compare this dynamic with the model example 
	considered earlier we need to find an explicit isomorphism between the corresponding Okamoto
	surfaces, choose the same bases in the Picard lattice, and then compute the translation directions
	in the symmetry sub-lattice. This is what we do next.


\subsubsection{Reduction to the standard form} 
\label{ssub:reduction_to_the_standard_formA2}
To match the surface $\mathcal{X}_{\mathbf{\theta}}$ described by the blow-up diagram on Figure~\ref{fig:dpa2-schlesinger-blowup-p2} 
with the surface $\mathcal{X}_{\mathbf{b}}$ 
described by diagram on Figure~\ref{fig:dpa2-standard-blowup}, we look for the blow-down structure
describing $\mathcal{X}_{\mathbf{b}}$ in $\operatorname{Pic}(\mathcal{X}_{\mathbf{\theta}})$,
i.e., we look for rational classes $\mathcal{H}_{f}$, $\mathcal{H}_{g}$, $\mathcal{E}_{1},\dots, \mathcal{E}_{8}$ 
in $\operatorname{Pic}(\mathcal{X}_{\mathbf{\theta}})$ such that 
\begin{equation*}
	\mathcal{H}_{f}\bullet \mathcal{H}_{g} = 1, \  \mathcal{E}_{i}^{2} = -1, \  
	\mathcal{H}_{f}^{2} = \mathcal{H}_{g}^{2}  = \mathcal{H}_{f}\bullet \mathcal{E}_{i} = \mathcal{H}_{g}\bullet \mathcal{E}_{i}
	 = \mathcal{E}_{i}\bullet \mathcal{E}_{j} = 0,\  1\leq i\neq j\leq 8,
\end{equation*}
and the resulting configuration matches diagram on Figure~\ref{fig:dpa2-standard-blowup}. By the (virtual) genus formula
$g(C) = (C^{2} + K_{\mathcal{X}}\bullet C)/2 + 1$, we see that we should look for classes of rational curves of self-intersection zero among
$F - F_{i}$ and for classes of rational curves of self-intersection $-1$ among $F_{i}$ or $F - F_{i} - F_{j}$.

Comparing the $-2$-curves on both diagrams, 
\begin{align*}
	D_{0} &= 2F - F_{1} - F_{2} - F_{3} - F_{4} - F_{5} - F_{6} = \mathcal{H}_{f} + \mathcal{H}_{g} - \mathcal{E}_{1} - \mathcal{E}_{2}
	- \mathcal{E}_{3} - \mathcal{E}_{4},\\
	D_{1} &= F - F_{6} - F_{7} - F_{8} = \mathcal{H}_{f} - \mathcal{E}_{5} - \mathcal{E}_{6},\\
	D_{2} &= F_{6} - F_{9} = \mathcal{H}_{g} - \mathcal{E}_{7} - \mathcal{E}_{8},
\end{align*}
we see that it makes sense to choose $\mathcal{E}_{i} = F_{i}$ for $i=1,\dots,4$. Then $\mathcal{H}_{f} + \mathcal{H}_{g} = F - F_{5} - F_{6}$,
and looking at $D_{1}$ we put $\mathcal{H}_{f} = F - F_{6}$, $\mathcal{E}_{5} = F_{7}$, $\mathcal{E}_{6} = F_{8}$. This then requires that 
$\mathcal{H}_{g} = F- F_{5}$, and looking at $D_{2}$ we get $\mathcal{E}_{7} + \mathcal{E}_{8} = F - F_{5} - F_{6} + F_{9}$. We put 
$\mathcal{E}_{7} = F - F_{5} - F_{6}$ (to ensure that $\mathcal{H}_{f}\bullet \mathcal{E}_{7} = \mathcal{H}_{g} \bullet \mathcal{E}_{7} = 0$), and 
then $\mathcal{E}_{8} = F_{9}$. To summarize, we get the following identification, which clearly satisfies all of the required conditions
\begin{alignat*}{5}
	\mathcal{H}_{f}&= F - F_{6}, & \quad  \mathcal{E}_{1} &= F_{1}, & \quad  \mathcal{E}_{3} &= F_{3}, & \quad  
	\mathcal{E}_{5} &= F_{7}, & \quad  \mathcal{E}_{7} &= F - F_{5} - F_{6}, \\ 
	\mathcal{H}_{g}&= F - F_{5}, & \quad \mathcal{E}_{2} &= F_{2},  & \quad  \mathcal{E}_{4} &= F_{4}, & \quad  
	\mathcal{E}_{6} &= F_{8}, & \quad  \mathcal{E}_{8} &= F_{9}.
\end{alignat*}

To complete the correspondence it remains to define the base coordinates $f$ and $g$ 
of the linear systems $|\mathcal{H}_{f}|$ and $|\mathcal{H}_{g}|$ that will map the exceptional fibers of the divisors
$\mathcal{E}_{i}$ to the points $\pi_{i}$ such that $\pi_{5}$ and $\pi_{6}$ are on the line $f=\infty$, 
$\pi_{7}$ and $\pi_{8}$ are on $g=\infty$, and $\pi_{1},\dots \pi_{4}$ are on the line $f+g=0$. Since the 
pencil $|\mathcal{H}_{f}|$ consists of all curves on $\mathbb{P}^{2}$ passing through $p_{6}(1:-1:0)$,
\begin{align*}
	|\mathcal{H}_{f}| = |F - F_{6}| &= \{a X + b Y + c Z = 0 \mid a - b =0\} = \{a (X + Y)  + c Z = 0\},
\intertext{we can define the projective base coordinate as $f_{1} = [X + Y  : Z]$.
Similarly,}
	|\mathcal{H}_{g}| = |F - F_{5}| &= \{a X + b Y + c Z = 0 \mid -\theta_{2}^{1} a  + \theta_{2}^{2} b+ c =0\} \\
	&= \{a (X + \theta_{2}^{1} Z) + b (Y - \theta_{2}^{2}Z) = 0\},
\end{align*}
and $g_{1} = [X + \theta_{2}^{1} Z : Y - \theta_{2}^{2} Z]$. Then 
\begin{alignat*}{2}
	f_{1}(\pi_{5}) &= f_{1}(p_{7}) = [-1:0],\qquad  &
	g_{1}(\pi_{7}) &= g_{1}(p_{6}) = [1:-1],\\
	f_{1}(\pi_{6}) &= f_{1}(p_{8}) = [1:0],\qquad &
	g_{1}(\pi_{8}) &= g_{1}(p_{6}) = [1:-1].
\end{alignat*}
In order to have $g_{1}(\pi_{7}) = g_{1}(\pi_{8}) = \infty$ we first make an affine change of coordinates
$\tilde{g}_{1} = [(X + \theta_{2}^{1} Z)  + (Y - \theta_{2}^{2} Z): Y - \theta_{2}^{2} Z])]$ to get 
$\tilde{g}_{1}(\pi_{7}) = \tilde{g}_{1}(\pi_{8}) = [0:-1]$ and then put
\begin{equation*}
	f_{2} = \frac{ X + Y }{ Z },\qquad g_{2} = \frac{ Y - \theta_{2}^{2} Z }{ (X + Y) + (\theta_{2}^{1} - \theta_{2}^{2})Z }.
\end{equation*}
Equation~(\ref{eq:dpa2-quadric-proj}) of the curve $Q$ in these coordinates becomes
\begin{equation}
	Z^{2}(f + \theta_{2}^{1} - \theta_{2}^{2})((\theta_{1}^{1} - \theta_{1}^{2})f_{2} + 
	(\theta_{2}^{1} - \theta_{2}^{2})((\theta_{2}^{1} - \theta_{2}^{2})g_{2} + \theta_{2}^{2})) = 0,
\end{equation}
and the points $\pi_{1},\dots,\pi_{4}$ lie on the line $(\theta_{1}^{1} - \theta_{1}^{2})f_{2} + 
(\theta_{2}^{1} - \theta_{2}^{2})((\theta_{2}^{1} - \theta_{2}^{2})g_{2} + \theta_{2}^{2})) = 0$.
Thus, if we finally put
\begin{align}
	f &= \frac{ (\theta_{1}^{1} - \theta_{1}^{2})}{(\theta_{2}^{1} - \theta_{2}^{2}) } f_{2} 
	= \frac{ (\theta_{1}^{1} - \theta_{1}^{2})(X+Y)  }{ (\theta_{2}^{1} - \theta_{2}^{2}) Z } 
	= \frac{ (\theta_{1}^{1} - \theta_{1}^{2})(x + y)  }{ (\theta_{2}^{1} - \theta_{2}^{2})},\\
	g &= (\theta_{2}^{1} - \theta_{2}^{2})g_{2} + \theta_{2}^{2} = 
	\frac{ \theta_{2}^{2} X + \theta_{2}^{1} Y }{ (X + Y) + (\theta_{2}^{1} - \theta_{2}^{2})Z }
	= \frac{ \theta_{2}^{2} x + \theta_{2}^{1} y }{ (x + y) + (\theta_{2}^{1} - \theta_{2}^{2}) },
\end{align}
points $\pi_{1},\dots,\pi_{4}$ will be on the line $f + g = 0$, points $\pi_{5}$ and $\pi_{6}$ will be on 
the line $f = \infty$, and points $\pi_{7}$ and $\pi_{8}$ will be on the line $g = \infty$, as requires.
Specifically, we get
\begin{alignat*}{4}
	\pi_{1}&(\theta_{1}^{2} + \theta_{3}^{1}, -\theta_{1}^{2} - \theta_{3}^{1}), &\quad 	
	\pi_{3}&(\theta_{1}^{2} + \theta_{3}^{3}, -\theta_{1}^{2} - \theta_{3}^{3}), &\quad 	
	\pi_{5}&(\infty,\theta_{2}^{1}), &\quad 	
	\pi_{7}&(\theta_{1}^{2} - \theta_{1}^{1},\infty),  \\
	\pi_{2}&(\theta_{1}^{2} + \theta_{3}^{2}, - \theta_{1}^{2} - \theta_{3}^{2}), &\quad 	
	\pi_{4}&(0,0), &\quad 	
	\pi_{6}&(\infty,\theta_{2}^{2}), &\quad 	
	\pi_{8}&(\theta_{1}^{2} + 1, \infty).
\end{alignat*}	
Thus, we immediately get the identification between the parameters in the Riemann scheme of our Fuchsian system and the 
parameters $b_{i}$ in the model equation:
\begin{alignat*}{4}
	b_{1} &= \theta_{1}^{2} + \theta_{3}^{1},  &\quad 	
	b_{3} &= \theta_{1}^{2} + \theta_{3}^{3},  &\quad 	
	b_{5} &= \theta_{2}^{1}, &\quad 	
	b_{7} &= \theta_{1}^{1} - \theta_{1}^{2},  \\
	b_{2} &= \theta_{1}^{2} + \theta_{3}^{2}, &\quad 	
	b_{4} &= 0, &\quad 	
	b_{6} &= \theta_{2}^{2}, &\quad 	
	b_{8} &= -\theta_{1}^{2} - 1.
\end{alignat*}	
This, in turn, allows us to see the effect of the standard Painlev\'e dynamic on 
the  Riemann scheme. Indeed, $\delta = b_{1} + \cdots + b_{8} = -1$, and, for example,
$\bar{k}_{1} = \bar{b}_{1} + \bar{b}_{8} + 1 = b_{1} + b_{8} - \delta + 1 = k_{1} +1$, 
and so on. So for the model equation we get 
\begin{align*}
	\left\{
	\begin{tabular}{cccc}
		$z_{1}$ 		& $	z_{2}$		& 	$z_{3}$	\\
		$\theta_{1}^{1}$	& $\theta_{2}^{1}$	& 	$\theta_{3}^{1}$ \\
		$\theta_{1}^{2}$	& $\theta_{2}^{2}$	& 	$\theta_{3}^{2}$ \\
		$ 0 $			& 	$ 0$		& 	$\theta_{3}^{3}$ 
	\end{tabular}
	\right\}		&	\overset{\text{d-$P(A_2^{(1)*})$}}{\strut\longmapsto}
	\left\{
	\begin{tabular}{cccc}
		$z_{1}$ 		& $	z_{2}$		& 	$z_{3}$	\\
		$\theta_{1}^{1}$	& $\theta_{2}^{1}-1$	& 	$\theta_{3}^{1}+1$ \\
		$\theta_{1}^{2}-1$	& $\theta_{2}^{2}-1$	& 	$\theta_{3}^{2}+1$ \\
		$ 0 $			& 	$ 0$		& 	$\theta_{3}^{3}+1$ 
	\end{tabular}
	\right\},
	\intertext{whereas our elementary Schlesinger transformation acts as }
	\left\{
	\begin{tabular}{cccc}
		$ z_{1}$ 		& $	z_{2} $			& 	$ z_{3}$	\\
		$\theta_{1}^{1}$	& $\theta_{2}^{1}$		& 	$\theta_{3}^{1}$ \\
		$\theta_{1}^{2}$	& $\theta_{2}^{2}$		& 	$\theta_{3}^{2}$ \\
		$0$			& 	$0$					& 	$\theta_{3}^{3}$ 
	\end{tabular}
	\right\} &\overset{\left\{\begin{smallmatrix} 1&2\\1&1\end{smallmatrix}\right\}}{\strut\longmapsto}
	\left\{
	\begin{tabular}{cccc}
		$ z_{1}$ 		& $	z_{2} $			& 	$ z_{3}$	\\
		$\theta_{1}^{1} - 1$	& $\theta_{2}^{1} + 1$		& 	$\theta_{3}^{1}$ \\
		$\theta_{1}^{2}$	& $\theta_{2}^{2}$		& 	$\theta_{3}^{2}$ \\
		$0$			& 	$0$			 		& 	$\theta_{3}^{3}$ 
	\end{tabular}
	\right\}. 
\end{align*}
Thus, these two transformations correspond to the different translation directions in the symmetry root sub-lattice
of the surface $\tilde{X}$ and so are not equivalent. 
Indeed, we compute the action of $\psi_{*}$ of an elementary Schlesinger transformation 
$\left\{\begin{smallmatrix} 1&2\\1&1\end{smallmatrix}\right\}$ 
on the classes $\mathcal{H}_{f}$, $\mathcal{H}_{g}$, and $\mathcal{E}_{i}$ to be
\begin{align*}
	 \mathcal{H}_{f}&\mapsto 2 \mathcal{H}_{f} + 3 \mathcal{H}_{g} - \mathcal{E}_{1} - \mathcal{E}_{2} - \mathcal{E}_{3} - 
	\mathcal{E}_{4} - 2 \mathcal{E}_{5} - 2\mathcal{E}_{8},
	\\
	 \mathcal{H}_{g}&\mapsto 3 \mathcal{H}_{f} + 5 \mathcal{H}_{g} - 2\mathcal{E}_{1} - 2\mathcal{E}_{2} - 2\mathcal{E}_{3} - 
	2\mathcal{E}_{4} - 3 \mathcal{E}_{5} - \mathcal{E}_{6} - 2\mathcal{E}_{8},\label{eq:dpa-hg-sch}\\
	 \mathcal{E}_{1}&\mapsto \mathcal{H}_{f} + 2\mathcal{H}_{g} - \mathcal{E}_{2} - \mathcal{E}_{3} - \mathcal{E}_{4} - 
	\mathcal{E}_{5} - \mathcal{E}_{8},
	\\
	 \mathcal{E}_{2}&\mapsto \mathcal{H}_{f} + 2\mathcal{H}_{g} - \mathcal{E}_{1} - \mathcal{E}_{3} - \mathcal{E}_{4} - 
	\mathcal{E}_{5} - \mathcal{E}_{8},
	\\
	 \mathcal{E}_{3}&\mapsto \mathcal{H}_{f} + 2\mathcal{H}_{g} - \mathcal{E}_{1} - \mathcal{E}_{2} - \mathcal{E}_{4} - 
	\mathcal{E}_{5} - \mathcal{E}_{8},
	\\
	 \mathcal{E}_{4}&\mapsto \mathcal{H}_{f} + 2\mathcal{H}_{g} - \mathcal{E}_{1} - \mathcal{E}_{2} - \mathcal{E}_{3} - 
	\mathcal{E}_{5} - \mathcal{E}_{8},
	\\
	 \mathcal{E}_{5}&\mapsto \mathcal{E}_{7},
	\\
	 \mathcal{E}_{6}&\mapsto 2\mathcal{H}_{f} + 2\mathcal{H}_{g} - \mathcal{E}_{1} - \mathcal{E}_{2} - \mathcal{E}_{3} - 
	\mathcal{E}_{4} - 2\mathcal{E}_{5} - \mathcal{E}_{8},
	\\
	 \mathcal{E}_{7}&\mapsto 2\mathcal{H}_{f} + 3\mathcal{H}_{g} - \mathcal{E}_{1} - \mathcal{E}_{2} - \mathcal{E}_{3} - 
	\mathcal{E}_{4} - 2\mathcal{E}_{5} - \mathcal{E}_{6} - 2\mathcal{E}_{8},
	\\
	 \mathcal{E}_{8}&\mapsto \mathcal{H}_{g} - \mathcal{E}_{5},
\end{align*}
and compare with the standard dynamic $\varphi_{*}$ given by~(\ref{eq:dpa-tr-std}) to see this explicitly:
\begin{align*}
	\psi_{*}: (\alpha_{0}, \alpha_{1}, \alpha_{2}, \alpha_{3}, \alpha_{4}, \alpha_{5}, \alpha_{6})&\mapsto	
	(\alpha_{0}, \alpha_{1}, \alpha_{2}, \alpha_{3}, \alpha_{4}, \alpha_{5}, \alpha_{6}) + \\
	&\qquad(0,0,0,-1,1,1,-1)\left(-K_{\mathcal{X}}\right),\\
	\varphi_{*}: (\alpha_{0}, \alpha_{1}, \alpha_{2}, \alpha_{3}, \alpha_{4}, \alpha_{5}, \alpha_{6})&\mapsto	
	(\alpha_{0}, \alpha_{1}, \alpha_{2}, \alpha_{3}, \alpha_{4}, \alpha_{5}, \alpha_{6}) + \\
	&\qquad(0,0,0,1,0,-1,0)\left(-K_{\mathcal{X}}\right).
\end{align*}

It is possible to represent the standard Painlev\'e dynamic as a composition of two elementary 
Schlesinger transformations, combined with some automorphisms of our Fuchsian system. We first demonstrate this
by looking at a sequence of actions on the Riemann scheme:
\begin{align*}
	\left\{
	\begin{tabular}{cccc}
		$ z_{1}$ 		& $	z_{2} $			& 	$ z_{3}$	\\
		$\theta_{1}^{1}$	& $\theta_{2}^{1}$		& 	$\theta_{3}^{1}$ \\
		$\theta_{1}^{2}$	& $\theta_{2}^{2}$		& 	$\theta_{3}^{2}$ \\
		$0$			& 	$0$					& 	$\theta_{3}^{3}$ 
	\end{tabular}
	\right\}
	&\overset{\left\{\begin{smallmatrix} 2&1\\1&1\end{smallmatrix}\right\}}{\strut\longmapsto}
	\left\{
	\begin{tabular}{cccc}
		$ z_{1}$ 		& $	z_{2} $			& 	$ z_{3}$	\\
		$\theta_{1}^{1}+1$	& $\theta_{2}^{1}-1$		& 	$\theta_{3}^{1}$ \\
		$\theta_{1}^{2}$	& $\theta_{2}^{2}$		& 	$\theta_{3}^{2}$ \\
		$0$			& 	$0$					& 	$\theta_{3}^{3}$ 
	\end{tabular}
	\right\} \\
	&\overset{\sigma_{1}(1,3)}{\strut\longmapsto}
	\left\{
	\begin{tabular}{cccc}
		$ z_{1}$ 		& $	z_{2} $			& 	$ z_{3}$	\\
		$0$	& $\theta_{2}^{1}-1$		& 	$\theta_{3}^{1}$ \\
		$\theta_{1}^{2}$	& $\theta_{2}^{2}$		& 	$\theta_{3}^{2}$ \\
		$\theta_{1}^{1}+1$			& 	$0$					& 	$\theta_{3}^{3}$ 
	\end{tabular}
	\right\}
	\\
	&\overset{\rho_{1}(-\theta_{1}^{1} - 1)}{\strut\longmapsto}
	\left\{
	\begin{tabular}{cccc}
		$ z_{1}$ 		& $	z_{2} $			& 	$ z_{3}$	\\
		$-\theta_{1}^{1} - 1$	& $\theta_{2}^{1}-1$		& 	$\theta_{3}^{1} +\theta_{1}^{1} + 1$ \\
		$\theta_{1}^{2}-\theta_{1}^{1} - 1$	& $\theta_{2}^{2}$		& 	$\theta_{3}^{2} +\theta_{1}^{1} + 1 $ \\
		$0$			& 	$0$					& 	$\theta_{3}^{3}+\theta_{1}^{1} + 1$ 
	\end{tabular}
	\right\}\\
	&\overset{\left\{\begin{smallmatrix} 2&1\\2&1\end{smallmatrix}\right\}}{\strut\longmapsto}
	\left\{
	\begin{tabular}{cccc}
		$ z_{1}$ 		& $	z_{2} $			& 	$ z_{3}$	\\
		$-\theta_{1}^{1} $	& $\theta_{2}^{1}-1$		& 	$\theta_{3}^{1} +\theta_{1}^{1} + 1$ \\
		$\theta_{1}^{2}-\theta_{1}^{1} - 1$	& $\theta_{2}^{2}-1$		& 	$\theta_{3}^{2} +\theta_{2}^{1} + 1 $ \\
		$0$			& 	$0$					& 	$\theta_{3}^{3}+\theta_{1}^{1} + 1$ 
	\end{tabular}
	\right\}\\
	&\overset{\sigma_{1}(1,3)}{\strut\longmapsto}
	\left\{
	\begin{tabular}{cccc}
		$ z_{1}$ 		& $	z_{2} $			& 	$ z_{3}$	\\
		$0$	& $\theta_{1}^{1}-1$		& 	$\theta_{3}^{1} +\theta_{1}^{1} + 1$ \\
		$\theta_{1}^{2}-\theta_{1}^{1} - 1$	& $\theta_{2}^{2}-1$		& 	$\theta_{3}^{2} +\theta_{1}^{1} + 1 $ \\
		$\theta_{1}^{1} $			& 	$0$					& 	$\theta_{3}^{3}+\theta_{1}^{1} + 1$ 
	\end{tabular}
	\right\}\\	
	&\overset{\rho_{1}(\theta_{1}^{1})}{\strut\longmapsto}	
	\left\{
	\begin{tabular}{cccc}
		$ z_{1}$ 		& $	z_{2} $			& 	$ z_{3}$	\\
		$\theta_{1}^{1} $	& $\theta_{2}^{1}-1$		& 	$\theta_{3}^{1} + 1$ \\
		$\theta_{1}^{2} - 1$	& $\theta_{2}^{2}-1$		& 	$\theta_{3}^{2} + 1 $ \\
		$0$			& 	$0$					& 	$\theta_{3}^{3}+ 1$ 
	\end{tabular}
	\right\}.
\end{align*}
Here $\left\{\begin{smallmatrix} 2&1\\1&1\end{smallmatrix}\right\}$ and $\left\{\begin{smallmatrix} 2&1\\2&1\end{smallmatrix}\right\}$
are the usual elementary Schlesinger transformations, the map $\rho_{i}(s):\mathbf{A}(z)\mapsto (z - z_{i})^{s}\mathbf{A}(z)$ is a scalar gauge
transformation, and $\sigma_{i}(j,k)$ is a map that exchanges the $j$-th and the $k$-th eigenvectors (and eigenvalues) of $\mathbf{A}_{i}$. Note that, 
if the eigenvalues $\theta_{i}^{j}$ and $\theta_{i}^{k}$ are non-zero, this map is just a permutation 
on the decomposition space $\mathcal{B} \times \mathcal{C}$. And even though for the map $\sigma_{1}(1,3)$ that we use above one of the eigenvectors
has the eigenvalue zero, this map is still well-defines as a map on $\mathcal{B} \times \mathcal{C}$, since 
$\mathbf{b}_{1}^{3} \in \operatorname{Ker}(\mathbf{C}_{1}^{\dag})$
and $\mathbf{c}_{1}^{3\dag}\in \operatorname{Ker}(\mathbf{B}_{1})$. In fact, if we combine $\sigma_{1}(1,3)$ with $\rho_{1}(-\theta_{1}^{1})$
to define a transformation $\Sigma_{1}(1,3) = \rho_{1}(-\theta_{1}^{1}) \circ \sigma_{1}(1,3)$, 
\begin{equation*}
		\Sigma_{1}(1,3) : \left\{
		\begin{tabular}{cccc}
			$ z_{1}$ 		& $	z_{2} $			& 	$ z_{3}$	\\
			$\theta_{1}^{1}$	& $\theta_{2}^{1}$		& 	$\theta_{3}^{1}$ \\
			$\theta_{1}^{2}$	& $\theta_{2}^{2}$		& 	$\theta_{3}^{2}$ \\
			$0$			& 	$0$					& 	$\theta_{3}^{3}$ 
		\end{tabular}
		\right\} {\strut\longmapsto}
		\left\{
		\begin{tabular}{cccc}
			$ z_{1}$ 						& $	z_{2} $				& 	$ z_{3}$	\\
			$-\theta_{1}^{1} $				& $\theta_{2}^{1} $		& 	$\theta_{3}^{1} + \theta_{1}^{1}$ \\
			$\theta_{1}^{2}-\theta_{1}^{1}$	& $\theta_{2}^{2}$		& 	$\theta_{3}^{2} + \theta_{1}^{1}$ \\
			$0$			& 	$0$			 		& 	$\theta_{3}^{3} + \theta_{1}^{1}$ 
		\end{tabular}
		\right\},	
\end{equation*}
the action of $\Sigma_{1}(1,3)$ on the decomposition space is explicitly given by 
\begin{multline*}
	\Sigma_{1}(1,3) : \left(\mathbf{b}_{1}^{1},\mathbf{b}_{1}^{2}; \mathbf{c}_{1}^{1\dag},\mathbf{c}_{1}^{2\dag};
	\mathbf{b}_{2}^{1},\mathbf{b}_{2}^{2}; \mathbf{c}_{2}^{1\dag},\mathbf{c}_{2}^{2\dag}\right) \mapsto \\
	\left((\mathbf{c}_{1}^{1\dag}\times \mathbf{c}_{1}^{2\dag})^{t},\mathbf{b}_{1}^{2}; 
	\frac{ - \theta_{1}^{1} (\mathbf{b}_{1}^{1}\times\mathbf{b}_{1}^{2})^{t} }{ 
	(\mathbf{c}_{1}^{1\dag}\times \mathbf{c}_{1}^{2\dag}) (\mathbf{b}_{1}^{1}\times\mathbf{b}_{1}^{2}) },\mathbf{c}_{1}^{2\dag};
	\mathbf{b}_{2}^{1},\mathbf{b}_{2}^{2}; \mathbf{c}_{2}^{1\dag},\mathbf{c}_{2}^{2\dag}\right),
\end{multline*}
where $\times$ is the usual cross-product and $t$ denotes transposition.
We also had to use the normalization condition $\mathbf{C}_{1}^{\dag} \mathbf{B}_{1} = \mathbf{\Theta}_{1}$.

It is also possible to show, by a direct computation, that
\begin{equation*}
	\text{d-$P(A_2^{(1)*})$} = 
	\Sigma_{1}(1,3)\circ\left\{\begin{smallmatrix} 2&1\\2&1\end{smallmatrix}\right\}\circ
	\Sigma_{1}(1,3)\circ\left\{\begin{smallmatrix} 2&1\\1&1\end{smallmatrix}\right\}
\end{equation*}
holds on the level of equations as well.

\subsection{Reductions to difference Painlev\'e equation of type d-$P\left(A_{1}^{(1)*}\right)$ with the symmetry group $E^{(1)}_{7}$.} 
\label{ssub:reductions_to_difference_painlev'e_equation_2}

\subsubsection{Model Example} 
\label{ssub:model_exampleA1}

For our model example of d-$P(A_{1}^{(1)*})$ equation we take the equation that first was appeared in 
\cite{GraRamOht:2003:AUDOTAQVADIEATST} as an asymmetric q-$P_{\text{IV}}$ equation, and
we use the variables as in Sakai's paper
\cite{Sak:2007:PDPEATLF}. We again consider d-$P(A_{1}^{(1)*})$ to be a
birational map $\varphi: \mathbb{P}^{1}\times \mathbb{P}^{1} \dashrightarrow \mathbb{P}^{1}\times \mathbb{P}^{1}$ 
with parameters $b, b_{1},\dots, b_{8}$, 
\begin{equation*}
	\left(\begin{matrix}
		b & b_{1} & b_{2} & b_{3} & b_{4}\\
		 & b_{5} & b_{6} & b_{7} & b_{8}
	\end{matrix}; f,g\right) \mapsto 
	\left(\begin{matrix}
		\bar{b} & \bar{b}_{1} & \bar{b}_{2} & \bar{b}_{3} & \bar{b}_{4}\\
		& \bar{b}_{5} & \bar{b}_{6} & \bar{b}_{7} & \bar{b}_{8}
	\end{matrix}; \bar{f},\bar{g}\right),
	\quad
	\begin{aligned}
		\bar{b}	 &= b - \delta,\\
		 \bar{b}_{i} &= b_{i},\quad i = 1,\dots,8,\\		
	\end{aligned}
\end{equation*} 
$\delta = b_{1} + \cdots + b_{8}$, and   $\bar{f}$ and $\bar{g}$ are given by the equations
\begin{equation}
	\left\{
	\begin{aligned}
		\frac{ (g + f - 2b) (g + \bar{f} - b - \bar{b})}{ (g + f) (g + \bar{f}) } &= 
		\frac{ \prod_{i=1}^{4} (g  - b + b_{i}) }{ \prod_{i=5}^{8} (g - b_{i}) }\\
		\frac{ (g + \bar{f} - b - \bar{b}) (\bar{g} + \bar{f} - 2\bar{b})}{ (g + \bar{f}) (\bar{g} + \bar{f}) } &= 
		\frac{ \prod_{i=1}^{4} (\bar{f} - \bar{b} - b_{i} ) }{ \prod_{i=5}^{8} (\bar{f} + b_{i}) }
	\end{aligned}
	\right.\,.\label{eq:dpA1-st}	
\end{equation}
It is convenient to introduce the notation
\begin{equation*}
	G_{14} = G_{14}(g) = \prod_{i=1}^{4} (g - b + b_{i} ),\quad 
	G_{14}^{i} = G_{14}^{i}(g) = \prod_{j=1,j\neq i}^{4} (g - b + b_{j}),
\end{equation*}
and similarly for $G_{58}$, $F_{14}$, and $F_{58}$. The maps $\varphi$ is then given by 
the sequence $(f,g)\to(\bar{f},g)\to(\bar{f},\bar{g})$ described by the equations
\begin{align}
	\bar{f} &= - \frac{ (g - b - \bar{b}) (f + g - 2b) G_{58} - g(f+g) G_{14}}{ (f + g - 2b) G_{58} - (f+g) G_{14}  }, \label{eq:dpA1-phi1f}\\
	\bar{g} &= - \frac{ (\bar{f} - 2\bar{b}) (\bar{f} + g - b - \bar{b}) F_{58} - \bar{f} (\bar{f} + g) F_{14} }{ 
		(\bar{f} + g - b - \bar{b}) F_{58} -  (\bar{f} + g) F_{14}  }, \label{eq:dpA1-phi1g}
		\intertext{and $\varphi^{-1}$ is given by $(\bar{f},\bar{g})\to(\bar{f},g)\to(f,g)$ given by }
	g &= - \frac{ (\bar{f} - b - \bar{b}) (\bar{f} + \bar{g} - 2\bar{b}) F_{58} - \bar{f}(\bar{f}+\bar{g}) F_{14}}{ 
	(\bar{f} + \bar{g} - 2\bar{b}) F_{58} - (\bar{f}+\bar{g}) F_{14} }, \notag\\
	f &= - \frac{ (g - 2b) (\bar{f} + g - b - \bar{b}) G_{58} - g (\bar{f} + g) G_{14} }{ 
		(\bar{f} + g - b - \bar{b}) G_{58} - (\bar{f} + g) G_{14}  }.\notag
\end{align}

It is easy to see that the indeterminate points of the first map $\bar{f} = \bar{f}(f,g)$ are either given by the conditions 
$f+g = 2b$ and $G_{14} = 0$, or by the conditions $f+g  =0$ and $G_{58} = 0$. Thus, we get 8 indeterminate points
lying on two curves of bi-degree $(1,1)$, $C_{b}: f+g = 2b$ and $C_{0}: f + g = 0$:
\begin{equation*}
	p_{i}(b+b_{i},b-b_{i}),\quad i=1,\dots,4\qquad\text{and}\qquad p_{i}(-b_{i},b_{i}),\quad i=5,\dots,8
\end{equation*}
on the $(f,g)$-plane. It is also easy to see that these are also the indeterminate points of all of the other maps
(with $b$ changed to $\bar{b}$ for $(\bar{f},\bar{g})$-coordinates). We then get the following blowup diagram
describing the Okamoto space of initial conditions $\mathcal{X}_{\mathbf{b}}$  
on Figure~\ref{fig:dpa1-standard-blowup}.
\begin{figure}[h]
	\centering
		\includegraphics{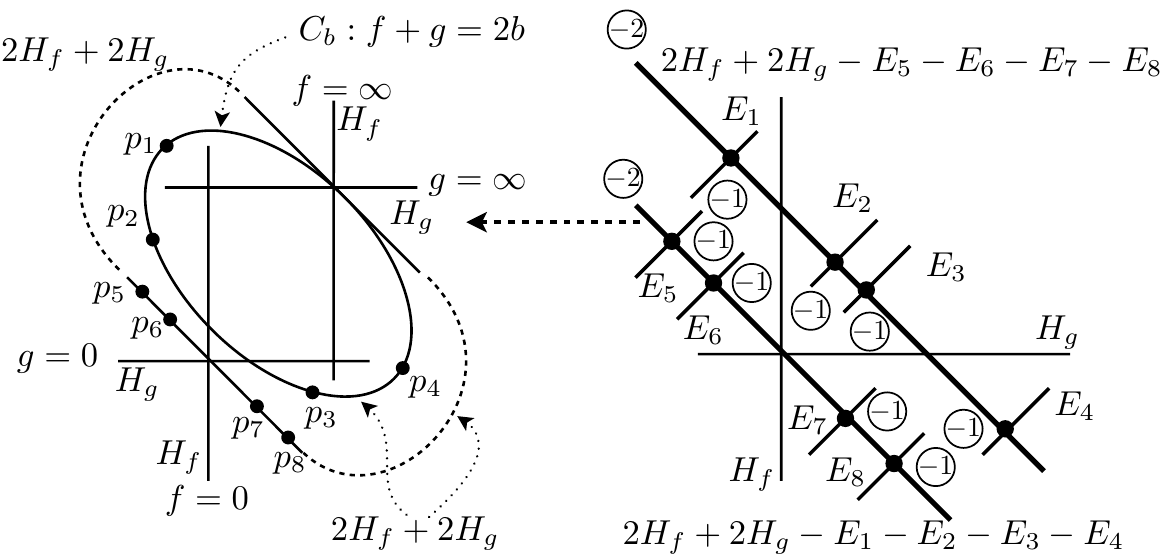}
	\caption{Okamoto surface $\mathcal{X}_{\mathbf{b}}$ for the model form of d-$P({A}_{1}^{(1)*})$.}
	\label{fig:dpa1-standard-blowup}
\end{figure}	

In $\operatorname{Pic}(\mathcal{X}_{\mathbf{b}}) = \mathbb{Z} H_{f}\oplus \mathbb{Z} H_{g} \oplus \bigoplus_{i=1}^{8} \mathbb{Z} E_{i}$,
the anti-canonical divisor again  decomposes uniquely as the sum of two connected components,
 \begin{align*}
	-K_{\mathcal{X}}&=2 H_{f} + 2 H_{g} - \sum_{i=1}^{8} E_{i} = D_{0} + D_{1},\qquad\text{where}\\
	D_{0} &= H_{f} + H_{g} - E_{1} - E_{2} - E_{3} - E_{4},\\
	D_{1} &= H_{f} + H_{g} - E_{5} - E_{6} - E_{7} - E_{8},	
 \end{align*}
and $D_{1}^{2} = D_{2}^{2} = - D_{1}\bullet D_{2} = -2$. Thus, the
configuration of components $D_{i}$ is described by the 
Dynkin diagram of type ${A}^{(1)}_{1}$. To this diagram again
correspond two different types of surfaces, the generic one corresponding to divisors $D_{1}$ and $D_{2}$
intersecting at two points gives a multiplicative
system of type ${A}^{(1)}_{1}$, and the degenerate configuration corresponding to two components touch at one point 
gives an additive system denoted by 
${A}_{1}^{(1)*}$, which 
is our case, see Figure~\ref{fig:dpa1-configs}.

\begin{figure}[h]
	\begin{center}
		\begin{tabular}{ccc}
			\includegraphics{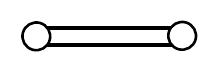} & \quad
			\includegraphics{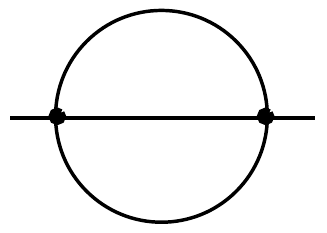} \quad & 
			\includegraphics{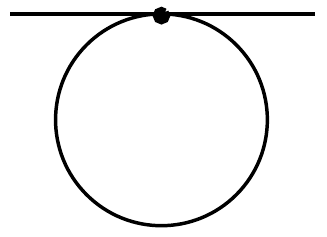} \\
			Dynkin diagram ${A}_{1}^{(1)}$ & ${A}^{(1)}_{1}$-surface & ${A}_{1}^{(1)*}$-surface.
		\end{tabular}
	\end{center}
	\caption{Configurations of type ${A}_{1}^{(1)}$}
	\label{fig:dpa1-configs}
\end{figure}

The symmetry sub-lattice
$R^{\perp}=\operatorname{Span}_{\mathbb{Z}}\{\alpha_{0},\dots,\alpha_{7}\}$ is of type ${E}^{(1)}_{7}$,
where the basis of $\alpha_{i}$ is given on 
Figure~\ref{fig:dpa1-symm}. 

\begin{figure}[h]
\begin{equation*}
	\begin{aligned}
		\alpha_{1}&= E_{3} - E_{4},& \quad \alpha_{5}&= E_{5} - E_{6}, \\  
		\alpha_{2}&= E_{2} - E_{3},& \quad  \alpha_{6}&= E_{6} - E_{7},\\
		\alpha_{3}&= E_{1} - E_{2},& \quad  \alpha_{0}&= E_{7} - E_{8}, \\
		\alpha_{4}&= H_{f} - E_{1} - E_{5},& \quad  \alpha_{7}&= H_{g} - H_{f}, 
	\end{aligned}\qquad\qquad
		\raisebox{-0.2in}{\includegraphics{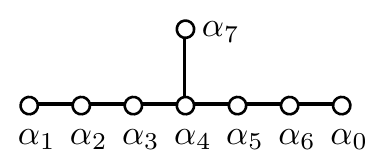}}.
\end{equation*}
	\caption{Symmetry sub-lattice for d-$P(\tilde{A}_{1}^{*})$}
	\label{fig:dpa1-symm}
\end{figure}	

To compute the action of $\varphi_{*}$ on $\operatorname{Pic}(\mathcal{X}_{\mathbf{b}})$,
we decompose $\varphi = \varphi_{2} \circ \varphi_{1}$, where $\varphi_{1}:(f,g)\to (\bar{f},g)$ is 
given by~(\ref{eq:dpA1-phi1f}), and $\varphi_{2}:(\bar{f},g)\to (\bar{f},\bar{g})$ is 
given by~(\ref{eq:dpA1-phi1g}). Then $(\varphi_{1})_{*}(H_{g}) = H_{g}$, 
it is straightforward to see that $(\varphi_{1})_{*}(E_{i}) = H_{g} - E_{i}$ and that 
$(\varphi_{1})_{*}(D_{0}) = D_{1}$, $(\varphi_{1})_{*}(D_{1}) = D_{0}$, which gives 
$(\varphi_{1})_{*} (H_{f}) = H_{f} + 4 H_{g} - E$, where $E = \sum_{i=1}^{8} E_{i}$.
The situation with $\varphi_{2}$ is completely symmetric, 
$(\varphi_{2})_{*}(H_{f}) = H_{f}$, $(\varphi_{2})_{*}(E_{i}) = H_{f} - E_{i}$, and
$(\varphi_{2})_{*} (H_{g}) = H_{g} + 4 H_{f} - E$. Composing these two linear maps, 
we get the action of $\varphi_{*}$:
\begin{align*}
	H_{f} &\mapsto 9 H_{f} + 4 H_{g} - 3 E\\
	H_{g} &\mapsto 4 H_{f} +  H_{g} -  E\\
	E_{i} &\mapsto 3 H_{f} +  H_{g} -  E + E_{i},\qquad i=1,\dots,8,
\end{align*}
and so the induced action $\varphi_{*}$ on the sub-lattice $R^{\perp}$	is given by the following {translation}:
\begin{align}
	(\alpha_{0}, \alpha_{1}, \alpha_{2}, \alpha_{3}, \alpha_{4}, \alpha_{5}, \alpha_{6},\alpha_{7})&\mapsto	
	(\alpha_{0}, \alpha_{1}, \alpha_{2}, \alpha_{3}, \alpha_{4}, \alpha_{5}, \alpha_{6},\alpha_{7}) + \label{eq:dpa1-trans-std}\\
	&\qquad
	(0,0,0,0,1,0,0,-2)(-K_{\mathcal{X}}). \notag
\end{align}


\subsubsection{Schlesinger Transformations} 
\label{ssub:schlesinger_transformationsA2}
We start with a $4\times4$ Fuchsian system of the spectral type $1111,1111,22$. In this case it is convenient
to have all singular points to be finite, since we need to consider two different kinds of elementary Schlesinger 
transformations --- one between two points with non-repeating eigenvalues, and the other when at one point we
have an eigenvalue of multiplicity two. As before, we can use scalar gauge transformations to make some of the 
eigenvalues to vanish, and so, putting $\theta_{3}=\theta_{3}^{1}=\theta_{3}^{2}$ we take our Riemann scheme to be 
\begin{equation*}
	\left\{
	\begin{tabular}{cccc}
		$z_{1}$ 		& $z_{2}$		& 	$z_{3}$	\\
		$\theta_{1}^{1}$	& $\theta_{2}^{1}$	& 	$\theta_{3}$ \\
		$\theta_{1}^{2}$	& $\theta_{2}^{2}$	& 	$\theta_{3}$ \\
		$\theta_{1}^{3}$	& $\theta_{2}^{3}$	& 	$0$ \\
		$0$	& $\theta_{2}^{4}$	& 	$0$ \\
	\end{tabular}
	\right\},\qquad 
	\theta_{1}^{1} + \theta_{1}^{2} + \theta_{1}^{3} + \theta_{2}^{1} + 
	\theta_{2}^{2} + \theta_{2}^{3} + \theta_{2}^{4}  + 2 \theta_{3}= 0.
\end{equation*}

We first consider an elementary Schlesinger transformation $\left\{\begin{smallmatrix} 1&2\\1&1\end{smallmatrix}\right\}$
for which $\bar{\theta}_{1}^{1} = \theta_{1}^{1}-1$ and $\bar{\theta}_{2}^{1} = \theta_{2}^{1}+1$. The multiplier matrix for 
this transformation is 
\begin{equation*}
	\mathbf{R}(z) = \mathbf{I} + \frac{ z_{1} - z_{2} }{ z - z_{1} }\mathbf{P},\qquad\text{where }
	\mathbf{P} = \frac{ \mathbf{b}_{2,1}\mathbf{c}_{1}^{1\dag} }{ \mathbf{c}_{1}^{1\dag} \mathbf{b}_{2,1} } \text{ and we put }
	\mathbf{Q} = \mathbf{I} - \mathbf{P}.
\end{equation*}
Since this transformation does not involve the point $z_{3}$ with multiple eigenvalues, the dynamic is 
again given by equations~(\ref{eq:bb-cb-generators}--\ref{eq:cb-am}) that now take the form
\begin{align}
	\bar{\mathbf{b}}_{1,1} &= \frac{ 1 }{ c_{1}^{1} }\mathbf{b}_{2,1}, \quad
	\bar{\mathbf{b}}_{1,j} = \frac{ 1 }{ c_{1}^{j} } \left(\mathbf{I} - \frac{ \mathbf{P} }{ \theta_{1}^{1} - \theta_{1}^{j}-1 }
	\left(\mathbf{A}_{2} + \frac{ z_{2} - z_{1} }{ z_{3} - z_{1} }\mathbf{A}_{1}\right)
	\right) \mathbf{b}_{1,j}\quad (j=2,3);\\	
	\bar{\mathbf{b}}_{2,1} &= \frac{ 1 }{ c_{2}^{1} }\Bigg(
	(\theta_{2}^{1} + 1)\mathbf{I} + \mathbf{Q}\left(
	\mathbf{I} + \frac{ \mathbf{b}_{2,2} \mathbf{c}_{2}^{2\dag}}{ \theta_{2}^{1} - \theta_{2}^{2} + 1 } + 
	\frac{ \mathbf{b}_{2,3} \mathbf{c}_{2}^{3\dag}}{ \theta_{2}^{1} - \theta_{2}^{3} + 1 } + 
	\frac{ \mathbf{b}_{2,4} \mathbf{c}_{2}^{4\dag}}{ \theta_{2}^{1} - \theta_{2}^{4} + 1 }
	\right)\times\notag\\
	&\hskip2.8in\left(\mathbf{A}_{1} + \frac{ z_{1}- z_{2} }{ z_{3} - z_{2} }\mathbf{A}_{3}\right)
	\Bigg) \frac{ \mathbf{b}_{2,1} }{ \mathbf{c}_{1}^{1\dag}\mathbf{b}_{2,1} }; \notag\\
	\bar{\mathbf{b}}_{2,j} &= \frac{ 1 }{ c_{2}^{j} }\mathbf{Q}\mathbf{b}_{2,j}\quad(j=2,3,4); \qquad
	\bar{\mathbf{b}}_{3,j} = \frac{ 1 }{ c_{3}^{j} }\mathbf{R}(z_{3})\mathbf{b}_{3,j}\quad(j=1,2);\notag\\
	\bar{\mathbf{c}}_{1}^{1\dag} &= c_{1}^{1} \frac{ \mathbf{c}_{1}^{1\dag} }{ \mathbf{c}_{1}^{1\dag} \mathbf{b}_{2,1} }
	\Bigg((\theta_{1}^{1} - 1)\mathbf{I} + \left(\mathbf{A}_{2} + \frac{ z_{2} - z_{1} }{ z_{3} - z_{1} }\mathbf{A}_{3}\right)\times \notag\\
	&\hskip0.7in \left(\mathbf{I} + \frac{ \mathbf{b}_{1,2} \mathbf{c}_{1}^{2\dag} }{ \theta_{1}^{1} - \theta_{1}^{2} - 1 } + 
	\frac{ \mathbf{b}_{1,3} \mathbf{c}_{1}^{3\dag} }{ \theta_{1}^{1} - \theta_{1}^{3} - 1 }\right)\mathbf{Q}
	\Bigg),\quad 
	\bar{\mathbf{c}}_{1}^{j\dag} = c_{1}^{j} \mathbf{c}_{1}^{j\dag} \mathbf{Q}\quad (j=2,3);\notag\\
	\bar{\mathbf{c}}_{2}^{1\dag} &= c_{2}^{1}\mathbf{c}_{1}^{1\dag},\quad
	\bar{\mathbf{c}}_{2}^{j\dag} = c_{2}^{j}\mathbf{c}_{2}^{j\dag}\left(\mathbf{I} - 
	\left(\mathbf{A}_{1} + \frac{ z_{1} - z_{2} }{ z_{3} - z_{2} }\mathbf{A}_{3}\right)\frac{ \mathbf{P} }{ \theta_{2}^{1} - \theta_{2}^{j} + 1 }
	\right)\quad (j=2,3,4); \notag\\
	\bar{\mathbf{c}}_{3}^{j\dag} &= c_{3}^{j} \mathbf{c}_{3}^{j\dag}\mathbf{R}^{-1}(z_{3}),\notag
\end{align}
where $c_{i}^{j}$ are again arbitrary non-zero constants.

Similarly to the previous example, we parameterize the matrices as 
\begin{alignat*}{2}
	\mathbf{B}_{1} &= \begin{bmatrix}
		1 & 0 & 0  \\ 0 & 1 & 0 \\ 0 & 0 & 1 
	\end{bmatrix},&\qquad 
	\mathbf{C}_{1}^{\dag} &= \begin{bmatrix}
		\theta_{1}^{1} & 0 & 0 & \alpha \\ 0 & \theta_{1}^{2} & 0 & \beta \\ 0 & 0 & \theta_{1}^{3} & \gamma
	\end{bmatrix},\\
	\mathbf{B}_{3} &= \begin{bmatrix}
		0 & 1 \\ 0 & 1 \\ 0 & 1 \\ 1 & 1
	\end{bmatrix},&\qquad 
	\mathbf{C}_{3}^{\dag} &= \begin{bmatrix}
		-(x + \theta_{3}) & 0 & x & \theta_{3} \\ 0 & \theta_{3} - y & y & 0
	\end{bmatrix},\\
	\mathbf{A}_{2} &= - (\mathbf{A}_{1} + \mathbf{A}_{3}).
\end{alignat*}
Using the condition that the eigenvalues of $\mathbf{A}_{2}$ are $\theta_{2}^{1},\dots,\theta_{2}^{4}$ we get a system of three linear equations on 
$\alpha$, $\beta$, and $\gamma$ with coefficients depending on $x$ and $y$, which gives us rational functions $\alpha(x,y)$, $\beta(x,y)$, and
$\gamma(x,y)$ (again, the resulting expressions are quite large and we omit them). Thus, the space of accessory parameters for Fuchsian systems of 
this type is two-dimensional and $x$ and $y$ are some coordinates on this space.

The resulting mapping $\psi: (x,y)\to (\bar{x},\bar{y})$ becomes
very complicated (and computing it requires a Computer Algebra System, in our work we have used \textbf{Mathematica}) and so we omit equations describing the map.
Nevertheless, it is possible to do a complete geometric analysis of the mapping. 

We find that the indeterminate points of $\psi$ are 
\begin{alignat*}{2}
	& p_{1}\left(\frac{ (\theta_{1}^{1} + \theta_{2}^{1} + \theta_{3})(\theta_{1}^{3} + \theta_{2}^{1}) }{ \theta_{1}^{1} - \theta_{1}^{3} },
	-\frac{ (\theta_{1}^{2} + \theta_{2}^{1} + \theta_{3})(\theta_{1}^{3} + \theta_{2}^{1}) }{ \theta_{1}^{2} - \theta_{1}^{3} }\right),&\quad 
	&p_{5}(0,0),\\
	&p_{2}\left(\frac{ (\theta_{1}^{1} + \theta_{2}^{2} + \theta_{3})(\theta_{1}^{3} + \theta_{2}^{2}) }{ \theta_{1}^{1} - \theta_{1}^{3} },
	-\frac{ (\theta_{1}^{2} + \theta_{2}^{2} + \theta_{3})(\theta_{1}^{3} + \theta_{2}^{2}) }{ \theta_{1}^{2} - \theta_{1}^{3} }\right),&\quad 
	&p_{6}(-\theta_{3},\theta_{3}),\\
	&p_{3}\left(\frac{ (\theta_{1}^{1} + \theta_{2}^{3} + \theta_{3})(\theta_{1}^{3} + \theta_{2}^{3}) }{ \theta_{1}^{1} - \theta_{1}^{3} },
	-\frac{ (\theta_{1}^{2} + \theta_{2}^{3} + \theta_{3})(\theta_{1}^{3} + \theta_{2}^{3}) }{ \theta_{1}^{2} - \theta_{1}^{3} }\right), &\\
	&p_{4}\left(\frac{ (\theta_{1}^{1} + \theta_{2}^{4} + \theta_{3})(\theta_{1}^{3} + \theta_{2}^{4}) }{ \theta_{1}^{1} - \theta_{1}^{3} },
	-\frac{ (\theta_{1}^{2} + \theta_{2}^{4} + \theta_{3})(\theta_{1}^{3} + \theta_{2}^{4}) }{ \theta_{1}^{2} - \theta_{1}^{3} }\right), 
\end{alignat*}
as well as the sequence of infinitely close points
\begin{equation*}
	p_{7}\left(\frac{ 1 }{ x } = 0,\frac{ 1 }{ y } = 0\right)\longleftarrow 
	p_{8}\left(\frac{ 1 }{ x } = 0,\frac{ x }{ y } = 
	- \frac{ (\theta_{1}^{1} + 1)(\theta_{1}^{2} - \theta_{1}^{3}) }{ (\theta_{1}^{2} + 1)(\theta_{1}^{1} - \theta_{1}^{3}) }\right).
\end{equation*}
Note also that the points $p_{1},\dots,p_{7}$  all lie on a $(2,2)$-curve $Q$ given by the equation
\begin{multline}
	\left((\theta_{1}^{3}-\theta_{1}^{1})x + (\theta_{1}^{3} - \theta_{1}^{2})y\right)^2 \\
	+(\theta_{1}^{1} - \theta_{1}^{2})\left((\theta_{1}^{3}- \theta_{1}^{2} - \theta_{3})(\theta_{1}^{3}-\theta_{1}^{1})x + 
	(\theta_{1}^{3}- \theta_{1}^{1} - \theta_{3})(\theta_{1}^{3} - \theta_{1}^{2})y\right)
	=0.
\end{multline}

Resolving indeterminate points of this map using blow-ups gives us the Okamoto surface $\mathcal{X}_{\theta}$ 
pictured on Figure~\ref{fig:dpa1-schlesinger-blowup}. Note that 
the $-2$ curves $D_{0} = 2H_{x} + 2H_{y} - F_{1} - F_{2} - F_{3} - F_{4} - F_{5} - F_{6} - 2F_{7}$
and $D_{1} = F_{7} - F_{8}$ touch at the point with coordinates 
$\left(\frac{ 1 }{ x } = 0, \frac{ x }{ y } = - \frac{\strut \theta_{1}^{2} - \theta_{1}^{3} }{\strut \theta_{1}^{1} - \theta_{1}^{3} }\right)$. 
Thus, we immediately see that this is indeed a surface of type $A_{1}^{(1)*}$.

\begin{figure}[h]
	\centering
		\includegraphics{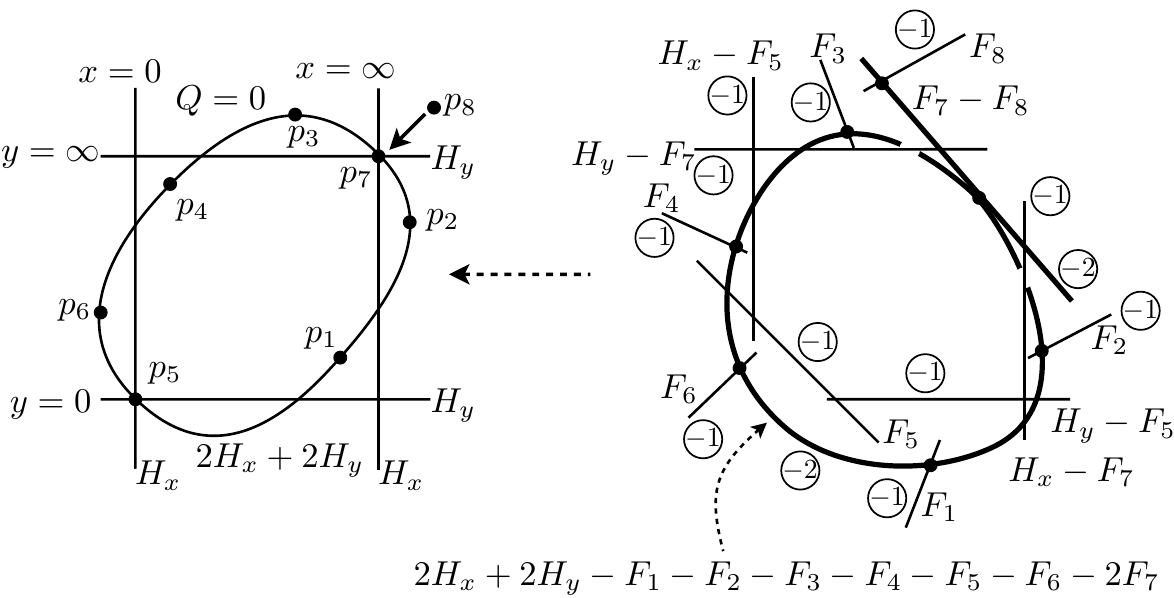}
	\caption{Okamoto surface $\mathcal{X}_{\mathbf{\theta}}$ for the Schlesinger transformations reduction to d-$P({A}_{1}^{(1)*})$.}
		\label{fig:dpa1-schlesinger-blowup}
\end{figure}

\subsubsection{Reduction to the standard form} 
\label{ssub:reduction_to_the_standard_formA1}

We now proceed to match the surface $\mathcal{X}_{\mathbf{\theta}}$ described by the blow-up diagram on Figure~\ref{fig:dpa1-schlesinger-blowup} 
with the surface $\mathcal{X}_{\mathbf{b}}$ 
described by diagram on Figure~\ref{fig:dpa1-standard-blowup}. As in the previous example,  
we look for rational classes $\mathcal{H}_{f}$, $\mathcal{H}_{g}$, $\mathcal{E}_{1},\dots, \mathcal{E}_{8}$ 
in $\operatorname{Pic}(\mathcal{X}_{\mathbf{\theta}})$ such that 
\begin{equation*}
	\mathcal{H}_{f}\bullet \mathcal{H}_{g} = 1, \  \mathcal{E}_{i}^{2} = -1, \  
	\mathcal{H}_{f}^{2} = \mathcal{H}_{g}^{2}  = \mathcal{H}_{f}\bullet \mathcal{E}_{i} = \mathcal{H}_{g}\bullet \mathcal{E}_{i}
	 = \mathcal{E}_{i}\bullet \mathcal{E}_{j} = 0,\  1\leq i\neq j\leq 8,
\end{equation*}
and the resulting configuration matches diagram on Figure~\ref{fig:dpa1-standard-blowup}. This time the (virtual) genus formula
$g(C) = (C^{2} + K_{\mathcal{X}}\bullet C)/2 + 1$ suggests we see that we should look for classes of rational curves of self-intersection zero among
$H_{x}$, $H_{y}$, or $H_{x} + H_{y} - F_{i} - F_{j}$ and for classes of rational curves of self-intersection $-1$ among $F_{i}$,
$H_{x} - F_{i}$, $H_{y} - F_{i}$, or $H_{x} + H_{y} - F_{i}-F_{j}-F_{k}$.

It is again convenient to start by comparing the $-2$-curves on both diagrams, 
\begin{align*}
	D_{0} &= 2H_{x} + 2 H_{y} - F_{1} - F_{2} - F_{3} - F_{4} - F_{5} - F_{6} - 2 F_{7} \\
	&= \mathcal{H}_{f} + \mathcal{H}_{g} - \mathcal{E}_{1} - \mathcal{E}_{2}
	- \mathcal{E}_{3} - \mathcal{E}_{4},\\
	D_{1} &= F_{7} - F_{8}= \mathcal{H}_{f} + \mathcal{H}_{g} - \mathcal{E}_{5} - \mathcal{E}_{6} - \mathcal{E}_{7} - \mathcal{E}_{8}.
\end{align*}
Given the uniformity of the coordinates of $p_{i}$, $i=1,\dots,4$, we see that 
it makes sense to choose $\mathcal{E}_{i} = F_{i}$ for $i=1,\dots,4$, and also we can put $\mathcal{E}_{8} = F_{8}$. 
This results in $\mathcal{H}_{f} + \mathcal{H}_{g} = 2 H_{x} + 2H_{y} - F_{5} - F_{6} - 2F_{7}$, which suggests taking
$\mathcal{H}_{f} = H_{x} + H_{y} - F_{5} - F_{7}$ and $\mathcal{H}_{y} = H_{x} + H_{y} - F_{6} - F_{7}$ (and so 
$\mathcal{H}_{f}\bullet \mathcal{H}_{g} = 1$ and $\mathcal{H}_{f}^{2} = \mathcal{H}_{g}^{2} = 0$). Then 
$\mathcal{E}_{5} + \mathcal{E}_{6} + \mathcal{E}_{7} = 2 H_{x} + 2 H_{y} - F_{5} - F_{6} - 3 E_{7}$, so we take
$\mathcal{E}_{5} = H_{y} - F_{7}$, $\mathcal{E}_{6} = H_{x} - F_{7}$, and $\mathcal{E}_{7} = H_{x} + H_{y} - F_{5} - F_{6} - F_{7}$.
It is not very hard to show that such choice satisfies all of our requirements and moreover, 
it is essentially unique (up to a permutation of the indices of exceptional divisors).
To summarize, we get the following identification:
\begin{alignat*}{4}
	\mathcal{H}_{f}&= H_{x} + H_{y} -F_{5} - F_{7}, & \quad  
	\mathcal{E}_{1} &= F_{1}, & \quad  \mathcal{E}_{3} &= F_{3}, & \quad  
	\mathcal{E}_{5} &= H_{y} - F_{7}, \\
	\mathcal{H}_{g}&= H_{x} + H_{y} -F_{6} - F_{7}, & \quad \mathcal{E}_{2} &= F_{2},  & \quad  
	\mathcal{E}_{4} &= F_{4}, & \quad  
	\mathcal{E}_{6} &= H_{x} - F_{7},\\ 
	\mathcal{E}_{7} &= H_{x} + H_{y} - F_{5} - F_{6} - F_{7}, & \quad  \mathcal{E}_{8} &= F_{8}.
\end{alignat*}

Let us now define the base coordinates $f$ and $g$ 
of the linear systems $|\mathcal{H}_{f}|$ and $|\mathcal{H}_{g}|$ that will map the exceptional fibers of the divisors
$\mathcal{E}_{i}$ to the points $\pi_{i}$ such that $\pi_{1},\dots,\pi_{4}$ are on a line $f + g = \text{const}$ and 
$\pi_{5},\dots,\pi_{8}$ are on the line $f+g=0$. 
Since the 
pencil $|\mathcal{H}_{f}|$ consists of $(1,1)$ curves on $\mathbb{P}^{1}\times \mathbb{P}^{1}$ passing through $p_{5}(0,0)$ and
$p_{7}(\infty,\infty)$,
\begin{align*}
	|\mathcal{H}_{f}| = |H_{x} + H_{y} - F_{5} - F_{7}| &= \{a xy + b x  + c y + d = 0 \mid a = d =0\} \\ &= \{b x + c y = 0\},
\intertext{and so we can initially define the base coordinate as $f_{1} = y/x$. Similarly, the 
pencil $|\mathcal{H}_{g}|$ consists of $(1,1)$ curves on $\mathbb{P}^{1}\times \mathbb{P}^{1}$ passing through $p_{6}(-\theta_{3},\theta_{3})$ and
$p_{7}(\infty,\infty)$, so}
	|\mathcal{H}_{g}| = |H_{x} + H_{y }- F_{6} - F_{7}| &= \{a xy + b x  + c y + d = 0 \mid a = 0, (c-b)\theta_{3} + d =0\} \\
	&= \{b(x + \theta_{3}) + c(y-\theta_{3}) = 0\},
\end{align*}
and $g_{1} = (y-\theta_{3})/(x + \theta_{3})$. Next we will do a series of affine change of variables to arrange that the points 
$p_{i}$, $i=1,\dots,4$, are on a line $f + g = \text{const}$. We have (below $i=1,\dots,4$)
\begin{alignat*}{2}
	f_{1}(\pi_{i}) &=- \frac{ (\theta_{1}^{1} - \theta_{1}^{3})(\theta_{1}^{2} + \theta_{2}^{i} + \theta_{3}) }{ 
	(\theta_{1}^{2} - \theta_{1}^{3})(\theta_{1}^{1} + \theta_{2}^{i} + \theta_{3})  },\qquad  &
	g_{1}(\pi_{i}) &=- \frac{ (\theta_{1}^{1} - \theta_{1}^{3})(\theta_{1}^{2} + \theta_{2}^{i}) }{ 
	(\theta_{1}^{2} - \theta_{1}^{3})(\theta_{1}^{1} + \theta_{2}^{i}) },\\
	\intertext{and therefore, it makes sense to put}
	f_{2} &= \frac{ (\theta_{1}^{2} - \theta_{1}^{3})}{(\theta_{1}^{1} - \theta_{1}^{3})}f_{1} + 1,\qquad & 
	g_{2} &= \frac{ (\theta_{1}^{2} - \theta_{1}^{3})}{(\theta_{1}^{1} - \theta_{1}^{3})}g_{1} + 1.\\
	\intertext{We then get}
	f_{2}(\pi_{i}) &=- \frac{ (\theta_{1}^{1} - \theta_{1}^{2})}{(\theta_{1}^{1} + \theta_{2}^{i} + \theta_{3})  },\qquad  &
	g_{2}(\pi_{i}) &=- \frac{ (\theta_{1}^{1} - \theta_{1}^{2})}{ (\theta_{1}^{1} + \theta_{2}^{i}) },\\	
	\intertext{and so we put}
	f_{3} &= \frac{ (\theta_{1}^{1} - \theta_{1}^{2})}{f_{2}},\qquad & 
	g_{3} &= \frac{ (\theta_{1}^{1} - \theta_{1}^{2})}{g_{2}}\\
	\intertext{to get}
	f_{3}(\pi_{i}) &= \theta_{1}^{1} + \theta_{2}^{i} + \theta_{3} ,\qquad  &
	g_{3}(\pi_{i}) &= \theta_{1}^{1} + \theta_{2}^{i}.
\end{alignat*}
Thus, our final change of coordinates is
\begin{align*}
	f&= f_{3} - \theta_{1}^{1} = - \frac{ \theta_{1}^{2}(\theta_{1}^{3} - \theta_{1}^{1})x - \theta_{1}^{1}(\theta_{1}^{2} - \theta_{1}^{3})y }{ 
	(\theta_{1}^{3} - \theta_{1}^{1})x - (\theta_{1}^{2} - \theta_{1}^{3})y  },\\
	g&= \theta_{1}^{1} - g_{3} = \frac{ \theta_{1}^{2}(\theta_{1}^{3} - \theta_{1}^{1})(x+\theta_{3}) - 
	\theta_{1}^{1}(\theta_{1}^{2} - \theta_{1}^{3})(y-\theta_{3}) }{ 
	(\theta_{1}^{3} - \theta_{1}^{1})(x+\theta_{3}) - (\theta_{1}^{2} - \theta_{1}^{3})(y-\theta_{3})  },
\end{align*}
and for $i=1,\dots,4$, $f(\pi_{i}) = f(p_{i}) = \theta_{2}^{i} + \theta_{3}$, $g(\pi_{i}) = g(p_{i}) = - \theta_{2}^{i}$ and so these 
points lie on the line $f + g = \theta_{3}$. We also get the identification of some of the parameters, $\theta_{2}^{i} = b_{i} - b$ for $i=1,\dots,4$, and 
$\theta_{3} = 2b$.
It remains to verify that this change of variables puts points $\pi_{5},\dots,\pi_{8}$ on the line $f + g=0$ and identify the remaining parameters.
The exceptional divisor $\mathcal{E}_{5}$ in the $(x,y)$-coordinates corresponds to the line $y=\infty$, and so $(f,g)(\pi_{5}) = (-\theta_{1}^{1},\theta_{1}^{1})$.
Similarly, $(f,g)(\pi_{6}) = (-\theta_{1}^{2},\theta_{1}^{2})$. The exceptional divisor $\mathcal{E}_{7}$ corresponds to the line $x + y = 0$, and 
so $(f,g)(\pi_{7}) = (-\theta_{1}^{3},\theta_{1}^{3})$. Finally, $\mathcal{E}_{8}$ is given by 
$ x = y = \infty, \frac{ x }{ y } = - \frac{\strut (\theta_{1}^{1} + 1)(\theta_{1}^{2} - \theta_{1}^{3}) }{\strut (\theta_{1}^{2} + 1)(\theta_{1}^{1} - \theta_{1}^{3}) }$,
and so  $(f,g)(\pi_{8}) = (1,-1)$. Thus, we see that indeed $\pi_{5},\dots,\pi_{8}$ lie on the line $f + g = 0$ and the remaining identification between
the parameters is $\theta_{1}^{1} = b_{5}$, $\theta_{1}^{2} = b_{6}$, $\theta_{1}^{3} = b_{7}$, and $b_{8} = -1$.

We are now in the position to compare the dynamic given by an elementary Schlesinger transformation with the dynamic of our model example
of d-$P(A_{1}^{(1)*})$. As in the previous example, there are two different ways to do so. First, we can compute the corresponding translation vector.
It is not very difficult to show that the action of $\psi_{*}$ of an elementary Schlesinger transformation 
$\left\{\begin{smallmatrix} 1&2\\1&1\end{smallmatrix}\right\}$ 
on the classes $\mathcal{H}_{f}$, $\mathcal{H}_{g}$, and $\mathcal{E}_{i}$ is
\begin{align*}
	 \mathcal{H}_{f}&\mapsto 4 \mathcal{H}_{f} + 3 \mathcal{H}_{g} - 3\mathcal{E}_{1} - \mathcal{E}_{2} - \mathcal{E}_{3} - 
	\mathcal{E}_{4} - 2 \mathcal{E}_{6} - 2\mathcal{E}_{7} -  2\mathcal{E}_{8},
	\\
	 \mathcal{H}_{g}&\mapsto 3 \mathcal{H}_{f} + 4 \mathcal{H}_{g} - 3\mathcal{E}_{1} - \mathcal{E}_{2} - \mathcal{E}_{3} - 
	\mathcal{E}_{4} - 2 \mathcal{E}_{6} - 2 \mathcal{E}_{7} - 2\mathcal{E}_{8},\\
	 \mathcal{E}_{1}&\mapsto \mathcal{E}_{5},
	\\
	 \mathcal{E}_{2}&\mapsto 2\mathcal{H}_{f} + 2\mathcal{H}_{g} - 2\mathcal{E}_{1} - \mathcal{E}_{3} - \mathcal{E}_{4} - 
	\mathcal{E}_{6} - \mathcal{E}_{7} - \mathcal{E}_{8},
	\\
	 \mathcal{E}_{3}&\mapsto 2\mathcal{H}_{f} + 2\mathcal{H}_{g} - 2\mathcal{E}_{1} - \mathcal{E}_{2} - \mathcal{E}_{4} - 
	\mathcal{E}_{6} - \mathcal{E}_{7} - \mathcal{E}_{8},
	\\
	 \mathcal{E}_{4}&\mapsto 2\mathcal{H}_{f} + 2\mathcal{H}_{g} - 2\mathcal{E}_{1} - \mathcal{E}_{2} - \mathcal{E}_{3} - 
	\mathcal{E}_{6} - \mathcal{E}_{7} - \mathcal{E}_{8},
	\\
	 \mathcal{E}_{5}&\mapsto 3\mathcal{H}_{f} + 3\mathcal{H}_{g} - 2\mathcal{E}_{1} - \mathcal{E}_{2} - \mathcal{E}_{3} - \mathcal{E}_{4} - 
	2\mathcal{E}_{6} - 2\mathcal{E}_{7} - 2\mathcal{E}_{8},
	\\
	 \mathcal{E}_{6}&\mapsto \mathcal{H}_{f} + \mathcal{H}_{g} - \mathcal{E}_{1} - \mathcal{E}_{7} - \mathcal{E}_{8},
	\\
	 \mathcal{E}_{7}&\mapsto \mathcal{H}_{f} + \mathcal{H}_{g} - \mathcal{E}_{1} - \mathcal{E}_{6} - \mathcal{E}_{8},
	\\
	 \mathcal{E}_{8}&\mapsto \mathcal{H}_{f} + \mathcal{H}_{g} - \mathcal{E}_{1} - \mathcal{E}_{6} - \mathcal{E}_{7}.
\end{align*}
Comparing the action of $\psi_{*}$ with the action of the standard dynamic $\varphi_{*}$ given by~(\ref{eq:dpa1-trans-std}) 
on the symmetry sub-lattice, we see that the translation vectors are different:
\begin{align*}
	\psi_{*}: (\alpha_{0}, \alpha_{1}, \alpha_{2}, \alpha_{3}, \alpha_{4}, \alpha_{5}, \alpha_{6},\alpha_{7})&\mapsto	
	(\alpha_{0}, \alpha_{1}, \alpha_{2}, \alpha_{3}, \alpha_{4}, \alpha_{5}, \alpha_{6},\alpha_{7}) + \\
	&\qquad
	(0,0,0,-1,0,1,0,0)(-K_{\mathcal{X}}),\\
	\varphi_{*}: (\alpha_{0}, \alpha_{1}, \alpha_{2}, \alpha_{3}, \alpha_{4}, \alpha_{5}, \alpha_{6},\alpha_{7})&\mapsto	
	(\alpha_{0}, \alpha_{1}, \alpha_{2}, \alpha_{3}, \alpha_{4}, \alpha_{5}, \alpha_{6},\alpha_{7}) + \\
	&\qquad
	(0,0,0,0,1,0,0,-2)(-K_{\mathcal{X}}). 
\end{align*}
To get more insight into the relationship between Schlesinger transformations and the standard d-$P(A_{1}^{(1)*})$ dynamic,
it is better to compute the action of d-$P(A_{1}^{(1)*})$ on the Riemann scheme of our Fuchsian system using the 
above identification of parameters:
\begin{align*}
	\left\{
	\begin{tabular}{cccc}
		$z_{1}$ 		& $z_{2}$		& 	$z_{3}$	\\
		$\theta_{1}^{1}$	& $\theta_{2}^{1}$	& 	$\theta_{3}$ \\
		$\theta_{1}^{2}$	& $\theta_{2}^{2}$	& 	$\theta_{3}$ \\
		$\theta_{1}^{3}$	& $\theta_{2}^{3}$	& 	$0$ \\
		$0$	& $\theta_{2}^{4}$	& 	$0$ \\
	\end{tabular}
	\right\}&\overset{\left\{\begin{smallmatrix} 1&2\\1&1\end{smallmatrix}\right\}}{\strut\longmapsto}
	\left\{
	\begin{tabular}{cccc}
		$z_{1}$ 		& $z_{2}$		& 	$z_{3}$	\\
		$\theta_{1}^{1}-1$	& $\theta_{2}^{1}+1$	& 	$\theta_{3}$ \\
		$\theta_{1}^{2}$	& $\theta_{2}^{2}$	& 	$\theta_{3}$ \\
		$\theta_{1}^{3}$	& $\theta_{2}^{3}$	& 	$0$ \\
		$0$	& $\theta_{2}^{4}$	& 	$0$ \\
	\end{tabular}
	\right\},	\\
	\left\{
	\begin{tabular}{cccc}
		$z_{1}$ 		& $z_{2}$		& 	$z_{3}$	\\
		$\theta_{1}^{1}$	& $\theta_{2}^{1}$	& 	$\theta_{3}$ \\
		$\theta_{1}^{2}$	& $\theta_{2}^{2}$	& 	$\theta_{3}$ \\
		$\theta_{1}^{3}$	& $\theta_{2}^{3}$	& 	$0$ \\
		$0$	& $\theta_{2}^{4}$	& 	$0$ \\
	\end{tabular}
	\right\}&\overset{\text{d-$P(A_{1}^{(1)*})$}}{\strut\longmapsto}
	\left\{
	\begin{tabular}{cccc}
		$z_{1}$ 		& $z_{2}$		& 	$z_{3}$	\\
		$\theta_{1}^{1}$	& $\theta_{2}^{1}-1$	& 	$\theta_{3}+2$ \\
		$\theta_{1}^{2}$	& $\theta_{2}^{2}-1$	& 	$\theta_{3}+2$ \\
		$\theta_{1}^{3}$	& $\theta_{2}^{3}-1$	& 	$0$ \\
		$0$					& $\theta_{2}^{4}-1$	& 	$0$ \\
	\end{tabular}
	\right\}.
\end{align*}
Thus, we see that the standard d-$P(A_{1}^{(1)*})$ dynamic changes the multiple eigenvalue and so 
it requires the use of rank-two elementary Schlesinger transformations. In fact, the action on the
Riemann scheme suggests that 
\begin{equation*}
	\text{d-$P(A_{1}^{(1)*})$} = \left\{\begin{smallmatrix} 2&3\\3&1\\4&2\end{smallmatrix}\right\} \circ 
	\left\{\begin{smallmatrix} 2&3\\1&1\\2&2\end{smallmatrix}\right\}.
\end{equation*}

Thus, consider the transformation $\left\{\begin{smallmatrix} 2&3\\1&1\\2&2\end{smallmatrix}\right\}$ changing the characteristic 
indices by $\bar{\theta}_{2}^{1} = \theta_{2}^{1} - 1$, $\bar{\theta}_{2}^{1} = \theta_{2}^{1} - 1$, and $\bar{\theta}_{3} = \theta_{3} - 1$.
This transformation is given by equations~(\ref{eq:bb-cb-generators-rank2}--\ref{eq:cb-am-rank2}) that now take the form
\begin{align}
	\bar{\mathbf{b}}_{1,j} &= \frac{ 1 }{ c_{1}^{j} }\mathbf{R}(z_{1})\mathbf{b}_{1,j}\quad(j=1,2,3);\qquad
	\bar{\mathbf{c}}_{1}^{j\dag} = c_{1}^{j} \mathbf{c}_{1}^{j\dag}\mathbf{R}^{-1}(z_{i})\quad(j=1,2,3);\\
	\bar{\mathbf{b}}_{2,1}& = \frac{ 1 }{ c_{2}^{1} }\mathbf{Q}_{2}\mathbf{b}_{3,1},\quad
	\bar{\mathbf{b}}_{2,2} = \frac{ 1 }{ c_{2}^{2} }\mathbf{Q}_{1}\mathbf{b}_{3,2},\notag\\
	\bar{\mathbf{b}}_{2,j}& = \frac{ 1 }{ c_{2}^{j} }\left(
	\mathbf{I} - \left(\frac{ \mathcal{P}_{1} }{ \theta_{2}^{1} - \theta_{2}^{j} - 1} + 
	\frac{ \mathcal{P}_{2} }{ \theta_{2}^{2} - \theta_{2}^{j} - 1}\right)\left(\frac{ z_{3} - z_{2} }{ z_{1} - z_{2} }\mathbf{A}_{1} + \mathbf{A}_{3}\right)
	\right)	\mathbf{b}_{2,j}\quad(j=3,4);\notag\\
	\bar{\mathbf{c}}_{2}^{1\dag} &= c_{2}^{1} \frac{ \mathbf{c}_{2}^{1\dag} }{ \mathbf{c}_{2}^{1\dag}\mathbf{Q}_{2}\mathbf{b}_{3,1} }
	\Bigg(
		(\theta_{2}^{1} - 1)\mathbf{I} + \left(\frac{ z_{3} - z_{2} }{ z_{1} - z_{2} }\mathbf{A}_{1} + \mathbf{A}_{3}\right)\times \notag\\
		&\hskip2.5in 
		\left(\mathbf{I} + \frac{ \mathbf{b}_{2,3}\mathbf{c}_{2}^{3\dag} }{ \theta_{2}^{1} - \theta_{2}^{3} - 1 } + 
		\frac{ \mathbf{b}_{2,4}\mathbf{c}_{2}^{4\dag} }{ \theta_{2}^{1} - \theta_{2}^{4} - 1 }\right)\mathcal{Q}
	\Bigg),\notag\\
	\bar{\mathbf{c}}_{2}^{2\dag} &= c_{2}^{2} \frac{ \mathbf{c}_{2}^{2\dag} }{ \mathbf{c}_{2}^{2\dag}\mathbf{Q}_{1}\mathbf{b}_{3,2} }
	\Bigg(
		(\theta_{2}^{2} - 1)\mathbf{I} + \left(\frac{ z_{3} - z_{2} }{ z_{1} - z_{2} }\mathbf{A}_{1} + \mathbf{A}_{3}\right)\times \notag\\
		&\hskip2.5in 
		\left(\mathbf{I} + \frac{ \mathbf{b}_{2,3}\mathbf{c}_{2}^{3\dag} }{ \theta_{2}^{2} - \theta_{2}^{3} - 1 } + 
		\frac{ \mathbf{b}_{2,4}\mathbf{c}_{2}^{4\dag} }{ \theta_{2}^{2} - \theta_{2}^{4} - 1 }\right)\mathcal{Q}
	\Bigg),\notag\\	
	\bar{\mathbf{c}}_{2}^{j\dag} &= c_{2}^{j}\mathbf{c}_{2}^{j\dag}\mathcal{Q}\quad(j=2,3);\notag\\
	\bar{\mathbf{b}}_{3,1} &= \frac{ 1 }{ c_{3}^{1} }\left((\theta_{3} + 1)\mathbf{I} + \mathcal{Q}
	\left(\frac{ z_{2} - z_{3} }{ z_{1} - z_{3} }\mathbf{A}_{1} + \mathbf{A}_{2}\right)
	\right) \frac{ \mathbf{b}_{3,1} }{ \mathbf{c}_{2}^{1\dag} \mathbf{Q}_{2} \mathbf{b}_{3,1} },\notag\\
	\bar{\mathbf{b}}_{3,2} &= \frac{ 1 }{ c_{3}^{2} }\left((\theta_{3} + 1)\mathbf{I} + \mathcal{Q}
	\left(\frac{ z_{2} - z_{3} }{ z_{1} - z_{3} }\mathbf{A}_{1} + \mathbf{A}_{2}\right)
	\right) \frac{ \mathbf{b}_{3,2} }{ \mathbf{c}_{2}^{2\dag} \mathbf{Q}_{1} \mathbf{b}_{3,2} };\notag\\
	\bar{\mathbf{c}}_{3}^{1\dag} &= c_{3}^{1}\mathbf{c}_{2}^{1\dag}\mathbf{Q}_{2},\quad
	\bar{\mathbf{c}}_{3}^{2\dag} = c_{3}^{2}\mathbf{c}_{2}^{2\dag}\mathbf{Q}_{1},
	\notag
\end{align}
where 
\begin{align*}
	\mathbf{P}_{1} &= \frac{ \mathbf{b}_{3,1}\mathbf{c}_{2}^{1\dag} }{ \mathbf{c}_{2}^{1\dag} \mathbf{b}_{3,1}},\quad \mathbf{Q}_{1} = \mathbf{I} - \mathbf{P}_{1},
	\qquad 
	\mathbf{P}_{2} = \frac{ \mathbf{b}_{3,2}\mathbf{c}_{2}^{2\dag} }{ \mathbf{c}_{2}^{2\dag} \mathbf{b}_{3,2}},\quad \mathbf{Q}_{2} = \mathbf{I} - \mathbf{P}_{2};\\
	\mathcal{P}_{1} &= \frac{ \mathbf{Q}_{2}\mathbf{P}_{1} }{ \operatorname{Tr}(\mathbf{Q}_{2}\mathbf{P}_{1}) },\quad
	\mathcal{P}_{2} = \frac{ \mathbf{Q}_{1}\mathbf{P}_{2} }{ \operatorname{Tr}(\mathbf{Q}_{1}\mathbf{P}_{2}) },\qquad
	\mathcal{P} = \mathcal{P}_{1} + \mathcal{P}_{2},\quad \mathcal{Q} = \mathbf{I} - \mathcal{P};\\
	\mathbf{R}(z) &= \mathbf{I} + \frac{ z_{2} - z_{3} }{ z - z_{2} }\mathcal{P}.
\end{align*}

The Okamoto surface for this dynamic is the same as before (it depends only on the Fuchsian system rather than a particular transformation),
and its action $\psi^{12}_{*}$ on $\operatorname{Pic}(\mathcal{X}_{\mathbf{\theta}})$ is given by
\begin{align*}
	 \mathcal{H}_{f}&\mapsto 6 \mathcal{H}_{f} + 3 \mathcal{H}_{g} - \mathcal{E}_{1} - \mathcal{E}_{2} - 3\mathcal{E}_{3} - 
	3\mathcal{E}_{4} - 2 \mathcal{E}_{5} - 2 \mathcal{E}_{6} - 2\mathcal{E}_{7} -  2\mathcal{E}_{8},
	\\
	 \mathcal{H}_{g}&\mapsto 3 \mathcal{H}_{f} + 2 \mathcal{H}_{g} - 2\mathcal{E}_{3} - 2\mathcal{E}_{4} - \mathcal{E}_{5} - 
	\mathcal{E}_{6} - \mathcal{E}_{7}  -2\mathcal{E}_{8},\\
	 \mathcal{E}_{1}&\mapsto 3\mathcal{H}_{f} + 2\mathcal{H}_{g} - \mathcal{E}_{2} - 2\mathcal{E}_{3} - 2\mathcal{E}_{4} - 
	\mathcal{E}_{5} - \mathcal{E}_{6} - \mathcal{E}_{7} - \mathcal{E}_{8},
	\\
	 \mathcal{E}_{2}&\mapsto 3\mathcal{H}_{f} + 2\mathcal{H}_{g} - \mathcal{E}_{1} - 2\mathcal{E}_{3} - 2\mathcal{E}_{4} - 
	\mathcal{E}_{5} - \mathcal{E}_{6} - \mathcal{E}_{7} - \mathcal{E}_{8},
	\\
	 \mathcal{E}_{3}&\mapsto \mathcal{H}_{f} - \mathcal{E}_{4},
	\\
	 \mathcal{E}_{4}&\mapsto \mathcal{H}_{f} - \mathcal{E}_{3},
	\\
	 \mathcal{E}_{5}&\mapsto 2\mathcal{H}_{f} + \mathcal{H}_{g} - \mathcal{E}_{3} - \mathcal{E}_{4} -
	\mathcal{E}_{6} - \mathcal{E}_{7} - \mathcal{E}_{8},
	\\
	 \mathcal{E}_{6}&\mapsto 2\mathcal{H}_{f} + \mathcal{H}_{g} - \mathcal{E}_{3} - \mathcal{E}_{4} -
	\mathcal{E}_{5} - \mathcal{E}_{7} - \mathcal{E}_{8},
	\\
	 \mathcal{E}_{7}&\mapsto 2\mathcal{H}_{f} + \mathcal{H}_{g} - \mathcal{E}_{3} - \mathcal{E}_{4} -
	\mathcal{E}_{5} - \mathcal{E}_{6} - \mathcal{E}_{8},
	\\
	 \mathcal{E}_{8}&\mapsto 2\mathcal{H}_{f} + \mathcal{H}_{g} - \mathcal{E}_{3} - \mathcal{E}_{4} -
	\mathcal{E}_{5} - \mathcal{E}_{6} - \mathcal{E}_{7},
\end{align*}
and the action on the symmetry sub-lattice is
\begin{align*}
	\psi^{12}_{*}: (\alpha_{0}, \alpha_{1}, \alpha_{2}, \alpha_{3}, \alpha_{4}, \alpha_{5}, \alpha_{6},\alpha_{7})&\mapsto	
	(\alpha_{0}, \alpha_{1}, \alpha_{2}, \alpha_{3}, \alpha_{4}, \alpha_{5}, \alpha_{6},\alpha_{7}) + \\
	&\qquad
	(0,0,1,0,0,0,0,-1)(-K_{\mathcal{X}}).
\end{align*}

Similarly, the action $\psi^{34}_{*}$ of the rank-two transformation 
$\left\{\begin{smallmatrix} 2&3\\3&1\\4&2\end{smallmatrix}\right\}$
on $\operatorname{Pic}(\mathcal{X}_{\mathbf{\theta}})$ is given by
\begin{align*}
	 \mathcal{H}_{f}&\mapsto 6 \mathcal{H}_{f} + 3 \mathcal{H}_{g} - 3\mathcal{E}_{1} - 3\mathcal{E}_{2} - \mathcal{E}_{3} - 
	\mathcal{E}_{4} - 2 \mathcal{E}_{5} - 2 \mathcal{E}_{6} - 2\mathcal{E}_{7} -  2\mathcal{E}_{8},
	\\
	 \mathcal{H}_{g}&\mapsto 3 \mathcal{H}_{f} + 2 \mathcal{H}_{g} - 2\mathcal{E}_{1} - 2\mathcal{E}_{2} - \mathcal{E}_{5} - 
	\mathcal{E}_{6} - \mathcal{E}_{7}  -2\mathcal{E}_{8},\\
	 \mathcal{E}_{1}&\mapsto \mathcal{H}_{f} - \mathcal{E}_{2},
	\\
	 \mathcal{E}_{2}&\mapsto \mathcal{H}_{f} - \mathcal{E}_{1},
	\\
	 \mathcal{E}_{3}&\mapsto 3\mathcal{H}_{f} + 2\mathcal{H}_{g} - 2\mathcal{E}_{1} - 2\mathcal{E}_{2} - \mathcal{E}_{4} - 
	\mathcal{E}_{5} - \mathcal{E}_{6} - \mathcal{E}_{7} - \mathcal{E}_{8},
	\\
	 \mathcal{E}_{4}&\mapsto 3\mathcal{H}_{f} + 2\mathcal{H}_{g} - 2\mathcal{E}_{1} - 2\mathcal{E}_{3} - \mathcal{E}_{3} - 
	\mathcal{E}_{5} - \mathcal{E}_{6} - \mathcal{E}_{7} - \mathcal{E}_{8},
	\\
	 \mathcal{E}_{5}&\mapsto 2\mathcal{H}_{f} + \mathcal{H}_{g} - \mathcal{E}_{1} - \mathcal{E}_{2} -
	\mathcal{E}_{6} - \mathcal{E}_{7} - \mathcal{E}_{8},
	\\
	 \mathcal{E}_{6}&\mapsto 2\mathcal{H}_{f} + \mathcal{H}_{g} - \mathcal{E}_{1} - \mathcal{E}_{2} -
	\mathcal{E}_{5} - \mathcal{E}_{7} - \mathcal{E}_{8},
	\\
	 \mathcal{E}_{7}&\mapsto 2\mathcal{H}_{f} + \mathcal{H}_{g} - \mathcal{E}_{1} - \mathcal{E}_{2} -
	\mathcal{E}_{5} - \mathcal{E}_{6} - \mathcal{E}_{8},
	\\
	 \mathcal{E}_{8}&\mapsto 2\mathcal{H}_{f} + \mathcal{H}_{g} - \mathcal{E}_{1} - \mathcal{E}_{2} -
	\mathcal{E}_{5} - \mathcal{E}_{6} - \mathcal{E}_{7},
\end{align*}
and its action on the symmetry sub-lattice is
\begin{align*}
	\psi^{34}_{*}: (\alpha_{0}, \alpha_{1}, \alpha_{2}, \alpha_{3}, \alpha_{4}, \alpha_{5}, \alpha_{6},\alpha_{7})&\mapsto	
	(\alpha_{0}, \alpha_{1}, \alpha_{2}, \alpha_{3}, \alpha_{4}, \alpha_{5}, \alpha_{6},\alpha_{7}) + \\
	&\qquad
	(0,0,-1,0,1,0,0,-1)(-K_{\mathcal{X}}).
\end{align*}

Thus,
\begin{align*}
	\psi^{34}_{*}\circ\psi^{12}_{*}: (\alpha_{0}, \alpha_{1}, \alpha_{2}, \alpha_{3}, \alpha_{4}, \alpha_{5}, \alpha_{6},\alpha_{7})&\mapsto	
	(\alpha_{0}, \alpha_{1}, \alpha_{2}, \alpha_{3}, \alpha_{4}, \alpha_{5}, \alpha_{6},\alpha_{7}) + \\
	&\qquad (0,0,1,0,0,0,0,-1)(-K_{\mathcal{X}}) + \\
	&\qquad (0,0,-1,0,1,0,0,-1)(-K_{\mathcal{X}})\\
	&=(\alpha_{0}, \alpha_{1}, \alpha_{2}, \alpha_{3}, \alpha_{4}, \alpha_{5}, \alpha_{6},\alpha_{7}) + \\
	&\qquad(0,0,0,0,1,0,0,-2)(-K_{\mathcal{X}}) = 	\varphi_{*}.
\end{align*}

Finally, using a Computer Algebra System we can verify by a direct calculation that 
	\begin{equation*}
		\text{d-$P(A_{1}^{(1)*})$} = \left\{\begin{smallmatrix} 2&3\\3&1\\4&2\end{smallmatrix}\right\} \circ 
		\left\{\begin{smallmatrix} 2&3\\1&1\\2&2\end{smallmatrix}\right\}
	\end{equation*}
holds on the level of equations as well.

\section{Conclusion} 
\label{sec:conclusion}
In this work we further develop a theory of discrete Schlesinger evolution equations that correspond to 
elementary Schlesinger transformations of ranks one and two of Fuchsian systems. We showed how to obtain 
difference Painlev\'e equations of types d-$P\left({A}_{2}^{(1)*}\right)$ and 
d-$P\left({A}_{1}^{(1)*}\right)$ as reductions of elementary Schlesinger transformations. We also 
tried to make our computations very detailed in order to illustrate general techniques on how to study 
discrete Painlev\'e equations geometrically.

One interesting observation is that standard examples of difference Painlev\'e equations of these types in 
both cases can be represented as compositions of elementary Schlesinger transformations. Thus, Schlesinger
dynamic should in principle be simpler and it would be interesting to find a nice and simple form of equations
giving this dynamic. 

\bibliographystyle{amsalpha}

\providecommand{\bysame}{\leavevmode\hbox to3em{\hrulefill}\thinspace}
\providecommand{\MR}{\relax\ifhmode\unskip\space\fi MR }
\providecommand{\MRhref}[2]{%
  \href{http://www.ams.org/mathscinet-getitem?mr=#1}{#2}
}
\providecommand{\href}[2]{#2}

\end{document}